\documentclass[11pt,a4paper,oneside,english]{article}
\usepackage[T1]{fontenc}
\usepackage[utf8]{inputenc}
\usepackage[english]{babel}
\usepackage{latexsym}
\usepackage{midpage}
\usepackage{newlfont}
\usepackage{color}
\usepackage{array}
\usepackage{amsmath} 
\usepackage{amsfonts} 
\usepackage{amsthm} 
\usepackage{amssymb}
\usepackage{setspace}
\usepackage[hidelinks]{hyperref} 
\usepackage{url}
\usepackage[normalem]{ulem}
\usepackage{cleveref}
\usepackage{graphicx} 
\usepackage{caption}
\usepackage{subcaption} 
\usepackage{mathdots}
\usepackage{longtable}
\usepackage{booktabs}
\usepackage[margin=1in,left=1.04in,right=1.04in,includefoot]{geometry}
\usepackage{rotating}
\usepackage{listings}
\usepackage{arydshln}
\usepackage{tikz}
\usepackage{color}
\allowdisplaybreaks

\definecolor{darkgreen}{rgb}{0, .5, 0}

\theoremstyle{plain} 
\newtheorem{thrm}{Theorem}[section] 
\newtheorem{cor}[thrm]{Corollary} 
\newtheorem{lem}[thrm]{Lemma} 
\newtheorem{prop}[thrm]{Proposition} 
\theoremstyle{definition} 
 
\theoremstyle{remark}

\numberwithin{equation}{section}

\providecommand{\abs}[1]{\ensuremath{\left\lvert#1\right\rvert}}
\newcommand{\E}{\mathbb{E}} 
\newcommand{\N}{\mathbb{N}} 
\newcommand{\Q}{\mathbb{Q}} 
\newcommand{\I}{\mathbb{I}} 
\newcommand{\R}{\mathbb{R}} 
\newcommand{\F}{\mathcal{F}} 

\newcommand{\U}{\mathcal{U}} 
\newcommand{\W}{\mathbb{W}} 
\newcommand{\q}{\mathcal{Q}} 
\newcommand{\h}{\mathcal{H}} 
\newcommand{\D}{\mathcal{D}} 
\newcommand{\V}{\mathcal{V}}
\newcommand{\set}{\mathcal{O}}
\DeclareMathOperator*{\argmin}{argmin}
\DeclareMathOperator*{\minimize}{minimize}
\newcommand{\norm}[1]{\left\lVert#1\right\rVert}

\author{Fred Espen Benth\thanks{\textsc{Department of Mathematics, University of Oslo, 0316 Blindern, Norway}; \textbf{fredb@math.uio.no}} 
	\and Nils Detering\thanks{\textsc{Department of Statistics and Applied Probability, University of California, Santa Barbara, CA 93106, USA}; \textbf{detering@pstat.ucsb.edu}} 
	\and Silvia Lavagnini\thanks{\textsc{Department of Mathematics, University of Oslo, 0316 Blindern, Norway}; \textbf{silval@math.uio.no}}}

\title{Accuracy of Deep Learning \\in Calibrating HJM Forward Curves}

\begin{document}
	\maketitle
\begin{abstract}
	We price European-style options written on forward contracts in a commodity market, which we model with an infinite-dimensional Heath-Jarrow-Morton (HJM) approach. For this purpose we introduce a new class of state-dependent volatility operators that map the square integrable noise into the Filipovi{\'{c}} space of forward curves. For calibration, we specify a fully parametrized version of our model and train a neural network to approximate the true option price as a function of the model parameters. This neural network can then be used to calibrate the HJM parameters based on observed option prices. We conduct a numerical case study based on artificially generated option prices in a deterministic volatility setting. In this setting we derive closed pricing formulas, allowing us to benchmark the neural network based calibration approach. We also study calibration in illiquid markets with a large bid-ask spread. The experiments reveal a high degree of accuracy in recovering the prices after calibration, even if the original meaning of the model parameters is partly lost in the approximation step.
\end{abstract}


\section{Introduction}
\label{introduction}
We follow the Heath-Jarrow-Morton (HJM) approach and model the commodity forward curve by a stochastic partial differential equation (SPDE) with values in the Filipovi{\'{c}} space and state-dependent volatility operator. In our setting, the Hilbert space valued Wiener process driving the forward curve takes values in $L^2(\set)$, $\set$ being some Borel subset of $ \R$ (possibly $\R$ itself). This requires that the volatility operator must smoothen elements in $L^2(\set)$ into elements of the Filipovi{\'{c}} space. We achieve this by constructing the volatility as an integral operator with respect to some suitably chosen kernel function. 

We then focus on the pricing of European-style options on forward contracts with delivery period, also called swaps. Typical examples are forward contracts in the electricity market, such as the ones traded at Nord Pool AS and the European Energy Exchange (EEX). In particular, when considering stochastic volatility operators, as in a state-dependent model, it is in general not possible to derive closed pricing formulas for options written on the forward curve. Hence one has to resort to time consuming numerical methods, such as Monte Carlo techniques for SPDEs (see \cite{barth12}, \cite{barth13}) or Finite Elements methods (see \cite{barth12a}, \cite{kov10}). Such costly procedures render calibration almost impossible as the pricing function has to be evaluated in each step of the optimization routine. As a result, these more accurate models are often not used in practice.

To overcome the computational burden and to make infinite dimensional models more tractable, especially from a calibration point of view, machine learning may be useful. Machine learning is in fact replacing standard techniques in many different fields of scientific computing, and recently also in finance. The first application of machine learning to option pricing is probably to be credited to \cite{Hutchinson}, where the authors propose to use a neural network for estimating the pricing formula of derivatives in a purely data-driven way without any parametric assumption on the dynamics of the underlying asset. More recently, in \cite{green}, for example, machine learning techniques are used for the evaluation of derivatives, and in \cite{teich19} for the problem of hedging a portfolio of derivatives with market frictions; in \cite{madan18} and \cite{kondratiev} machine learning tools are employed to learn the volatility surface and the term structure of forward crude oil prices, respectively. For model calibration, a first successful attempt is due to \cite{Hernandez}, who proposed to train a neural network that returns the optimal model parameters  based on market price observations. This approach has then been extended in \cite{deep2} by approximating the direct pricing map instead, and invert it afterwards in a second calibration step. Recently, \cite{christa} suggested to employ tools from generative adversarial networks to solve the calibration problem in the context of local stochastic volatility models. 
Finally, \cite{dixon} develops a neural network approach for interpolation of European vanilla option prices which also enforces arbitrage conditions.

Here we adapt the strategy presented in \cite{deep2} to the infinite dimensional HJM setting described above, and approximate the pricing function by a neural network in an off-line training step. Then, for calibration to market data, we use the trained neural network to recover the model parameters in an optimization routine. The training step is a priori cost demanding, especially if the underlying stochastic model is complex and pricing requires costly simulations, as for example in a stochastic volatility model. But even if generating the training set and training the neural network is computationally expensive, the advantage is that it is performed one-off. To ensure that the model reflects the current market situation, it is sufficient to rerun the calibration step regularly, say daily or even intra-daily. This step requires only evaluation of the neural network and is therefore fast.

In order to apply this procedure, we specify a fully parametric version of our HJM model. To access the accuracy of neural networks for pricing and calibration in the infinite dimensional HJM setting, we restrict the numerical experiments to a setting with deterministic volatility. In this setting we derive analytic pricing formulas based on a representation theorem for swaps presented in \cite{fredkru1}. We can thus benchmark the neural network based pricing map without introducing additional error due to Monte Carlo simulation. For the calibration step we consider the two different approaches presented in \cite{nn} for a stochastic volatility model in finite dimension, the pointwise and the grid-based learning approach. We also extend the framework to allow for calibration in markets with a wide bid and ask spread, where using a mid price is not feasible.  To our knowledge, this is the first application of deep neural networks in an infinite dimensional HJM setup for the purpose of model calibration. 

In the approximation step we observe a high degree of accuracy, with an average relative error for the test set in the range $0.3\%$-$3\%$. The picture is different when it comes to calibration. Here the trained neural network might fail to recover the true parameters, and we do indeed observe average relative errors reaching almost $50\%$ in some cases. In the specified model, several parameter vectors lead to similar prices for the training set, making it difficult to recover the true parameters. Moreover, the trained neural network itself can be non-injective in the input parameters, and this may cause the original meaning of the parameters to get lost in the approximation step. Nevertheless, the prices estimated after calibration have an average relative error around $5\%$. Hence the level of accuracy achieved for the prices shows that neural networks may indeed be an important and promising tool to make infinite dimensional models more tractable.

For the calibration in market environments with bid-ask spread, we propose a simple loss function which works rather well. We test it with respect to different bid-ask spread sizes, and in fact, after calibration, almost all prices lie within the bid-ask range and only few are outside, but still very close to either the bid or the ask price. In particular, whenever the bid-ask spread becomes more narrow, one can simply increase the number of iterations for the optimization routine to obtain good results. The optimization is fast because it does not require any simulation. Moreover, the observed errors in recovering the parameters are not increasing dramatically as compared to the calibration based on a zero bid-ask spread.

The rest of the article is organized as follows. In Section \ref{HJMsection} we define the HJM forward curve dynamics and our volatility operators, and we specify a deterministic version of it. In Section \ref{forwardsection} we introduce forward contracts with delivery period and options written on them, and we derive the pricing formulas used in our case study. In Section \ref{nnsection} we define neural networks and introduce the two-step approach for model calibration, together with the newly proposed bid-ask loss function. Finally, in Section \ref{oursetting} we specify the parametric setting for the experiments, in Section \ref{numericalsection} we show our findings and Section \ref{conclusion} concludes. Appendix \ref{appendix:proofs} contains the proofs to all results, while in Appendix \ref{injectivity} we analyse the non-injectivity issue related to the calibration problem.

\section{The forward curve dynamics}
\label{HJMsection}
In energy markets, forwards and futures on power and gas deliver the underlying commodity over a period of time, rather than at a fixed delivery time. While one usually derives the arbitrage-free forward price in a model free way from the buy-and-hold strategy in the underlying, in energy markets this strategy can not be applied because storing electricity is costly. This implies that the forward price as derived from the spot price model is not backed by a replication strategy. Instead, what is often adopted as alternative in energy markets is the direct modelling of the tradable forward prices. This is referred to as the Heath-Jarrow-Morton (HJM) approach, as it has first been introduced by Heath, Jarrow and Morton in \cite{HJM} for interest rate markets. In this setting, the original idea was to model the entire forward rate curve directly because short-rate models are not always flexible enough to be
calibrated to the observed initial term-structure \cite{filipovic2}. Later this modelling idea has been transferred to other markets. The authors of \cite{carmona} and \cite{kallsen}, for example, model the whole call option price surface using the HJM methodology, while the authors of \cite{clew} and \cite{fredsteen} transferred the approach to commodity forward markets, the latter one, in particular, in the context of power markets. 

Another important characteristic of energy forward markets is the high degree of idiosyncratic risk across different maturities, which has been observed by several studies, such as \cite{andre10}, \cite{koeke05} and \cite{frestad08}. In \cite{fredbook}, for example, the authors performed a Principal Component Analysis on the Nord Pool AS forward contracts, revealing that more than ten factors are needed to explain $95\%$ of the volatility. This points out the necessity of modelling the time dynamics of the forward curve by a high dimensional, possibly infinite dimensional, noise process. In \cite{fredkru2} the authors show that a reasonable state space for the forward curve is the so-called Filipovi{\'{c}} space, which is a separable Hilbert space first introduced by \cite{filipovic}. We shall adopt in this paper a similar framework, which is now introduced in detail.

Let $\left(\Omega, \F, (\F_t)_{t\geq 0},\mathbb{Q}\right)$ be a filtered probability space, with $\mathbb{Q}$ the risk-neutral probability. We denote by $\h_{\alpha}:= \h_{\alpha}(\R_{+})$ the Filipovi{\'{c}} space on $\R_{+}$: for a given continuous and non-decreasing function $\alpha : \R_{+}\to [1,\infty)$ with $\alpha(0)=1$ this is the Hilbert space of all absolutely continuous functions $f: \R_{+} \to \R$ for which
\begin{equation}
	\label{filipoviccond}
	\int_{\R_+}f'(x)^2\alpha(x)dx < \infty .
\end{equation}
It turns out that $\h_{\alpha}$ is a separable Hilbert space with inner product defined by
\begin{equation*}
	\langle f_1,f_2\rangle_{\alpha} := f_1(0)f_2(0)+\int_{\R_+}f_1'(x)f'_2(x)\alpha(x)dx,
\end{equation*}
for $f_1, f_2\in\h_{\alpha}$  and norm $\norm{ f}^2_{\alpha}:= \langle f,f\rangle_{\alpha}$.
We assume $\int_{\R_+}\alpha^{-1}(x)dx<\infty$. A typical example is $\alpha(x)= e^{\alpha x}$, for $\alpha >0$. We point out that condition \eqref{filipoviccond} reflects the fact that the forward curves flatten for large time to maturity $x$. Moreover, the term $f_1(0)f_2(0)$ in the inner product definition is required to distinguish two functions $f_1$ and $f_2$ such that $f_1-f_2$ is constant. At the same time it ensures that $\langle f,f\rangle_{\alpha}=0$ if and only if $f$ is the function constantly equal to zero. We refer to \cite{filipovic} for more properties of $\h_{\alpha}$.

For a Borel set $\set\subseteq\R$, we define $\h:= L^2(\set)$ as the Hilbert space of all square integrable functions. We further denote by $\lambda_{\set}$ the Lebesgue measure induced on $\set$, and by $\norm{\,\cdot \,}$ and $\langle \,\cdot\,,\,\cdot\,\rangle$ the norm and scalar product on $\h$. We assume that $\set$ is such that $\h$ is a separable Hilbert space, and we shall consider $\h$ as the noise space, similar as in \cite{fredbarth}, \cite{fredkru2} and \cite{fredflor}. 

A random variable $X$ with values in $\h$ is said to have a Gaussian distribution if for every $h\in \h$ there exist $\mu,\sigma \in \mathbb{R}$ with $\sigma \geq 0$ such that $\langle h,X \rangle $ is normally distributed on $\mathbb{R}$ with mean $\mu$ and variance $\sigma^2$ (i.e. $X \sim \mathcal{N} (\mu,\sigma^2)$).
	Following \cite[Sec. 2.3.1]{sdebook}, for a Gaussian distributed random variable $X$ with values in $\h$  there exist $m\in \h $ such that $\E [\langle h ,X \rangle ] = \langle m ,h \rangle$ for all $h\in \h$ and an operator $\q\in \mathcal{L}(\h, \h)$ with finite trace such that $\E [ \langle h_1,X-m \rangle \langle h_2,X-m  \rangle]= \langle \q h_1 ,h_2 \rangle $ for all $h_1,h_2 \in \h$. The vector $m\in \h$ is called the {\em mean} and $\q$ the {\em covariance operator}. We write $X\sim \mathcal{N} (m ,\q )$. 
	
	Following \cite[Sec. 4.1.1, Def. 4.2]{sdebook} an $\h$-valued process $\W :=\{\W(t)\}_{t\ge 0}$ is called $\q$-Wiener process if 
	\begin{itemize}
		\item $\W(0) = 0$, 
		\item $\W$ has continuous trajectories,
		\item $\W$ has independent increments,
		\item $\W (t) - \W (s) \sim \mathcal{N} (0, (t-s) \q)$.
	\end{itemize}

When it becomes clear from the context we will just refer to Wiener process instead of $\q$-Wiener process. We observe that for an $\h$-valued Wiener process it holds that $\{ \langle h, \W(t) \rangle \}_{t\ge 0}$ is a real-valued Wiener process, for every $h\in \h$. Thus the definition of an $\h$-valued Wiener process naturally generalizes Wiener processes in $\mathbb{R}^d$ in which case $\{ \langle h, \W(t) \rangle \}_{t\ge 0}$ is just a weighted sum of the components of the Wiener process itself. Since the elements in $\h$ are functions defined on the uncountable domain $\mathcal{O}$, it is however important to keep in mind that an $\h$-valued Wiener process is not simply an uncountable collection of one-dimensional Wiener processes indexed by elements in $\mathcal{O}$.

Let now $f(t, T)$ be the forward price in the commodity market at time $t\le T$ for a contract with time of delivery $T$. We then turn to the Musiela parametrization denoted by $g(t,x):= f(t, t+x)$ for $x:= T-t$. Since forward contracts are tradable assets, they should be martingales under $\Q$ and the drift in the dynamics of $f(t, T)$ has to be zero. However, when introducing the Musiela parametrization, we get an extra term in the drift of the dynamics for the curve. Indeed the differential of $g$ is
\begin{equation*}
	d g(t, x) = d f(t, t+x) + \partial_T f(t, t+x)dt =  d f(t, T) +\partial_x g(t, x)dt,
\end{equation*}
where $d f(t, T)$ must have zero drift due to the martingale condition. Then, similar to \cite[Section 3.2]{fredkru1}, a reasonable dynamics of $g$ under the risk neutral measure is given by
\begin{equation}
	\label{sdeg}
	dg (t,x) = \partial _x g (t,x) dt + \sigma (t,x) d\W (t, x),
\end{equation} 
where $\partial_x$ is the generator for the shift-semigroup $\{\U_t\}_{t\ge 0}$ given by $\U_tg(x)=g(t+x)$, for any $t,x\in \R_{+}$ and $g\in \h_{\alpha}$. Here 
$\sigma (t,\cdot)\in \mathcal{L}(\mathcal{H},\mathcal{H}_{\alpha})$ is a linear and bounded operator from $\h$ to $\h_{\alpha}$. We point out that in interest rates markets, the bonds are tradable assets. Hence a crucial requirement is that the discounted bond price process is a martingale under the risk-neutral probability measure $\Q$. However, since in energy markets we require the forward prices themselves to be martingales under $\Q$, the condition of null drift must be imposed on the forward dynamics instead.

We want now to rewrite equation \eqref{sdeg} in order to model the entire forward curve $\left(g(t,x)\right)_{x\ge 0}$ for any time $t\ge0$, and to allow for the volatility function $\sigma$ to be stochastic itself. Without introducing any other external noise source, for instance with a second dynamics for the volatility, we consider $\sigma$ to be state-dependent, namely depending on the current level of the forward curve $g$. We then obtain the following SPDE
\begin{equation}
	\label{sdeg2}
	dg_t = \partial _x g_t dt + \sigma_t (g_t) d\W_t,
\end{equation} 
where $g_t := g(t, \cdot)$, $\sigma_t(g_t) := \sigma (t, g_t)$ and $\W_t := \W (t, \cdot)$. For our purpose, we are interested in mild solutions of equation \eqref{sdeg2} in terms of the semigroup $\{\mathcal{U}_t\}_{t\geq 0}$. These are defined for every $\tau \ge t$ and for the initial condition $g_t = g(t, \cdot )\in \h_{\alpha}$ by
\begin{equation}
	\label{mildsol}
	g_{\tau} = \U_{\tau-t} g _t + \int_t^{\tau}\U_{\tau-s}\sigma_s (g_s) d\W_s.
\end{equation}
The mild solution is the natural extension of the variation of constants method for ordinary differential equations. Here $\mathcal{U}_t$ plays a similar role as the exponential $\exp{(tA)}$ in the solution $v: \mathbb{R} \rightarrow \mathbb{R}^d$ of the homogeneous linear equation $dv-Av dt=0$ for $A \in \mathbb{R}^d \times \mathbb{R}^d$. However, since $\partial_x$ is unbounded, the operator exponential $\exp{(t \partial_x)}$ is not defined and one needs new concepts to relate $\partial_x$ and $\mathcal{U}_t$ \cite{EngelNagel}. We finally notice that, since $\sigma_t (g_t) \in \mathcal{L}(\mathcal{H},\mathcal{H}_{\alpha})$ is an operator, then by writing $\sigma_t (g_t) d\W_t$ in equation \eqref{sdeg2} or \eqref{mildsol}, we in fact mean the application of $\sigma_t (g_t)$ to $d\W_t$.

The following theorem states conditions to ensure that the mild solution \eqref{mildsol} is well defined and unique for every initial condition $g_t\in \h_{\alpha}$.
\begin{thrm}
	\label{solutionth}
	Let us assume that the mapping
	\begin{equation*}
		\sigma: \R_{+}  \times \h_{\alpha} \to  \mathcal{L}(\h, \h_{\alpha}), \quad(t, g_t) \mapsto   \sigma_t(g_t)
	\end{equation*}
	is measurable and that there exists an increasing function $C:\R_{+}\to \R_{+}$ such that for all $f_1,f_2 \in \h_{\alpha}$ and $t \in\R_{+}$ we have
	\begin{align*}
		&\norm{\sigma_t(f_1) - \sigma_t(f_2)}_{\mathcal{L}(\h, \h_{\alpha})} \le C(t)\norm{f_1-f_2}_{\alpha},\\
		&\norm{\sigma_t(f_1)}_{\mathcal{L}(\h, \h_{\alpha})}\le C(t)(1+\norm{f_1}_{\alpha}).
	\end{align*}
	Then for every $\tau \ge t$ and every initial condition $g_t = g(t, \cdot )\in \h_{\alpha}$, there exists a unique mild solution to equation \eqref{sdeg2} given by equation \eqref{mildsol}.
\end{thrm}
\begin{proof}
	We refer to \cite{T2012}.  
\end{proof}

\subsection{The volatility operator}
\label{volsection}
We focus in this section on possible specifications for the volatility operator $$\sigma: \R_{+}  \times \h_{\alpha} \to  \mathcal{L}(\h, \h_{\alpha}).$$ The results derived here might be of interest independent of this study. In our framework, the volatility operator $\sigma_t(f)$ must turn elements of $\h$ into elements of $\h_{\alpha}$. It thus has to smoothen the noise, and one way to do that is by integrating it over a suitably chosen kernel function. The following theorem states sufficient conditions on such a kernel function for our purpose.
\begin{thrm}
	\label{kernel:cond}
	For $t\ge0$, let $\kappa_t:  \mathbb{R}_+ \times \set \times \h_{\alpha} \rightarrow \mathbb{R}_+$ be a kernel function satisfying the following assumptions: 
	\begin{enumerate}
		\item $\kappa_t (x,\cdot,f) \in \h$ for every $x\in  \mathbb{R}_+, f \in \h_{\alpha}$. 
		\item For every $x\in \mathbb{R}_+ ,f \in \h_{\alpha}$, the derivative $\frac{\partial \kappa_t (x,y, f) }{\partial x}$ exists for $\lambda_{\set}$ almost all $y\in \set $. Moreover, there exist a neighbourhood $I_x$ of $x$ and a function $\bar{\kappa}_x \in \h $ such that $\left|\frac{\partial \kappa_t (x,y,f) }{\partial x} \right|  \leq \bar{\kappa}_x(y)$ for $\lambda_{\set}$ almost all $y$ on $I_x$.
		\item  $\int_{\mathbb{R}_+} \norm{ \frac{\partial \kappa_t (x,\cdot,f) }{\partial x} }^2 \alpha (x) dx < \infty $.
	\end{enumerate}
	Then
	\begin{equation}
		\label{sigma:operator}
		\sigma_t(f) :\h \rightarrow \h_{\alpha}, \quad h\mapsto \sigma_t(f) h := \int_{\set} \kappa_t (\cdot,y,f) h(y) dy
	\end{equation}
	is a linear and bounded operator from $\h$ to $\h_{\alpha}$, namely $\sigma_t(f) \in \mathcal{L}(\h, \h_{\alpha})$. 
	In particular, for every $x \in \R_{+}$, the equality $\sigma_t(f) h(x) = \langle \kappa_t (x,\cdot, f), h \rangle$ holds.
\end{thrm}
\begin{proof}
	See Appendix \ref{kernel:cond:proof}.
\end{proof}

In Section \ref{forwardsection} we also need the adjoint operator of $\sigma_t(f)$,
namely the operator $\sigma_t(f)^{*}\in \mathcal{L}(\h_{\alpha},\h)$ which for every $f_1\in \h_{\alpha}$ and every $h\in \h$ satisfies $$\langle\sigma_t(f) h, f_1\rangle_{\alpha}=\langle h, \sigma_t(f)^{*}f_1\rangle.$$ By \cite[Theorem 6.1]{funcanal} any operator $\sigma_t(f)^{*}$ satisfying this equality is automatically bounded and is thus the unique adjoint operator of $\sigma_t(f)$. Its expression is given in the following theorem.
\begin{thrm}	
	\label{dualthrm}
	Under the assumptions of Theorem \ref{kernel:cond}, the adjoint operator $\sigma_t(f)^{*}$ is given by 
	\begin{align*}
		&\sigma_t(f)^{*} : \h_{\alpha} \rightarrow \h, \\& f_1 \mapsto
		\sigma_t(f)^{*} f_1:= \kappa_t(0,\cdot,f) f_1(0) + \int_{\R_{+}}\frac{\partial \kappa_t(x,\cdot,f)}{\partial x}f_1'(x)\alpha(x)dx.
	\end{align*}
	In particular, for every $y \in \R_{+}$, the equality $\sigma_t(f)^{*} f_1(y) = \langle \kappa_t (\cdot,y,f), f_1\rangle_{\alpha}$ holds.
\end{thrm}
\begin{proof}
	See Appendix \ref{dualthrm:proof}.
\end{proof}

We shall now look at conditions that also ensure that $\sigma_t (f)$ fulfils the assumptions of Lipschitz continuity and linear growth of Theorem \ref{solutionth}.
\begin{thrm}
	\label{kernel:lipcond}
	Let $\kappa_{t}: \mathbb{R}_+ \times \set \times \h_{\alpha}\rightarrow \mathbb{R}_+$ be a kernel function satisfying the assumptions of Theorem~\ref{kernel:cond}. Further there exists an increasing function $C: \R_{+}\to \R_{+}$ such that, for every $f_1, f_2\in \h_\alpha$, it holds:
	\begin{enumerate}
		\item $\norm{\kappa_{t} (0,\cdot, f_1)- \kappa_{t} (0,\cdot, f_2) } \leq C(t)  \abs{f_1(0) -f_2(0)}\,,$
		\item[] $\norm{\frac{\partial \kappa_{t} (x,\cdot, f_1)}{\partial x} - \frac{\partial \kappa_{t} (x,\cdot, f_2 )}{\partial x}} \leq C(t) \abs{f_1'(x)-f_2'(x)}\,.$
		\item $\norm{\kappa_{t} (0,\cdot,f_1)} \leq C(t) (1+ |f_1(0)|)\,,$
		\item[] $ \norm{\frac{\partial \kappa_{t} (x,\cdot,f_1 )}{\partial x}} \leq C(t) |f_1'(x)|\,.$
	\end{enumerate}
	Then $\sigma_t$ defined in equation \eqref{sigma:operator} satisfies the Lipschitz and growth conditions of Theorem~\ref{solutionth}.	
\end{thrm}
\begin{proof}
	See Appendix \ref{kernel:lipcond:proof}.
\end{proof}

If the kernel function $\kappa_t$ satisfies the assumptions of Theorem \ref{kernel:cond} and Theorem \ref{kernel:lipcond}, then there exists a mild solution to equation \eqref{sdeg2}, which models the dynamics of the forward curve.

\subsection{A deterministic specification}
\label{nonrandomvol}
In the previous section we have defined the volatility as an integral operator with respect to a kernel function $\kappa_t$. We now specify $\kappa_t$ in order to reflect some properties that we believe to be crucial for energy markets. For instance, we shall include time dependency to account for seasonality. Seasonality was first observed in electricity markets by \cite{fredsteen} for the volatility structure in a log normal model for the forward curve. The same authors also found a maturity effect, namely a monotone decay in the volatility when the time to maturity increases, also known as the Samuelson effect. Such an effect can be easily achieved with some decay function. Finally, we want to incorporate that contracts with a certain maturity are mainly influenced by the randomness of the noise in a neighbourhood of this maturity.

We restrict to a deterministic, time-dependent diffusion term. We therefore drop the state-dependency and define the kernel function $\kappa_t$ by 
\begin{align}
	&\label{kappa}
	\kappa_t(x,y) :=a(t) e^{-bx}\,\omega(x-y), \\
	&\label{seasonality} a(t) := a + \sum_{j=1}^J\left(s_j \sin(2\pi j t)+c_j\cos(2\pi j t)\right),
\end{align}
where $\omega: \R\to \R_{+}$ is a continuous weight function, while the term $e^{-bx}, b\ge0$, captures the Samuelson effect, and $a(t)$ is the seasonal function defined for $a\ge0$, $s_j$ and $c_j$ real constants, and $t$ measured in years.

In the following proposition, we state some assumptions on the weight function $\omega$ to ensure that $\kappa_t(x,y)$ as defined above fulfils the assumptions of Theorem~\ref{kernel:cond} for every $t\ge 0$.
\begin{prop}
	\label{weight:cond}
	Let $ \omega :  \R \rightarrow \mathbb{R}_+$ be such that:
	\begin{enumerate} 
		\item For every $x\in \R_+$, $\omega (x - \cdot)\left. \right|_{\set}  \in \h$.
		\item The derivative $\omega'(x)$ exists for almost all $x\in \R$ and whenever it exists, there exists a neighbourhood $I_x$ of $x$ and a function $\bar{\omega}_x \in \h $ with $\norm{\bar{\omega}_x} \leq C_1$ for some $C_1$ independent of $x$, and such that $| (\omega'(x-y)-b\omega(x-y)) \left. \right|_{\set}  \leq \bar{\omega}_x(y)$ on $I_x$.
	\end{enumerate}
	Let further $\int_{\mathbb{R}_+} e^{-2bx} \alpha (x) dx < \infty $.
	Then for every $t\ge0$ the volatility operator $\sigma_t$ given by $$\sigma_t :\h \rightarrow \h_{\alpha}, \quad h\mapsto \sigma_t h := \int_{\set} \kappa_t (\cdot,y) h(y) dy$$ is well defined, and satisfies the Lipschitz and linear growth conditions of Theorem \ref{solutionth}.
\end{prop}
\begin{proof}
	See Appendix \ref{weight:cond:proof}.
\end{proof}

For the function $\alpha(x) = e^{\alpha x}$ for instance, the integrability assumption of Proposition \ref{weight:cond} is satisfied if $0<\alpha<2b$.

\section{Forward contracts with delivery period}
\label{forwardsection}
Now we consider energy forward contracts with a delivery period, which we call swaps to avoid confusion with the contracts discussed in the previous section. For $0\le t\le T_1 \le T_2$, we denote by $F(t, T_1, T_2)$ the price at time $t$ of a swap contract on energy delivering over the interval $[T_1, T_2]$. Following \cite[Proposition 4.1]{fredbook}, this price can be expressed by
\begin{equation}
	\label{priceF}
	F(t,T_1,T_2) = \int_{T_1}^{T_2} w(T; T_1, T_2)f(t,T)dT,
\end{equation}
$f(t,T)$ being the forward price introduced above, and $w(T;T_1,T_2)$ a deterministic weight function. Focusing on forward style swaps in the electricity markets as traded, for example, at Nord Pool AS and the European Energy Exchange (EEX), the weight function takes the form 
\begin{equation}
	\label{wswap}
	w(T;T_1,T_2) = \frac{1}{T_2-T_1}.
\end{equation}

According to \cite{fredkru2}, we introduce the Musiela representation of $F(t,T_1,T_2)$. For $x:= T_1-t$ the time until start of delivery, and $\ell:= T_2-T_1>0$ the length of delivery of the swap, we define the weight function $w_{\ell}(t,x,y):= w(t+y;t+x,t+x+\ell)$. Motivated by practical examples (see \cite[Section 2]{fredkru2}), we shall consider only time-independent and stationary weight functions. With abuse of notation, let then $w_{\ell}:\R_{+}\to \R_{+}$ be bounded and measurable, such that
\begin{equation}
	\label{Gell}
	G_{\ell}^{w}(t,x):= F(t,t+x, t+x+\ell)=\int_{x}^{x+\ell}w_{\ell}(y-x)g_t(y)dy,
\end{equation}
for $g_t$ the mild solution in equation \eqref{mildsol}. With $w(T;T_1,T_2)$ as in equation \eqref{wswap}, one simply gets $w_{\ell}(y-x)= \frac{1}{\ell}$.

After performing a simple integration-by-parts on the right hand side of equation \eqref{Gell}, we obtain the following representation for $G_{\ell}^{w}$ as a linear operator $\D_{\ell}^w$ acting on $g_t\in \h_{\alpha}$:
\begin{equation}
	\label{Dell}
	G_{\ell}^{w}(t, \cdot) = \D_{\ell}^w(g_t)(\cdot), \mbox{ where } \D_{\ell}^w := \mathcal{W}_{\ell}(\ell)\mathrm{Id}+\mathcal{I}_{\ell}^w.
\end{equation}
Here $\mathrm{Id}$ is the identity operator, while
\begin{align}
	&\label{W}
	\mathcal{W}_{\ell}(u) := \int_{0}^{u}w_{\ell}(v)dv,\quad u\ge 0,\\
	&\label{int}
	\mathcal{I}_{\ell}^w(g_t)(\cdot) := \int_{0}^{\infty} q_{\ell}^{w}(\cdot, y)g_t'(y)dy,\\
	&\label{qell}
	q_{\ell}^{w}(x, y):= \left(\mathcal{W}_{\ell}(\ell)-\mathcal{W}_{\ell}(y-x)\right)\I_{[0,\ell]}(y-x).
\end{align}
The operator $\D_{\ell}^w$ transforms the curve of instantaneous forwards delivering at $x\geq 0$, into the corresponding curve of forwards delivering over a fixed duration $\ell>0$ (which are $w_{\ell}$-weighted sums of those instantaneous prices) starting at $x\geq0$. For a swap contract of forward type with weight function in equation \eqref{wswap}, it turns out that $\D_{\ell}^w$ has an easier representation as provided in the following lemma.
\begin{lem}
	\label{deliveryfunction}
	For a forward-style swap contract, the operator $\D_{\ell}^w$ can be represented as
	\begin{equation*}
		\D_{\ell}^w(g_t)(x) = \int_{0}^{\infty} d_{\ell}(x, y) g_t(y)dy,
	\end{equation*}
	where $d_{\ell}: \R_{+}\times \R_{+}\to \R_{+}, d_{\ell}(x,y):= \frac{1}{\ell}\I_{[x,x+\ell]}(y)$ is called the delivery period function.
\end{lem}
\begin{proof}
	See Appendix \ref{deliveryfunction:proof}.
\end{proof}

From equation \eqref{Gell} and \eqref{Dell} for every $0\le t\le \tau \le T_1\le T_2$, we can express the price of the swap contract at time $\tau$ with delivery over $[T_1, T_2 ]$ by
\begin{equation}
	\label{Foperator}
	F(\tau, T_1, T_2)= G_{\ell}^w(\tau, T_1-\tau)= \D_{\ell}^w(g_{\tau})(T_1-\tau) = \delta_{T_1-\tau} \D_{\ell}^w(g_{\tau}), 
\end{equation}
for $\ell= T_2-T_1$ the length of the delivery period and $\delta_{x}:\h_{\alpha}\to \R$ the point evaluation, namely the linear functional such that $\delta_x(f)= f(x)$ for every $f\in \h_{\alpha}$. We now combine equation \eqref{Foperator} with the expression of the mild solution $g_{\tau}$ in equation \eqref{mildsol}. Since the shift operator $\U_{x}$ and the operator $\D_{\ell}^w$ commute, namely $\U_x\D_{\ell}^w = \D_{\ell}^w \U_x$, and since $\delta_y\U_x = \delta_{x+y}$ for $x, y \in \R$, equation \eqref{Foperator} becomes
\begin{equation}
	\label{Ffirst}
	F(\tau, T_1, T_2) = \delta_{T_1-t}\D_{\ell}^wg_t+\int_{t}^{\tau}\delta_{T_1-s}\D_{\ell}^w\sigma_s(g_s)d\W_s.
\end{equation}
The expression shows that the price of the swap contract can be calculated by first applying a linear functional on the initial forward curve $g_t$.

Let now $\Gamma_s := \delta_{T_1-s}\D_{\ell}^w\sigma_s(g_s): \h \to \R$ with dual operator $\Gamma_s^*: \R \to \h$. The result of the operation $\Gamma_sd\W_s = \Gamma_s\left(d\W_s\right)$ can be seen as a trivial scalar product in $\R$ with the real element $1$, namely
\begin{equation}
	\label{RHS}
	\Gamma_s\left(d\W_s\right) = \langle 1, \Gamma_s\left(d\W_s\right)\rangle = \langle  \Gamma_s^*\left(1\right), d\W_s\rangle.
\end{equation}
In particular, by definition of the Wiener process $\W$, the right hand side of equation \eqref{RHS} has Gaussian distribution with variance $\langle \q \Gamma_s^*\left(1\right), \Gamma_s^*\left(1\right)\rangle ds = \langle \Gamma_s \q \Gamma_s^*\left(1\right), 1\rangle ds$. This allows us to introduce a one-dimensional Wiener process $W$ such that 
$\Gamma_sd\W_s = \left( \Gamma_s \q \Gamma_s^*\left(1\right)\right) dW_s.$
For every $0\le t\le \tau \le T_1$, we then rewrite equation \eqref{Ffirst} by
\begin{align}
	&\label{F}
	F(\tau, T_1, T_2) = \delta_{T_1-t}\D_{\ell}^wg_t+\int_{t}^{\tau}\Sigma_sdW_s,\\
	&\label{sigma}
	\Sigma^2_s= \left(\delta_{T_1-s}\D_{\ell}^w\sigma_s(g_s)\q\,\sigma_s(g_s)^* \left(\delta_{T_1-s}\D_{\ell}^w\right)^*\right)(1).
\end{align}
In particular, equation \eqref{F} tells us that the swap price $F(\tau, T_1, T_2)$ which is driven by the $\h$-valued Wiener process $\W$ with covariance operator $\q$ and volatility operator $\sigma_t$, can in fact be represented as driven by a one-dimensional Wiener process with diffusion term given in equation \eqref{sigma}. We refer to \cite[Theorem 2.1]{fredkru1} for a more rigorous proof.

We consider a covariance operator $\q$ of integral form 
\begin{equation*}
	\q h(x) = \langle h, q(x,\cdot)\rangle = \int_{\set}q(x,y)h(y)dy, \quad h\in \h,
\end{equation*}
with kernel $q(x,y)$ such that $\q$ is well defined. For a swap contract of forward type, using the volatility operator in equation \eqref{sigma:operator}, we can rewrite the univariate volatility in equation \eqref{sigma} more explicitly.

\begin{prop}
	\label{sigma:integral:prop}
	For a swap contract of forward type with weight function in equation \eqref{wswap}, the volatility $\Sigma^2_s$, $t\le s \le \tau$, defined in equation \eqref{sigma} is more explicitly given by the four-fold integral
	\begin{equation*}
		\Sigma^2_s = \int_{\R_{+}}\int_{\R_{+}}\int_{\set}\int_{\set}d_{\ell}(T_1-s, u)d_{\ell}(T_1-s, v)\kappa_s(v,z, g_s)  q(z,y)\kappa_s(u,y, g_s)dy dz du dv,
	\end{equation*}
	where $d_{\ell}$ is the delivery period function introduced in Lemma \ref{deliveryfunction}.
\end{prop}
\begin{proof}
	See Appendix \ref{sigma:integral:prop:proof}.
\end{proof}

In the deterministic setting introduced in Section \ref{nonrandomvol}, the univariate volatility formula of Proposition \ref{sigma:integral:prop} can be further simplified.
\begin{cor}
	\label{sigma:integral:cor}
	With a volatility kernel factorized as in equation \eqref{kappa}, the formula for $\Sigma^2_s$, $t\le s \le \tau$, in Proposition  \ref{sigma:integral:prop} is equivalent to
	\begin{equation*}
		\Sigma^2_s = a(s)^2\int_{\R_{+}}\int_{\R_{+}}\int_{\set}
		\int_{\set}e^{-bu}e^{-bv}d_{\ell}(T_1-s, u)d_{\ell}(T_1-s, v)\omega(v-z)q(z,y)\omega(u-y)dy dz du dv.
	\end{equation*}
\end{cor}
\begin{proof}
	This is a direct consequence of Proposition~\ref{sigma:integral:prop}. 
\end{proof}

\subsection{European options on swaps}
We focus on European-style options written on energy swap contracts. These derivatives are traded, for example, at Nord Pool AS. With the price of the swap at time $t$ being $F(t,T_1, T_2)$, we consider an option with payoff function $\pi:\R \to \R$ and exercise time $0\le \tau \le T_1$. Classical examples are standard call and put options with strike $K\ge 0$, for which the payoff function is defined by $\pi(x)=\max(x-K, 0)$, respectively $\pi(x)=\max(K-x, 0)$. 

It follows from equation \eqref{F} that the price at time $0\le t\le \tau$ of the option with payoff $\pi (F(\tau,T_1, T_2))$ at time $0\le\tau \le T_1$ is given by
\begin{equation}
	\label{option:price:formula}
	\Pi(t) = e^{-r(\tau-t)}\E\left[\left. \pi\left(\delta_{T_1-t}\D_{\ell}^wg_t+\int_{t}^{\tau}\Sigma_sdW_s\right)\right|\F_t\right].
\end{equation}
Here $r>0$ is the risk-free interest rate, considered to be constant. Assuming $\pi$ to be measurable and of at most linear growth, the expectation \eqref{option:price:formula} is well defined. 
We refer to \cite{fredkru2} for more details.

We end the section with a result from \cite{fredkru2}, which allows us to rewrite the price functional in equation \eqref{option:price:formula} for a deterministic volatility operator as the one we introduced in Section \ref{nonrandomvol}.
\begin{prop}
	\label{propprice}
	For $\sigma_t$ deterministic, the price functional in equation \eqref{option:price:formula} becomes
	\begin{align}
		&\label{option:price:formula2}
		\Pi(t) = 
		e^{-r(\tau-t)}\E\left[\left. \pi\left(\mu(g_t)+\xi X\right)\right| \F_t\right],\\
		&\label{xi}
		\xi^2 := \int_t^{\tau}\Sigma^2_sds, \\
		&\label{m}
		\mu(g_t):= \delta_{T_1-t}\D_{\ell}^w g_t,
	\end{align}
	with $X$ standard normal distributed random variable and $\Sigma^2_s$ in equation \eqref{sigma}.
\end{prop}
\begin{proof}
	We refer to \cite[Proposition 3.7]{fredkru2}.
\end{proof}

Proposition~\ref{propprice}, together with our Corollary~\ref{sigma:integral:cor}, allows us to calculate the price of an option on the forward curve explicitly after specifying the volatility and the covariance operator. Before providing a full specification, we explain in the next section the two-step approach to calibrate the HJM model with neural networks.

\section{The neural networks approach}
\label{nnsection}
For the purpose of calibration, we shall specify a fully parametric model, depending on a parameter vector $\theta$ taking values in a set $\Theta\subset\R^n$. In the framework described in Section \ref{HJMsection}, $\theta$ is a vector of parameters defining the volatility operator, the covariance operator and the initial forward curve. Moreover, the option price function depends on some features of the contract, such as time to delivery, strike, etc. We denote the vector of these contract parameters by $\lambda\in\Lambda \subset \R^{m}$. Then, the price function \eqref{option:price:formula} is $\Pi(t) = \Pi(t; \lambda, \theta)$. 

To obtain a fully specified price functional we need to calibrate the chosen model and determine the vector $\theta$ that best matches the observed prices of liquidly traded options. This will give us the best coefficient functions (in terms of calibration) for the HJM model defined in Section  \ref{HJMsection}. We do this by the two-step approach with neural networks presented in \cite{deep2}, which we adapt to our setting. In what follows, we first define feedforward neural networks, and then the calibration problem, together with the two-step approach.

\subsection{Feedforward neural networks}
We define an $L$-layer feedforward neural network as a function $\mathcal{N}:\R^d\to \R^p$ of the form
\begin{equation}
	\label{NNdef}
	\mathcal{N}(x) := H_L(\rho(H_{L-1}(\rho(\ldots \rho(H_1(x)))))),
\end{equation}
where each $H_i: \R^{n_{i-1}}\to \R^{n_i}$ is an affine map of the form $H_i(x)= V_{i}x+v_{i}$, for $V_{i}\in \R^{n_{i}\times n_{i-1}}$ and $v_i\in \R^{n_{i}}$, with $n_0 = d$ and $n_L = p$. We set $\mathbf{n}:= (n_0, \dots, n_L)$, where $n_i$ represents the number of nodes of the $i$-th layer and $L$ is the depth of the network. The map $\rho: \R \to \R$ is the so-called activation function, which is typically non-linear and applies component wise on the output of the affine maps $H_i$.  Standard choices are, for example, the Rectified Linear Unit (ReLU), $\rho(x) = \max(x,0)$, or the Exponential Linear Unit (ELU), $\rho(x) = \max(e^x-1, x)$. We denote by $\V$ the set of all parameters involved in the linear maps and use the notation $\mathcal{N} = \mathcal{N}(x; \V)$. The cardinality of $\V$ is then $M:=|\V| = \sum_{i=1}^Ln_i(n_{i-1}+1)$. The architecture of the neural network, namely the depth, the number of nodes per layer and the activation function, are the hyper-parameters which must be chosen accordingly to the problem of interest.

Given a set of input-output pairs $\{\left(x_i, \varphi(x_i)\right)\}_{i=1}^{N}$ of size $N$, we can approximate an unknown map $\varphi: \R^d\to \R^p$ with neural networks. This is an optimization problem, called the training of the neural network, which is solved by finding the best set of parameters $\hat{\V}\in \R^M$ so that the neural network output $\mathcal{N}(x; \hat{\V})$ best approximates the observations $\{ \varphi(x_i)\}_{i=1}^{N}$ with respect to some loss function to be chosen, usually the mean squared error:
\begin{equation*}
	\label{calibrationNN}
	\minimize_{\V\in\R^M}  \frac{1}{N}\sum_{i=1}^N \left(\mathcal{N}(x_i; \V)-\varphi(x_i)\right)^2.
\end{equation*}
Given the optimal weights $\hat{\V}$, we denote by $\hat{\mathcal{N}}(x) := \mathcal{N}(x; \hat{\V})$ the trained neural network. We point out that this is a non-convex optimization problem and one typically only finds an approximation to a local solution. For more details on feedforward neural networks, activation functions and training of the network, we refer the reader to \cite{deeplearn} and \cite{Bengio16}.

\subsection{The calibration problem and the two-step approach}
\label{twosteps}
Now we focus on the option price $\Pi = \Pi (\lambda, \theta)$, omitting the time dependence to simplify the notation, and we consider $N$ option contracts with features $\{\lambda_i\}_{i=1}^N \in \Lambda$, whose prices $\{\Pi_i\}_{i=1}^N $ can be observed in the market. Calibrating the HJM model for the forward curve means to find a vector of model parameters $\hat{\theta}\in \Theta$ which minimizes the distance between the prices observed in the market and the corresponding prices given by the stochastic model, namely
\begin{equation}
	\label{calibration}
	\minimize_{\theta\in\Theta}\frac{1}{N}\sum_{i=1}^N\left(\Pi(\lambda_i, \theta)-\Pi_i\right)^2,
\end{equation}
where the chosen distance is the mean squared error. When no closed formula for the price functional $\Pi(\lambda, \theta)$ is available, the calibration problem in equation \eqref{calibration} must be adjusted by substituting $\Pi(\lambda, \theta)$ with an approximation $\tilde{\Pi}(\lambda, \theta)$. This approximation is obtained for example by Monte Carlo simulation, which makes the procedure costly. The more complex the underlying stochastic model, the greater the time needed for simulation and hence for calibration.

To overcome this issue, we divide the calibration problem in equation \eqref{calibration} into two-step. We  start by approximating the price functional $\Pi$ (or its approximation $\tilde{\Pi}$) with a neural network, $\mathcal{N}(\lambda, \theta; \V) \approx \Pi(\lambda, \theta)$. For this we generate a training set that consists of many different parameters and option specifications ($\theta$ and $\lambda$) and their corresponding prices, and we use this set to train a neural network (see Section \ref{numericalsection} for details). This step is computationally demanding because we need to generate a large training set by (potentially) costly simulations and because the training of the neural network is an $M=|\V|$ dimensional optimization problem, with $M$ usually large. However, the training is off-line since it does not use any market information, hence it does not require frequent updating.

The second step is calibration. We now replace in equation \eqref{calibration} the function $\Pi$ (or $\tilde{\Pi}$) with the trained neural network $\hat{\mathcal{N}}$, whose evaluation is fast and does not require any further simulation. Market data is used in this step, which therefore has to be updated frequently. However, this can be done with low computational cost as evaluation of neural networks is very fast and most machine learning frameworks have very efficient implementations. In what follows, we present two alternative approaches to this procedure. 

\subsubsection{The pointwise learning approach}
\label{pointwisesection}
Let $( \lambda, \theta)\in \Lambda\times \Theta\subset \R^d$, for $d= n+m$, such that $\Pi = \Pi( \lambda, \theta)$. In the pointwise learning approach, we approximate the pricing map $\Pi$ (or $\tilde{\Pi}$) by a neural network that maps the vector $(\lambda, \theta)$ into prices. Given the training set $\{(( \lambda_i, \theta_i), \Pi(\lambda_i, \theta_i) )\}_{i=1}^{N_{train}}$ for $(\lambda_i, \theta_i)\in \Lambda\times \Theta$ and $N_{train} $ the size of the training set, we train a neural network $\mathcal{N}:\Lambda\times \Theta\to \R_+$ by computing
\begin{equation*}
	\hat{\V} \in \argmin_{\V\in\R^M}  \frac{1}{N_{train}}\sum_{i=1}^{N_{train}} \left(\mathcal{N}(\lambda_i, \theta_i; \V)-\Pi(\lambda_i, \theta_i)\right)^2.
\end{equation*}
We then use the trained neural network $\hat{\mathcal{N}}( \lambda, \theta)= \mathcal{N}( \lambda, \theta; \hat{\V})$ for the calibration step with respect to the model parameters: we look for an approximate solution $\hat{\theta}\in \Theta$ to the problem
\begin{equation*}
	\minimize_{\theta\in\Theta}  \frac{1}{N_{cal}}\sum_{i=1}^{N_{cal}} \left(\hat{\mathcal{N}}(\lambda_i, \theta)-\Pi_i\right)^2,
\end{equation*}
where $\{\Pi_i\}_{i=1}^{N_{cal}}$ are observed market prices for derivatives with contract features $\{\lambda_i\}_{i=1}^{N_{cal}}$, and $N_{cal}$ is the size of the calibration set.

\subsubsection{The grid-based learning approach}
\label{gridsection}
We train a neural network which is a function only of the model parameters $\theta$. The output is a discrete grid of values corresponding to option prices for different specifications of $\lambda$. The options considered are those which are traded in the market and used for calibration. Let us suppose $m=2$ and $\lambda = (\lambda^1, \lambda^2)$. For $m_1, m_2\in \N$, we create a grid of values $\{(\lambda^1_j, \lambda^2_k)\}_{j=1, k=1}^{m_1, m_2}$ and the training set $\{(\theta_i, \{\Pi(\theta_i, (\lambda^1_j, \lambda^2_k))\}_{j=1, k=1}^{m_1, m_2} )\}_{i=1}^{N_{train}}$. We then train a neural network $\mathcal{N}:\Theta\to \R_+^{m_1\times m_2}$ by solving the following optimization problem
\begin{equation*}
	\hat{\V} \in \argmin_{\V\in\R^M}  \frac{1}{N_{train}}\frac{1}{m_1m_2}\sum_{i=1}^{N_{train}} \sum_{j, k=1}^{m_1, m_2}\left(\mathcal{N}(\theta_i; \V)_{j,k}-\Pi(\theta_i, (\lambda^1_j, \lambda^2_k))\right)^2.
\end{equation*}
We use the trained neural network $\hat{\mathcal{N}}(\theta)= \mathcal{N}(\theta; \hat{\V})$ in the calibration step to find the optimal model parameters $\hat{\theta}\in \Theta$ for fitting the market observations $\{\Pi_{j,k}\}_{j=1, k=1}^{m_1, m_2}$:
\begin{equation*}
	\minimize_{\theta\in\Theta}  \frac{1}{m_1m_2} \sum_{j, k=1}^{m_1, m_2}\left(\hat{\mathcal{N}}(\theta)_{j,k}-\Pi_{j,k}\right)^2.
\end{equation*}
The main difference of the grid-based approach compared with the pointwise one, is that the neural network is trained to price only specific options, namely options related to the grid $\{(\lambda^1_j, \lambda^2_k)\}_{j=1, k=1}^{m_1, m_2}$ defined in the first step. However, every price for a contract not included in this grid can possibly be obtained by interpolation.

\subsection{The bid-ask constraint}
\label{bidasksection}
In order-book based markets there is no single price, but the so-called bid and ask price corresponding to the cheapest sell and the most expensive buy order. Depending on market liquidity, the spread between bid and ask price (bid-ask spead) can be significant. Thus the calibration problem described in equation \eqref{calibration} breaks down as we do not have exact prices to aim at in a mean squared loss sense. Given bid and ask prices $\{(y_i^{bid}, y_i^{ask})\}_{i=1}^N$, to penalize the values $\{x_i\}_{i=1}^N$ lying outside the bid-ask range, we introduce the bid-ask loss function
\begin{equation*}
	\frac{1}{N}\sum_{i=1}^N \left\{\left(x_i-y_i^{bid}\right)^2\I_{\left\{ x_i<y_i^{bid}\right\}}	+ \left(x_i-y_i^{ask}\right)^2\I_{\left\{x_i>y_i^{ask} \right\}}	\right\},
\end{equation*}
which equals zero for those prices within the bid-ask interval, and it is a quadratic function of the distance to the boundary outside the interval, as it is shown in Figure \ref{loss_functions}.
\begin{figure}[!tbp]
	\centering
	\includegraphics[width=0.4\textwidth]{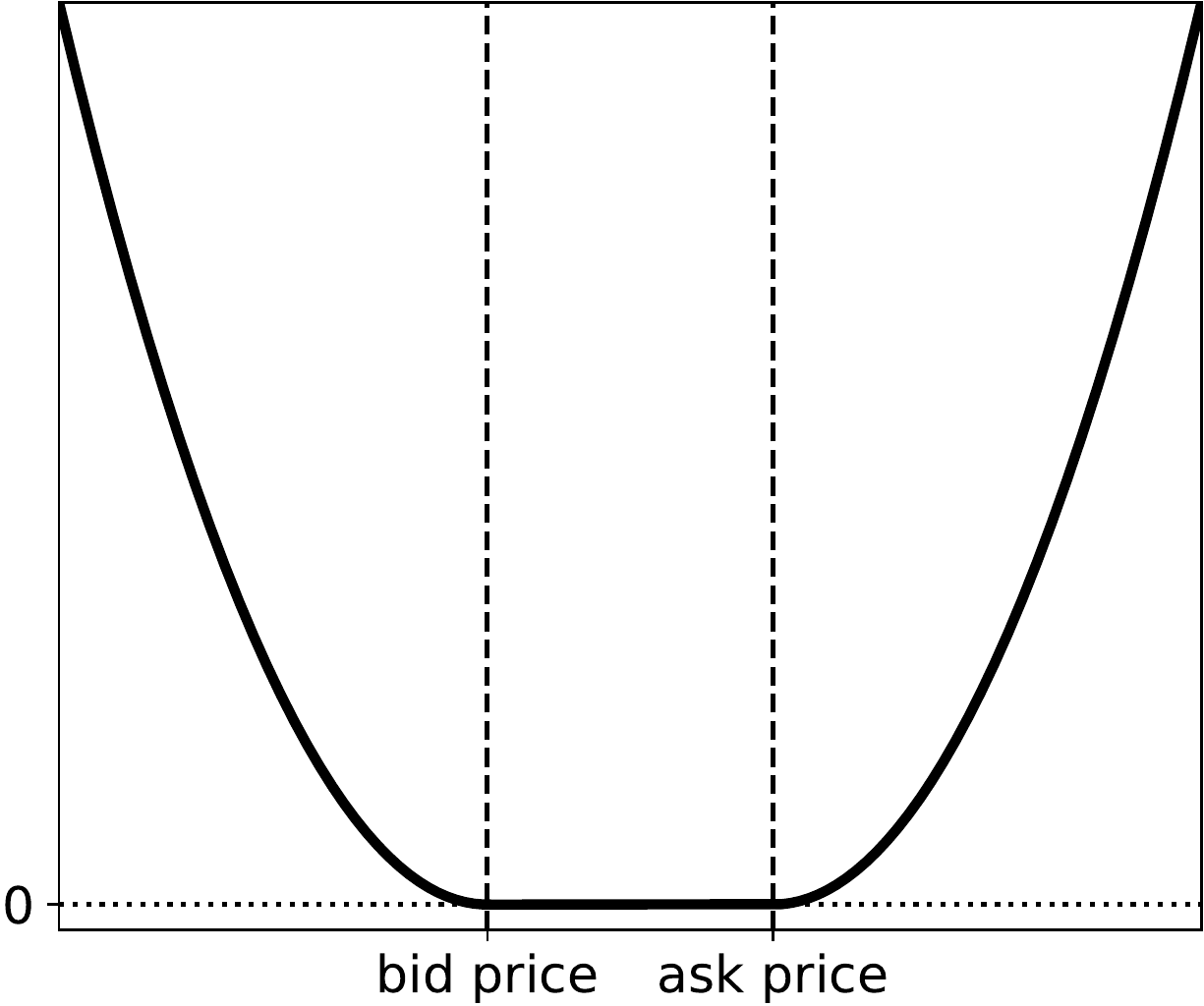}
	\caption{The bid-ask loss function.\label{loss_functions}}
\end{figure}

In the grid-based learning framework described in Section \ref{gridsection} with $m=2$ and $\lambda = (\lambda^1, \lambda^2)\in\{(\lambda^1_j, \lambda^2_k)\}_{j=1, k=1}^{m_1, m_2}$, we consider the market observations $\{\Pi_{j,k}^{bid}, \Pi_{j,k}^{ask}\}_{j=1, k=1}^{m_1, m_2} $, where $\Pi_{j,k}^{bid}$ and $\Pi_{j,k}^{ask}$ are, respectively, the bid price and the ask price of the contract with features $(\lambda^1_j, \lambda^2_k)$. Given the trained neural network $\hat{\mathcal{N}}:\Theta\to \R_+^{m_1\times m_2}$, we then look for the optimal model parameters $\hat{\theta}\in \Theta$ to fit the market observations:
\begin{equation*}
	\minimize_{\theta\in\Theta}  \frac{1}{m_1m_2} \sum_{j, k=1}^{m_1, m_2} \left\{\left(\hat{\mathcal{N}}(\theta)_{j,k}-\Pi_{j,k}^{bid}\right)^2\I_{\left\{\hat{\mathcal{N}}(\theta)_{j,k} < \Pi_{j,k}^{bid}\right\}}	+ \left( \hat{\mathcal{N}}(\theta)_{j,k}- \Pi_{j,k}^{ask} \right)^2\I_{\left\{\hat{\mathcal{N}}(\theta)_{j,k} > \Pi_{j,k}^{ask}\right\}}	\right\}.
\end{equation*}
The bid-ask loss function can be also applied to the pointwise learning in a similar manner.

\section{The setting for the experiments}
\label{oursetting}
We now provide a fully parametrized version of our model. Let $\alpha(x) = e^{\alpha x}$, $x \in \R_{+}$, be the weight function for the Filipovi{\'{c}} space $\h_{\alpha}$. By Proposition~\ref{weight:cond} the real constant $\alpha$ must satisfy $0< \alpha < 2b$, where $b$ models the Samuelson effect in the deterministic specification of the volatility in equation \eqref{kappa}. Moreover, we deal with swap contracts of forward type on energy with delivery over an interval $[T_1, T_2]$. Thus we take $w(T;T_1,T_2) = \frac{1}{T_2-T_1}$ (see equation \eqref{wswap}) and, consequently, $w_{\ell}(y-x)= \frac{1}{\ell}$. Then $\mathcal{W}_{\ell}(\ell) = 1$ (see equation \eqref{W}) and $q_{\ell}^{w}(x, y)= \frac{1}{\ell} \left(x+\ell-y)\right)\I_{[x,x+\ell]}(y)$ (see equation \eqref{qell}). In our case study we consider only European-style call options with payoff function $\pi(x) = \max(x-K, 0)$, for $K>0$ the strike price. However, the framework works of course equally well with other options.

In order to benchmark the two-step approach, as explained in Section \ref{introduction}, we focus on deterministic volatility operators, such as the one introduced in Section \ref{nonrandomvol}, which will be further specified here. For $\sigma_t$ deterministic, we have seen in Proposition \ref{propprice} that the price $\Pi(t)$ at time $0 \le t \le \tau$ of the option with payoff $\pi(F(\tau, T_1, T_2))$ at time $0 \le \tau \le T_1$, can be expressed in terms of a standard Gaussian random variable $X$, a variance $\xi^2 = \int_t^{\tau}\Sigma^2_sds$ and a drift $\mu(g_t)= \delta_{T_1-t}\D_{\ell}^w g_t$. In the case of a European-style call option, $\Pi(t)$ has a closed form solution, by means of a Black-76 type formula. This allows for exact price values to be used for training the model avoiding additional sources of error, such as resulting from a Monte Carlo simulation approach. As direct consequence of Proposition \ref{propprice}, we find in the next proposition explicitly the price functional for a call option.
\begin{prop}
	The price of a European-style call option with strike price $K>0$ and maturity time $\tau \le T_1$ is given by
	\begin{equation}
		\label{priceVexp}
		\Pi(t) = e^{-r(\tau-t)}\left\{\xi \phi\left(\frac{\mu(g_t)-K}{\xi}\right) + \left(\mu(g_t)-K\right)\Phi\left(\frac{\mu(g_t)-K}{\xi}\right)\right\},
	\end{equation}
	$\phi$ and $\Phi$ being, respectively, the density function and the cumulative distribution function of a standard Gaussian random variable.
\end{prop}
\begin{proof}
	The proof follows by direct calculation from equation \eqref{option:price:formula2} for the payoff function $\pi(x) = \max(x-K, 0)$, and using standard techniques for the expected value of a Gaussian random variable. \qed
\end{proof}

To compute the price in equation \eqref{priceVexp}, we need the variance $\xi$ and the shift $\mu(g_t)$, hence we need to specify the volatility operator $\sigma_t$ and the covariance operator $\q$. Last, we must parametrize an appropriate initial forward curve, $g_t\in \h_{\alpha}$. 

\subsection{The covariance operator}
\label{covsection}
We introduce a covariance operator $\q$ that depends only on a finite number of parameters. The analysis conducted by \cite{fredflor} reveals a covariance structure that is well approximated by an exponential function, i.e. $\text{Cov} (\W_t(x), \W_t(y) ) \approx  e^{-k|x-y|}$. Despite the operation $\W_t(x)=\delta_x (\W_t)$ not being well defined on our space $\h$, we can approximate it by the scalar product $\delta_x\!\left(\W_t\right)\approx \langle \eta_x, \W_t\rangle$, with some "bell-shaped" function $\eta_x$ centred in $x$, such as a Gaussian density function. For every $x, y \in \set$, denoting by $c(x,y)$ the empirical covariance function between $\W_t(x)$ and $\W_t(y)$, we then get:
\begin{align*}
	c(x,y)& = \E[\W_t(x)\W_t(y)]=\E[\delta_x(\W_t)\delta_y(\W_t)]\\
	& \approx \E[\langle \eta_x,\W_t\rangle \langle \eta_y,\W_t\rangle]=\langle \q \eta_y,\eta_x\rangle\\
	&\approx \q \eta_y(x)=\int_{\set} e^{-k|x-z|} \eta_y(z)dz\\
	&\approx e^{-k|x-y|},
\end{align*}
which shows that in $\h$ a covariance operator based on an exponential kernel indeed approximates the empirically observed covariance structure of the Wiener process $\W$ across different maturities. We thus define a one parameter covariance operator by
\begin{equation}
	\label{covQ}
	\q h(x) = \int_{\set}e^{-k|x-y|}h(y)dy, \quad h\in \h.
\end{equation}
Because $e^{-k|\cdot |}$ is the characteristic function of a Cauchy distributed random variable with location parameter $0$ and scale parameter $k$, it follows by Bochner's Theorem that it is positive-definite. Since it is also symmetric and continuous, for $\set$ compact it follows from \cite[Theorem A.8]{infdimbook} that $\q$ is in fact a covariance operator. In the following, we therefore choose $\set := [-\gamma, \gamma]$ for some large $\gamma$ which ensures that all maturities of interest are covered.

\subsection{The volatility operator}
\label{volsubsection}
We consider a specification of the volatility operator that does not depend on time and the current state $g_t$, which we denote by $\sigma$ instead of $\sigma_t$ to simplify the notation. We choose $a(t)=a \ge 0$ for every $t\in \R_{+}$, thus we do not account for seasonality and the level $a$ corresponds to the implied spot price volatility, as pointed out in \cite{fredsteen}. Moreover, we define the following weight function
\begin{equation}
	\label{omega}
	\omega(x):= (1-\abs{x})\,\I_{\{\abs{x}\leq 1\}} =
	\begin{cases}
		(1-\abs{x}) & \text{if $\abs{x}\leq 1$} \\
		0       & \text{otherwise}
	\end{cases}.
\end{equation}
Let us notice that this parameter-free specification of $\omega$ fulfils the assumptions of Proposition~\ref{weight:cond}. The kernel $\kappa_t$ in equation \eqref{kappa} becomes then $\kappa_t(x,y) = \kappa(x,y) = ae^{-bx} \omega (x-y)$, where $b\ge0$ determines the strength of the maturity effect. The volatility operator is given by
\begin{equation}
	\label{vol}
	\sigma h(x) = ae^{-bx} \int_{\set} (1-\abs{x-y})\,\I_{\{\abs{x-y}\leq 1\}} h(y) dy,
\end{equation}
and is well defined by Theorem~\ref{kernel:cond}. Let us recall that the role of $\sigma$ is to smoothen the noise from the space $\h$ to $\h_{\alpha}$, and we have achieved this with an integral operator. The weight function $\omega$ has then a double role. First of all, it functions as (a part of) the kernel for the integral operator which smoothen the noise. On the other hand, it weights the randomness coming from the Wiener process $\W_t$ so that a contract with time to maturity $x$ is only influenced by $\W_t(y)$, for $y$ in a neighbourhood of $x$. Other weight functions $\omega$ could be considered to obtain a similar weighting effect.

We have to calculate the volatility $\Sigma^2_s$ using the formula provided in Corollary \ref{sigma:integral:cor}. However, the expression turns out very cumbersome when integrating over $\set = [-\gamma, \gamma]$. For this reason, we integrate over $\mathbb{R}$ instead and calculate $\Sigma^2_s$ according to the formula
\begin{equation*}
	\Sigma^2_s := a^2\int_{\R_{+}}\int_{\R_{+}}\int_{\R}\int_{\R}e^{-bu}e^{-bv} d_{\ell}(T_1-s, u)d_{\ell}(T_1-s, v) \omega(v-z)q(z,y)\omega(u-y)dy dz du dv.
\end{equation*}
Since all terms are positive and $\Sigma^2_s<\infty$, it is then easy to see that, for
\begin{equation*}
	\Sigma^2_{s,\gamma}:=a^2 \int_{\R_{+}}\int_{\R_{+}}\int_{-\gamma}^{\gamma}\int_{-\gamma}^{\gamma}e^{-bu}e^{-bv}d_{\ell}(T_1-s, u)d_{\ell}(T_1-s, v)\omega(v-z)q(z,y)\omega(u-y)dy dz du dv,
\end{equation*}
the limit $\lim_{\gamma \rightarrow \infty} \Sigma^2_{s,\gamma}= \Sigma^2_s$ holds. We can calculate $\Sigma^2_s$ explicitly.
\begin{prop}
	\label{prop1}
	In the setting described above, the volatility $\Sigma^2_s$, for $t\le s \le \tau$, is given by
	\begin{equation*}
		\Sigma^2_s = \frac{2a^2}{kb^4\ell^2}\,e^{-2b(T_1-s)} \left\{ \frac{2}{3}\left(b^2+3\right)\left(1+e^{-2b\ell}\right)-2e^{-b\ell}\left(\frac{b^2}{6}\left(3(\ell-2)\ell^2+4 \right)-3\ell+2 \right) \right\}.
	\end{equation*}
\end{prop}
\begin{proof}
	See Appendix \ref{appendix:vol}.
\end{proof}

\noindent And, consequently, we can calculate $\xi^2$ in equation \eqref{xi}.
\begin{prop}
	\label{prop2}
	In the setting described above, we get
	\begin{multline*}
		\xi^2 = \frac{a^2}{kb^5\ell^2}\left( e^{-2b(T_1-\tau)}-e^{-2b(T_1-t)}\right)\cdot \\\cdot\left\{ \frac{2}{3}\left(b^2+3\right)\left(1+e^{-2b\ell}\right)-2e^{-b\ell}\left(\frac{b^2}{6}\left(3(\ell-2)\ell^2+4 \right)-3\ell+2 \right) \right\}.
	\end{multline*}
\end{prop}

\begin{proof}
	The results follows by integrating $\Sigma^2_s$ from Proposition \ref{prop1} over the interval $[t,\tau]$. 
\end{proof}

\subsection{The initial forward curve}
\label{NSsection}
Finally, we introduce a suitably parametrized initial forward curve, $g_t = g(t,\cdot)\in \h_{\alpha}$. We choose the Nelson-Siegel curve, which is defined for $x\ge 0$ by
\begin{equation}
	\label{NS}
	g_t(x) =g_{NS}(x) := \alpha_0 +\left(\alpha_1+\alpha_2\alpha_3x\right)e^{-\alpha_3x},
\end{equation}
for $\alpha_0, \alpha_1, \alpha_2, \alpha_3 \in \R$, $\alpha_3>0$. Let us notice that it does not depend on time $t$. This curve has been first introduced for modelling the forward rates in \cite{NS87}, and has already been applied in the context of energy markets by \cite{krigingfred}. We need to ensure that $g_{NS} \in \h_{\alpha}$.
\begin{lem}
	The Nelson-Siegel curve in equation \eqref{NS} belongs to $\h_{\alpha}$ if and only if $\alpha < 2\alpha_3$.
\end{lem}
\begin{proof}
	The condition is found by calculating the $\h_{\alpha}$-norm of $g_{NS}$, namely
	\begin{align*}
		\norm{g_{NS} }_{\alpha}^2 &= g_{NS}^2(0)+\int_{\R_{+}}g_{NS}'^{\,2}(x)\alpha(x)dx \\&= (\alpha_0+\alpha_1)^2+\int_{\R_{+}}\alpha_3^2\left(\alpha_2-\alpha_1-\alpha_2\alpha_3x\right)^2e^{(\alpha-2\alpha_3)x}dx,
	\end{align*}
	where the integral converges if and only if $\alpha - 2\alpha_3<0$. \qed
\end{proof}
With the explicit representation for $g_t$, we compute $\mu(g_t)$: from Lemma \ref{deliveryfunction} we get that
$$\mu(g_t) = \frac{1}{\ell}\int_{T_1-t}^{T_1+\ell-t}g_{NS}(y)dy,$$
and by direct computation, we then calculate the drift
\begin{multline}
	\label{NSdrift}
	\mu(g_t) = \frac{1}{\ell}\left\{\alpha_0\ell+\frac{1}{\alpha_3} e^{-\alpha_3(T_1-t)}\left(\alpha_1+\alpha_2+\alpha_2\alpha_3(T_1-t) \right)+\right.\\\left.-\frac{1}{\alpha_3}e^{-\alpha_3(T_1+\ell-t)}\left(\alpha_1+\alpha_2+\alpha_2\alpha_3(T_1+\ell-t) \right) \right\},
\end{multline}
which is the last component we need in order to get a fully specified option price functional.

\section{Implementation and results}
\label{numericalsection}
In this section, we describe implementation details and report our findings. From Section \ref{oursetting}, the vector of model parameter is $\theta= (a,b,k,\alpha_0, \alpha_1, \alpha_2, \alpha_3)\in \R^7$, i.e., $n=7$. In this vector, $a, b\ge 0$ are parameters of the volatility operator introduced in Section \ref{volsubsection}, and $k\ge 0$ is the parameter of the covariance operator, see Section \ref{covsection}. Finally $\alpha_0, \alpha_1, \alpha_2, \alpha_3 \in \R$ with $\alpha_3>0$ are the parameters for the Nelson-Siegel curve as introduced in Section \ref{NSsection}. 

Since we are considering European-style call options written on forward-style swaps with delivery period, the vector of contract features is given by $\lambda = (K, \tau, T_1, \ell)\in \R^4$, hence $m=4$, where $K>0$ is the strike price, $\tau$ is the time to maturity, $T_1 \ge \tau$ is the start of delivery of the swap, and, finally, $\ell>0$ is the length of the delivery period. Since for the grid-based learning approach we need a grid of values for $\lambda$, we set $T_1 = \tau$, namely the time to maturity for the option coincides with the start of delivery of the swap, and $\ell =1/12$, namely we consider only contracts with one month of delivery as the time unit is one year. Then $\lambda = (\tau, K)\in \R^2$ and $m=2$, so that the grid, as introduced in Section \ref{gridsection}, will be two-dimensional. Finally, we fix $t=0$ as evaluation time, and $r=0$ as risk-free interest rate.

In all the experiments we use the Adam optimizer \cite{adam}, both for training the neural network and in the calibration step, and we work with the mean squared error loss function. We set the number of epochs to $200$ in the approximation step and to $1000$ in the calibration step because the accuracy improves only marginally for larger values. The batch size is fixed to $30$ for all the experiments. We point out that we tested several different configurations for the neural networks (numbers of layers, number of nodes per layer and activation functions). The configuration reported in what follows led to the best results for our study.\footnote{The code is implemented in TensorFlow 2.1.0 and is available at \urlstyle{tt} \url{https://github.com/silvialava/HJM_calibration_with_NN}.}

\subsection{Grid-based learning}
\label{grid-dense}
We create a grid of values of times to maturity and strike prices which we believe to be reasonable for a case study in the electricity markets. Let $\Lambda^{grid}:= \Lambda^{grid}_{\tau}\times  \Lambda^{grid}_{K}\subset \R_{+}^{7 \times 9}$ with
\begin{align*}
	&\Lambda^{grid}_{\tau}=\{1/12, 2/12, 3/12, 4/12, 5/12, 6/12, 1\},\\
	&\Lambda^{grid}_{K}=\{31.6, 31.8, 32.0, 32.2, 32.4, 32.6, 32.8, 33.0, 33.2\}.
\end{align*} 
Here the range in $\Lambda^{grid}_{K}$ is relatively narrow if compared to what is available, for example, at EEX. The reason for this choice is that with a wider range, some options are far out of the money and become very cheap for parameter choices that result in low overall volatility. These options are not very interesting and pose also numerical challenges, hence we prefer to leave them out from our case study. Alternatively, to widen the range, one may for example consider a grid where the range of strike prices depends on the time to maturity, or consider options with high strikes only in combination with model parameters that lead to large overall volatility. However, since the network architecture would not actually change significantly, we leave this for future applications.

For $\Theta\subset \R^7$, we define a four-layer neural network $\mathcal{N}: \Theta \to \R_{+}^{7\times9}$ with number of nodes $\mathbf{n} = (30, 30, 30, 63)$ and the ReLU activation function. The final output is reshaped into a grid of dimension $7\times9$. The number of parameters to calibrate is thus $M = 4053$. We consider a training set of size $N_{train}= 40000$ and a test set of size $N_{test} = 4000$, which is also used for the calibration step. Moreover, $\theta = (a,b,k,\alpha_0, \alpha_1, \alpha_2, \alpha_3)\in\Theta$ with $\Theta$ defined by
\begin{equation*}
	\Theta:= [0.2, 0.5] \times [0.5, 0.8] \times [8.0, 9.0]\times [34.2, 34.7] \times [-1.5, -1.0] \times [0.2, 1.2] \times [4.5, 5.0].
\end{equation*}
Thus the annual implied spot price volatility is between $20\%$ and $50\%$, while the half-life for the process, that is the time taken for the price to revert half-way back to its long-term level, is approximately between $4.5$ and $7$ months, see \cite{clew} for details. In particular, the samples are generated by considering a uniform grid for each of the $n=7$ dimensions composing $\Theta$, and then randomly matching the values in each dimension. The corresponding prices have then been calculated with the formula \eqref{priceVexp}. We point out that each of the $N_{train}$ (or $N_{test}$) samples is a grid of size $63=7\times 9$ with each dimension corresponding to the different values for $K$'s and $\tau$'s. Each vector calibration is thus performed starting from a sample of $63$ prices, as defined in Section \ref{gridsection}.

In Figure \ref{step1_grid_dense} we report the average relative error and the maximum relative error in the training step, both for the training and test set. In particular, the errors have been clustered in order to obtain a grid corresponding to the different contracts $(\tau, K) \in \Lambda^{grid}_{\tau} \times \Lambda^{grid}_{K}$. We notice that the average relative error is quite low, showing a good performance of the neural network in approximating the price functional. The worst accuracy is for the contracts with higher strike price, probably due to the fact that prices for these contracts are small.

In Figure \ref{adam_grid_dense} we report the relative error for the components of $\hat{\theta}$ in the calibration step. For some of the parameters, such as $a$ and $\alpha_2$ we notice a certain pattern, namely, for higher values of the parameter the error is smaller and vice versa. However, the performance in calibration is not particularly good: $a$, $b$, and $\alpha_2$ have a mean relative error which is more than $20\%$. On the other hand, even if the model parameters are not accurately estimated, by substituting $\hat{\theta}$ in the neural network, we get a price approximation $\hat{\Pi}$ which is quite good. This can be seen in Figure \ref{step2_adam_grid_dense}, where we report the average and the maximum relative error after calibration. For mid-maturity contracts we observe the best accuracy.

\begin{figure}[!htbp]
	\setlength{\tabcolsep}{10pt}
	\resizebox{1\textwidth}{!}{
		\begin{tabular}{@{}>{\centering}m{0.5\textwidth}>{\centering\arraybackslash}m{0.5\textwidth}@{}}
			\multicolumn{2}{c}{\large \textbf{Training set}}\\
			\textbf{Average relative error} & \textbf{Maximum relative error}\\
			\includegraphics[width=0.5\textwidth]{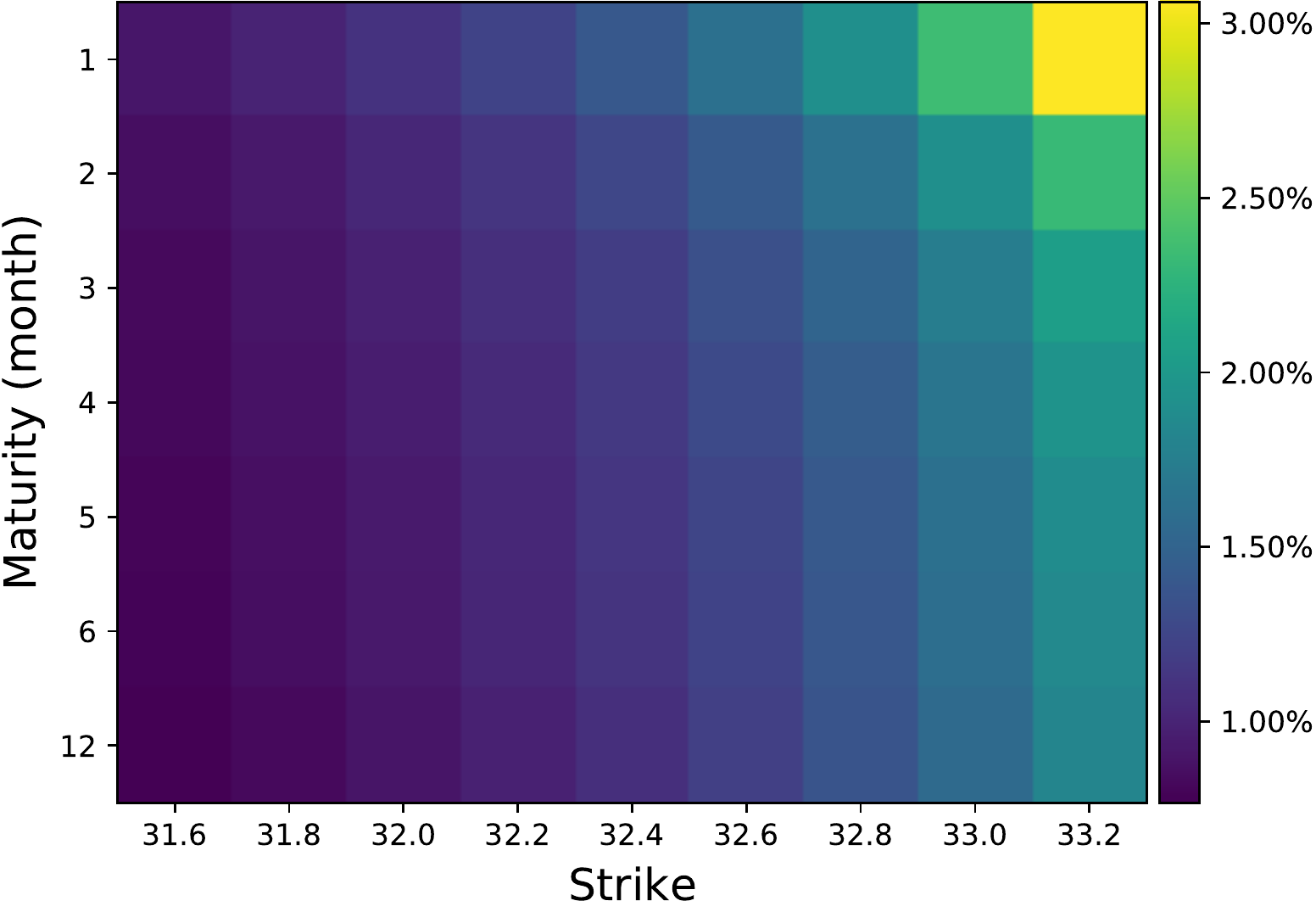} 
			&\includegraphics[width=0.5\textwidth]{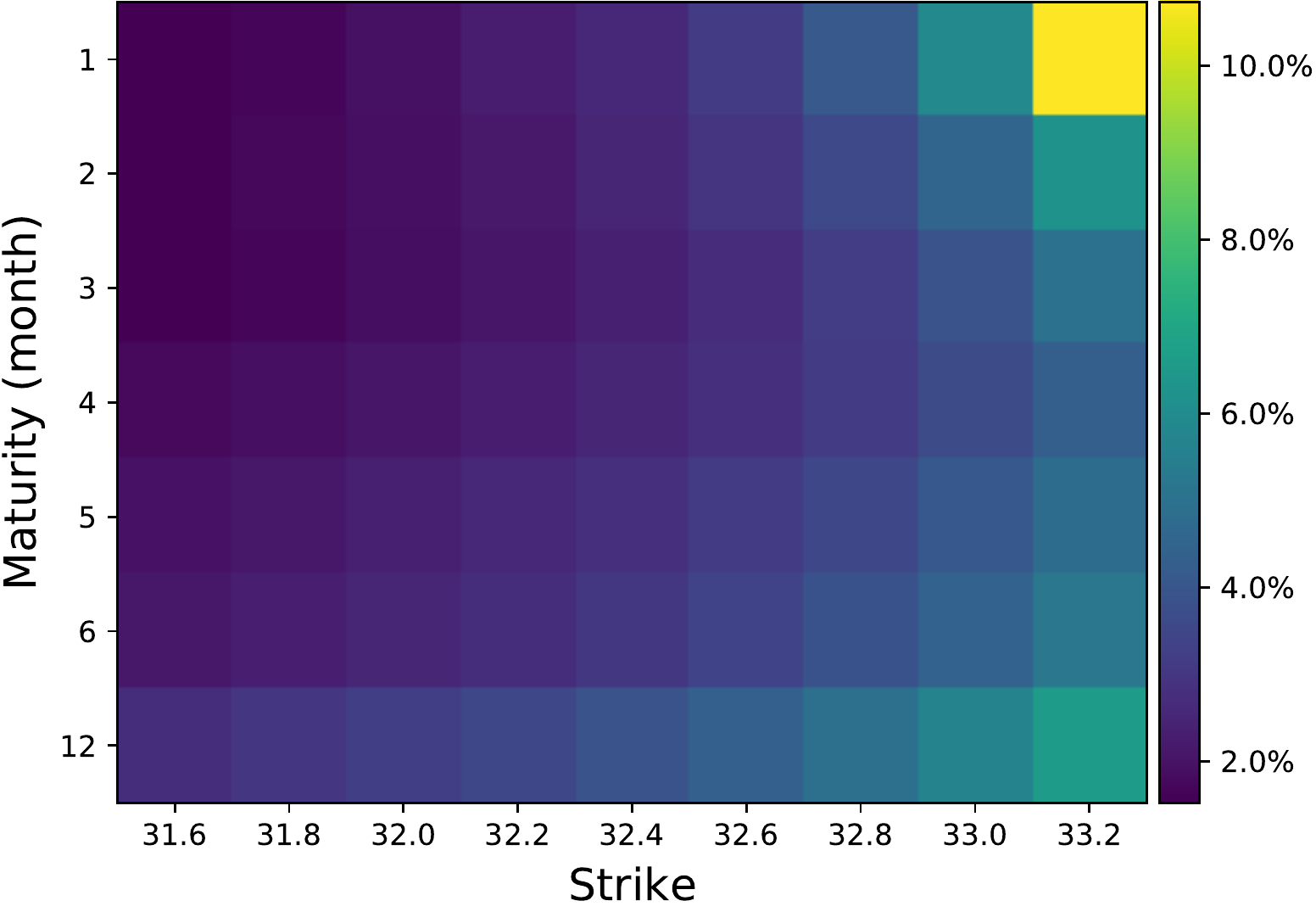}\\
			\multicolumn{2}{c}{\large \textbf{Test set}}\\
			\textbf{Average relative error} & \textbf{Maximum relative error}\\
			\includegraphics[width=0.5\textwidth]{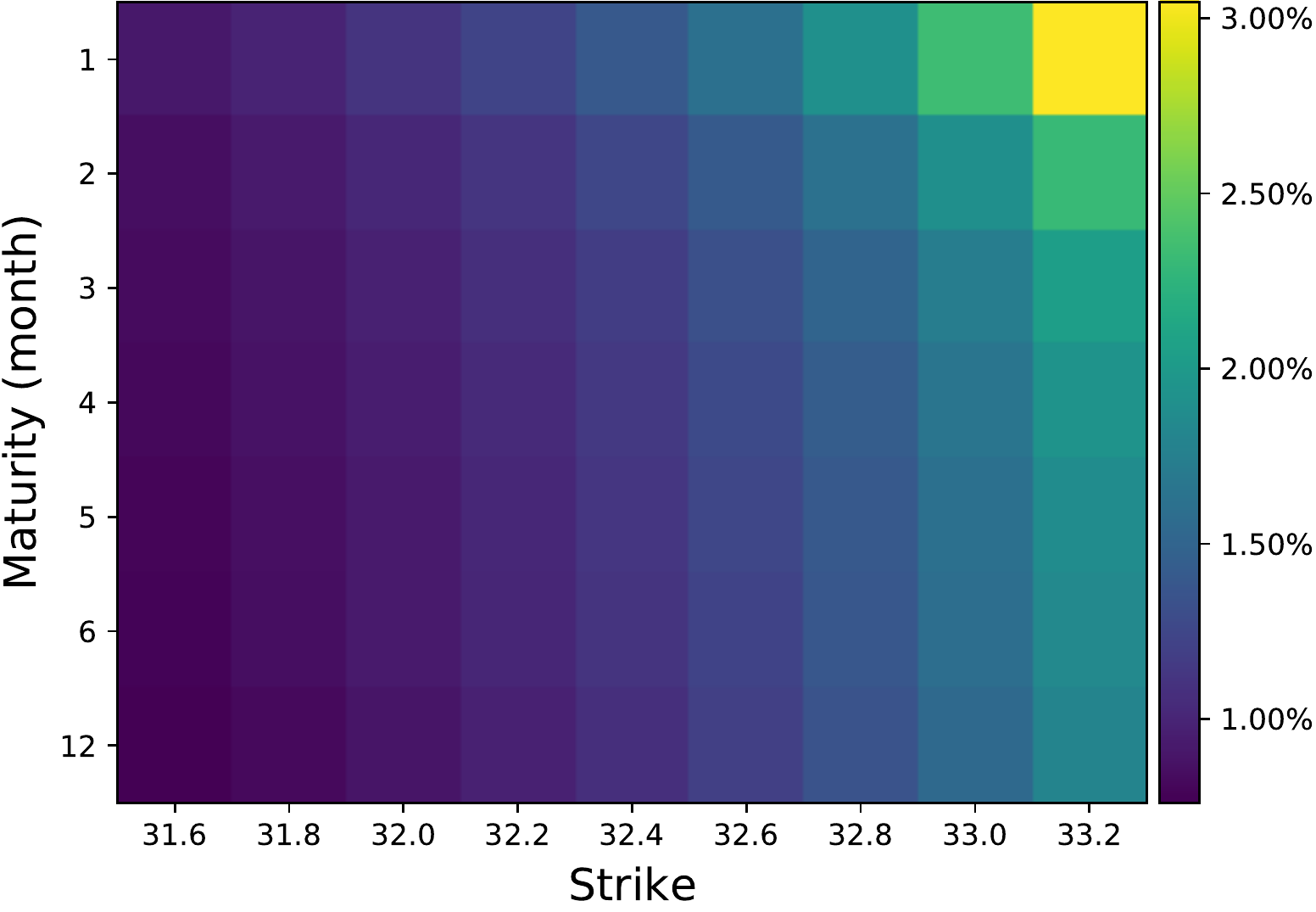} 
			&\includegraphics[width=0.5\textwidth]{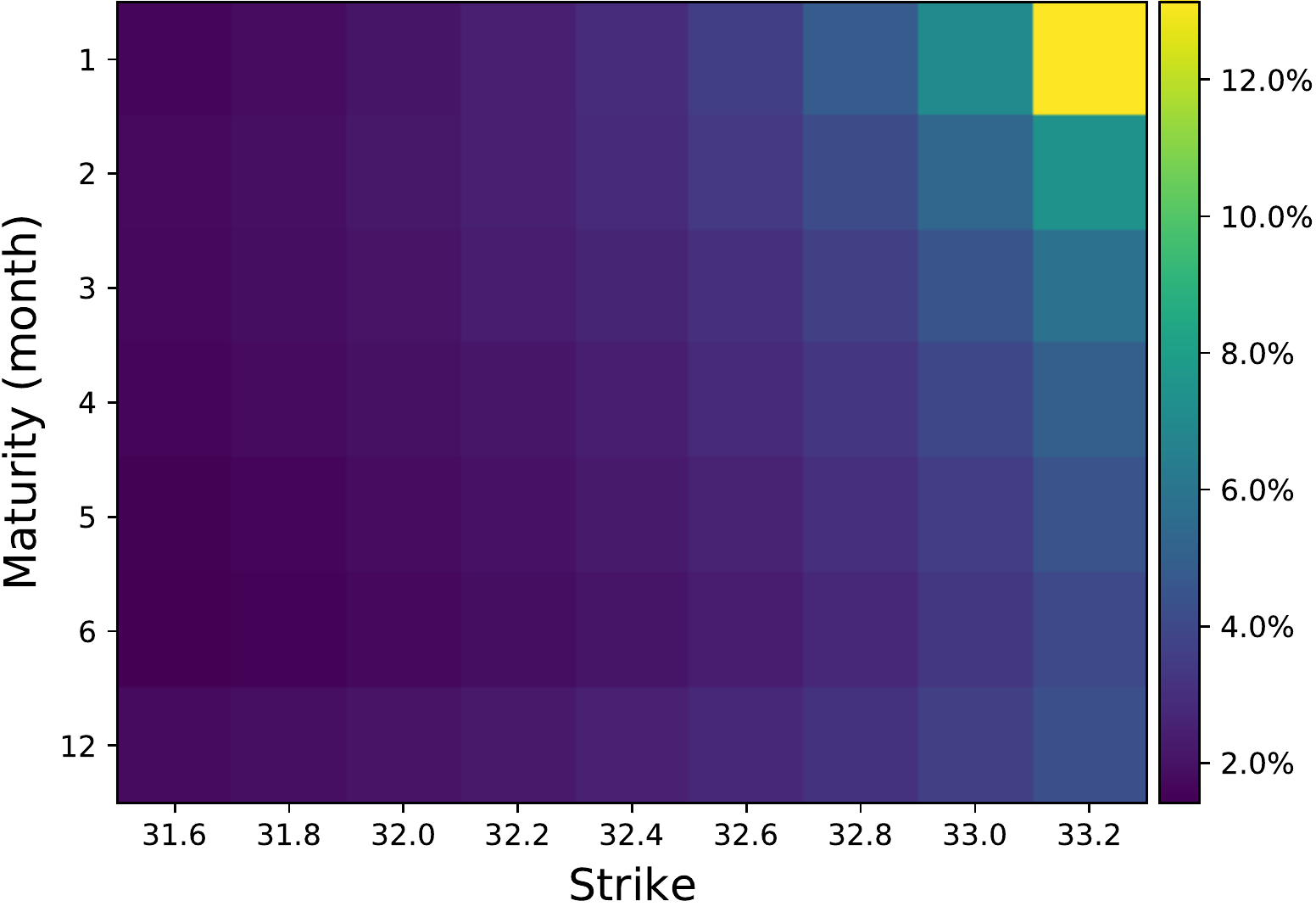}
		\end{tabular}
	}
	\caption{Average relative error and Maximum relative error in the approximation step with the grid-based learning approach.\label{step1_grid_dense}}
\end{figure}

\begin{figure}[!htbp]
	\setlength{\tabcolsep}{5pt}
	\resizebox{1\textwidth}{!}{
		\begin{tabular}{@{}>{\centering}m{0.33\textwidth}>{\centering}m{0.33\textwidth}>{\centering\arraybackslash}m{0.33\textwidth}@{}}
			$\boldsymbol{a}$ & $\boldsymbol{b}$&$\boldsymbol{k}$\\
			\includegraphics[width=0.33\textwidth]{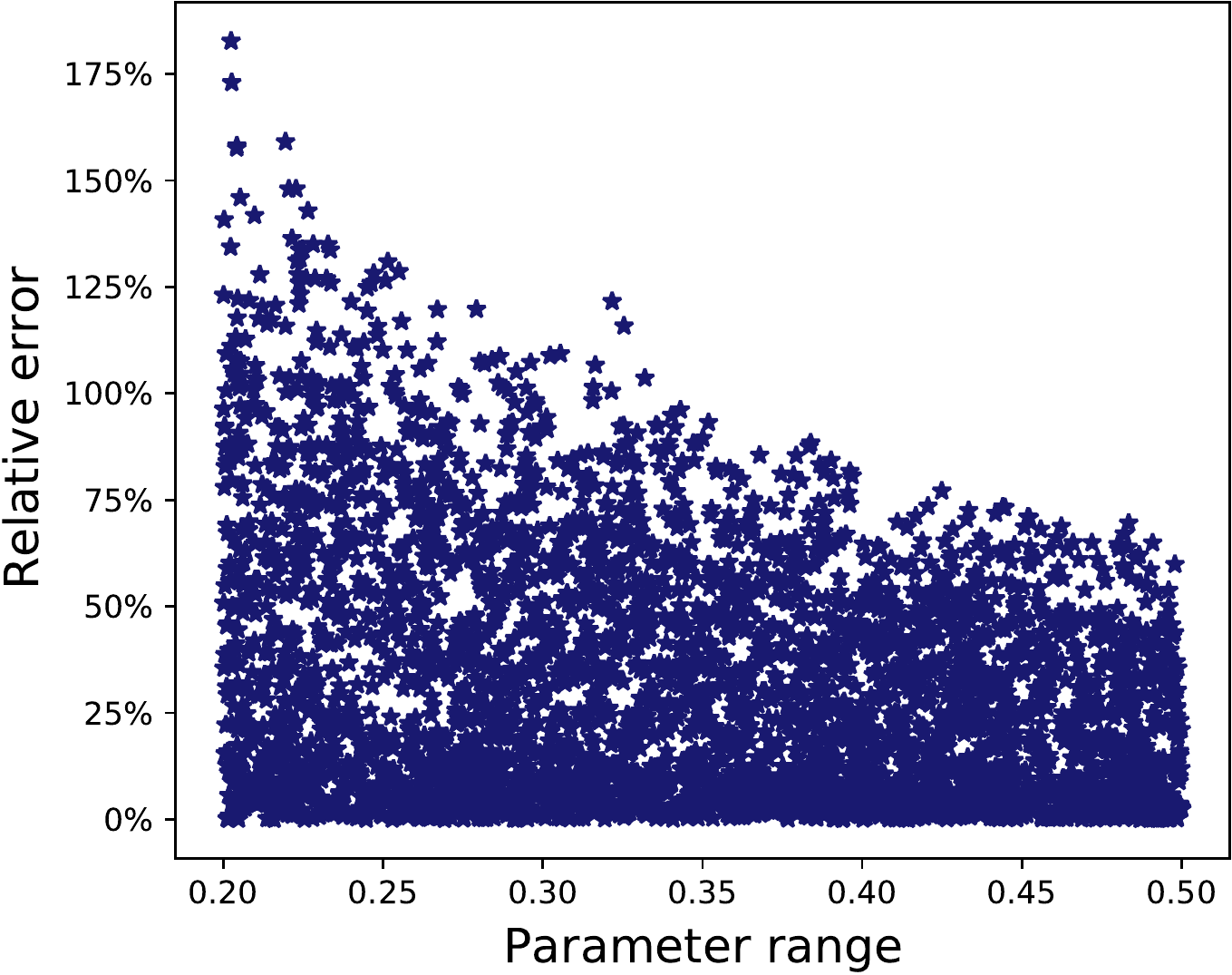} 
			&\includegraphics[width=0.33\textwidth]{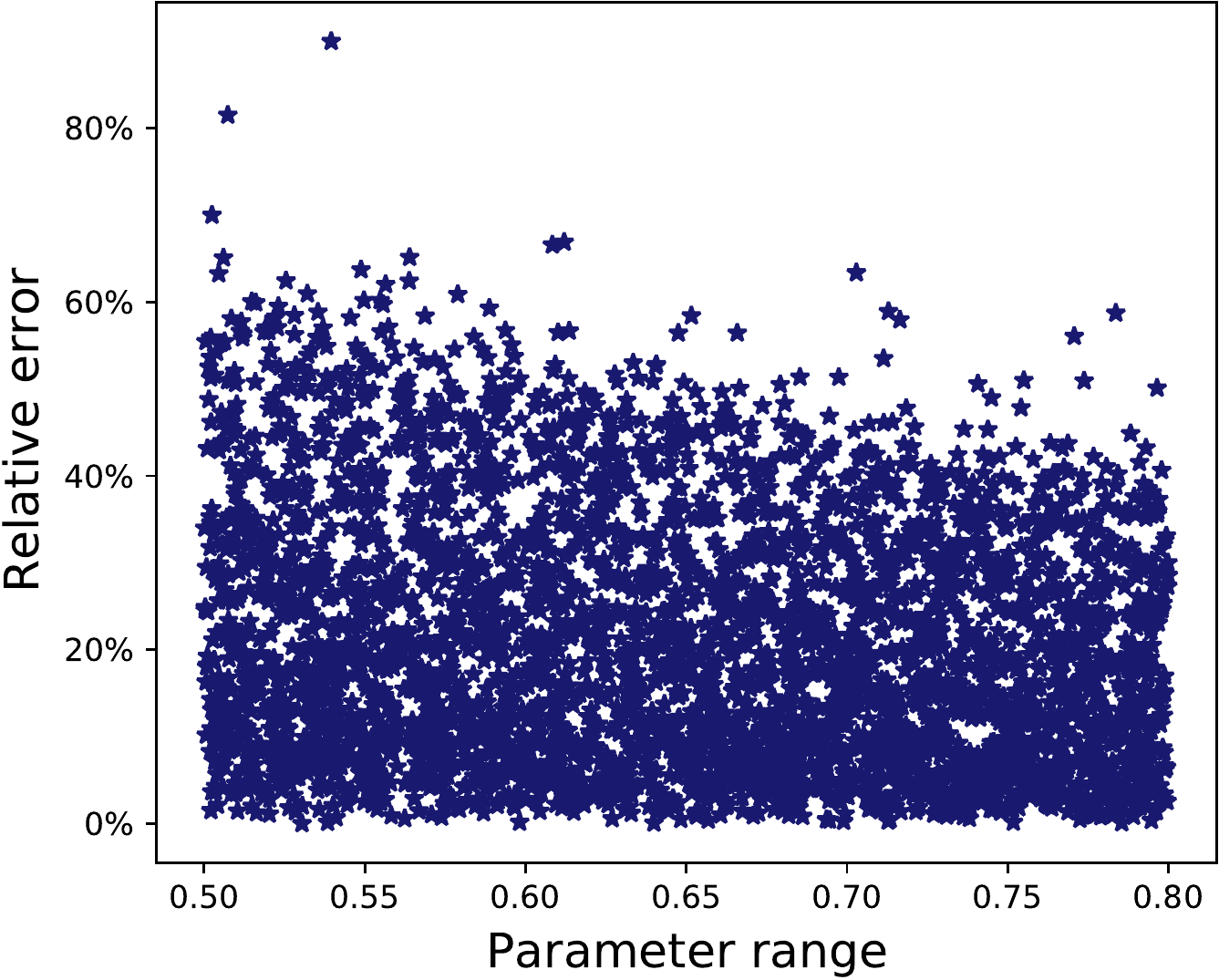}&\includegraphics[width=0.33\textwidth]{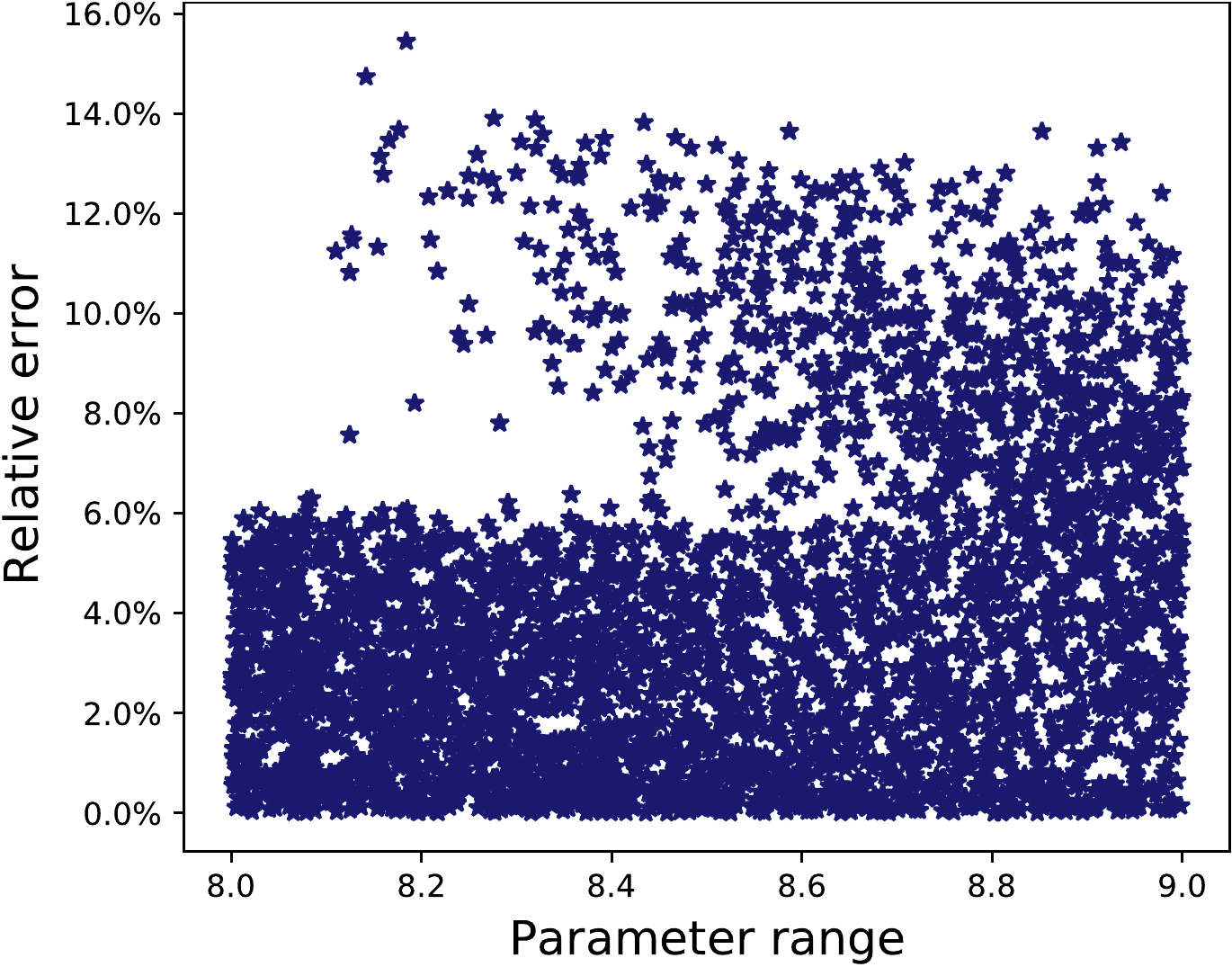}\\
			$\boldsymbol{\alpha_0}$ & $\boldsymbol{\alpha_1}$ &$\boldsymbol{\alpha_2}$\\
			\includegraphics[width=0.33\textwidth]{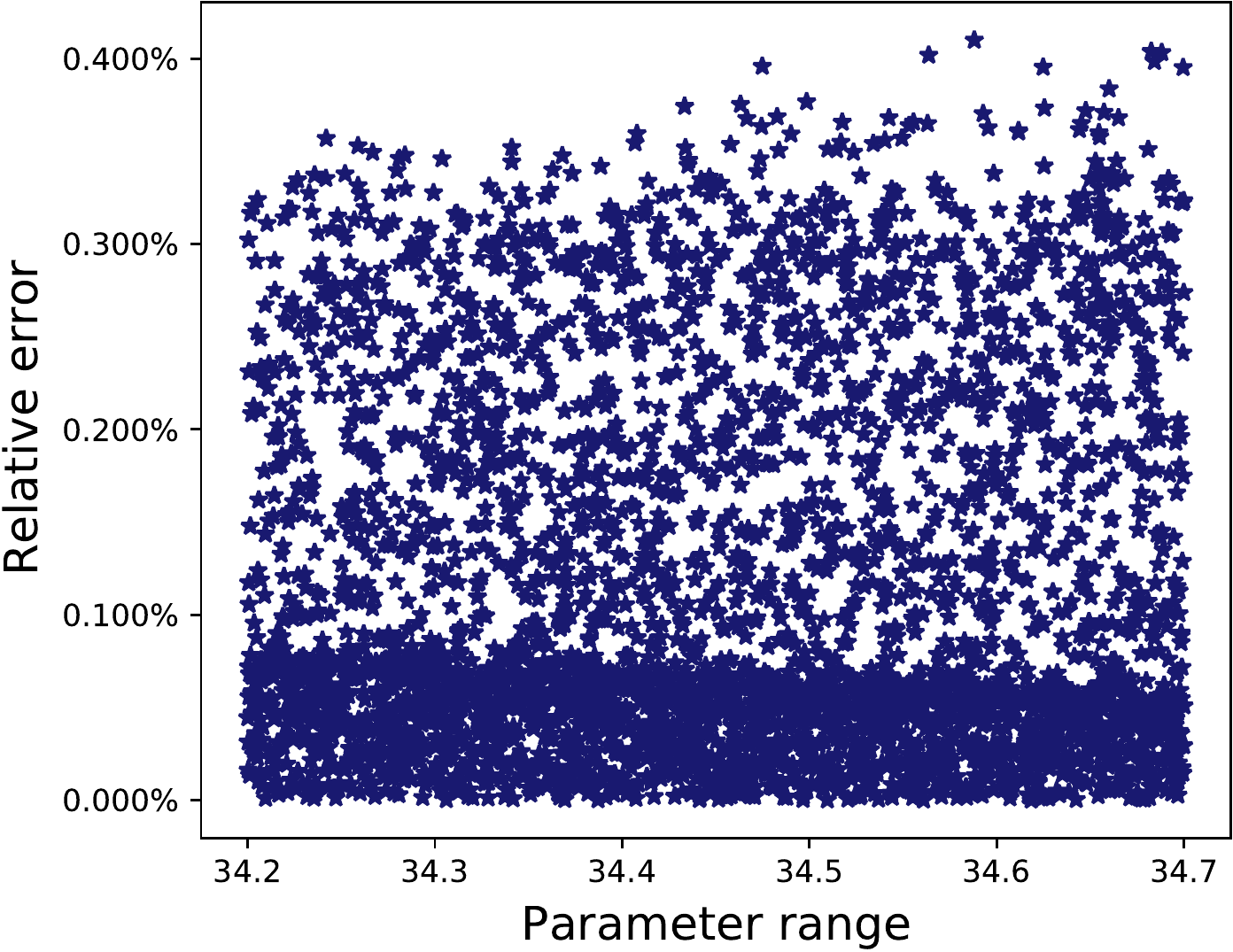} 
			&\includegraphics[width=0.33\textwidth]{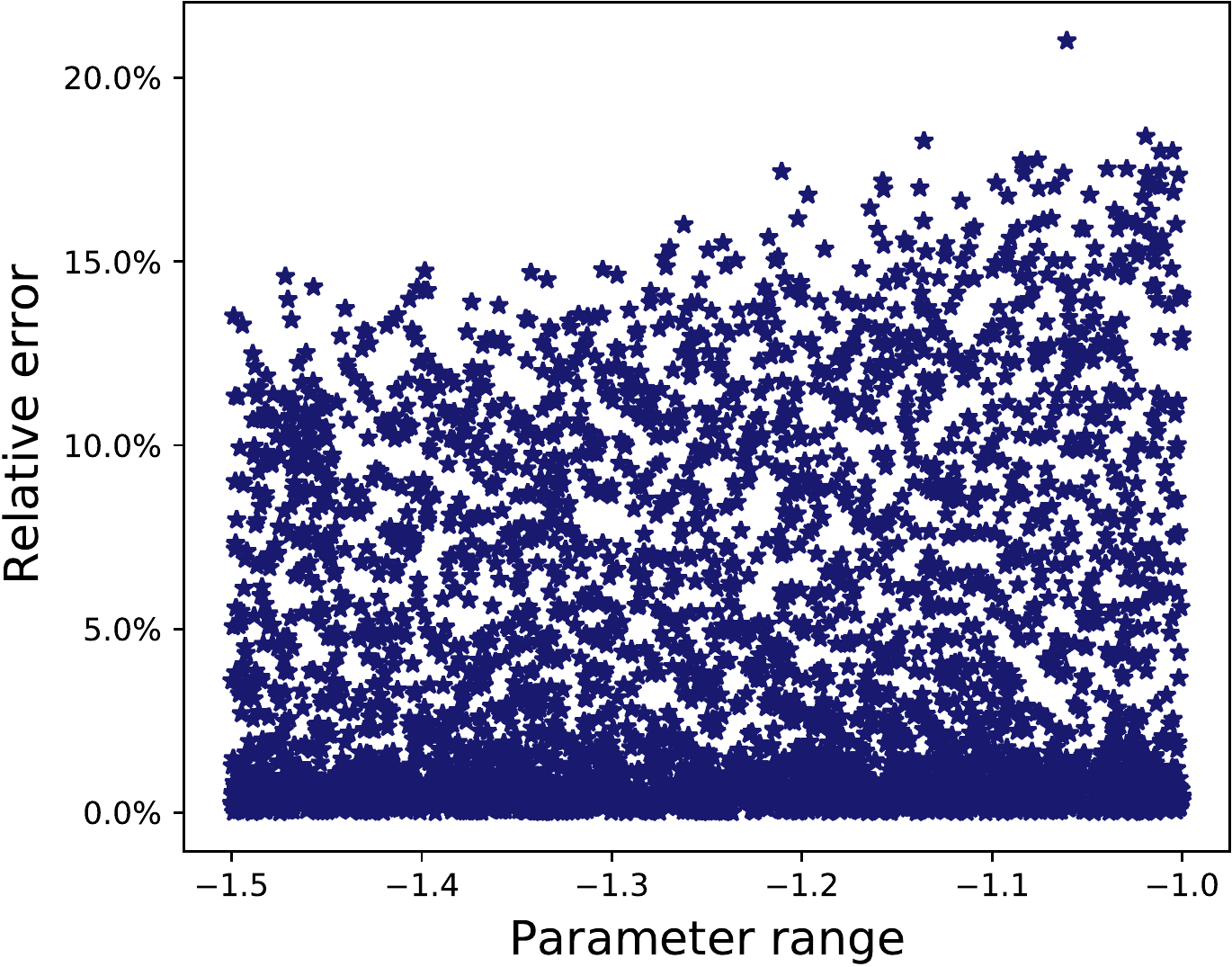}&\includegraphics[width=0.33\textwidth]{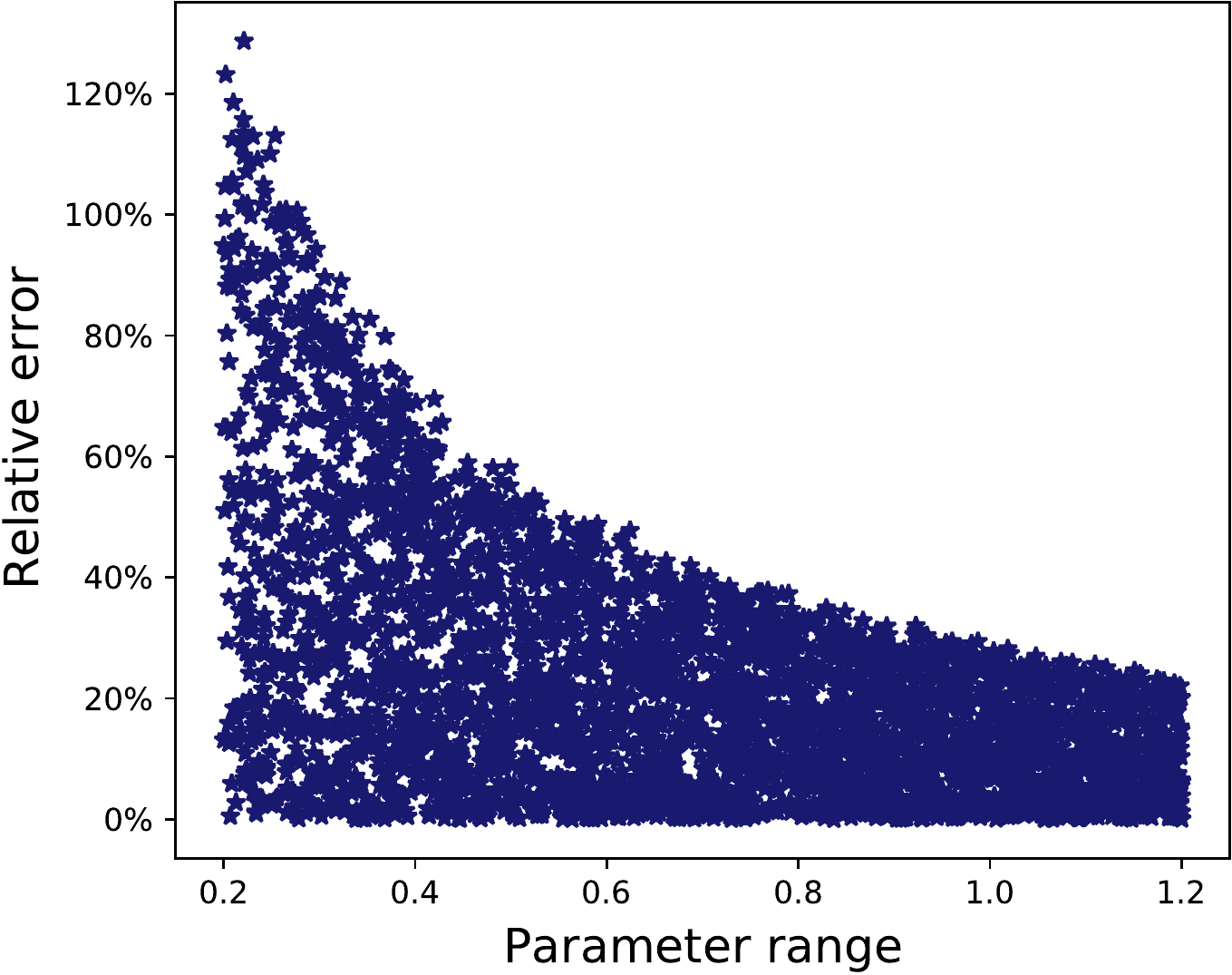}\\
			$\boldsymbol{\alpha_3}$ &  &\\
			\includegraphics[width=0.33\textwidth]{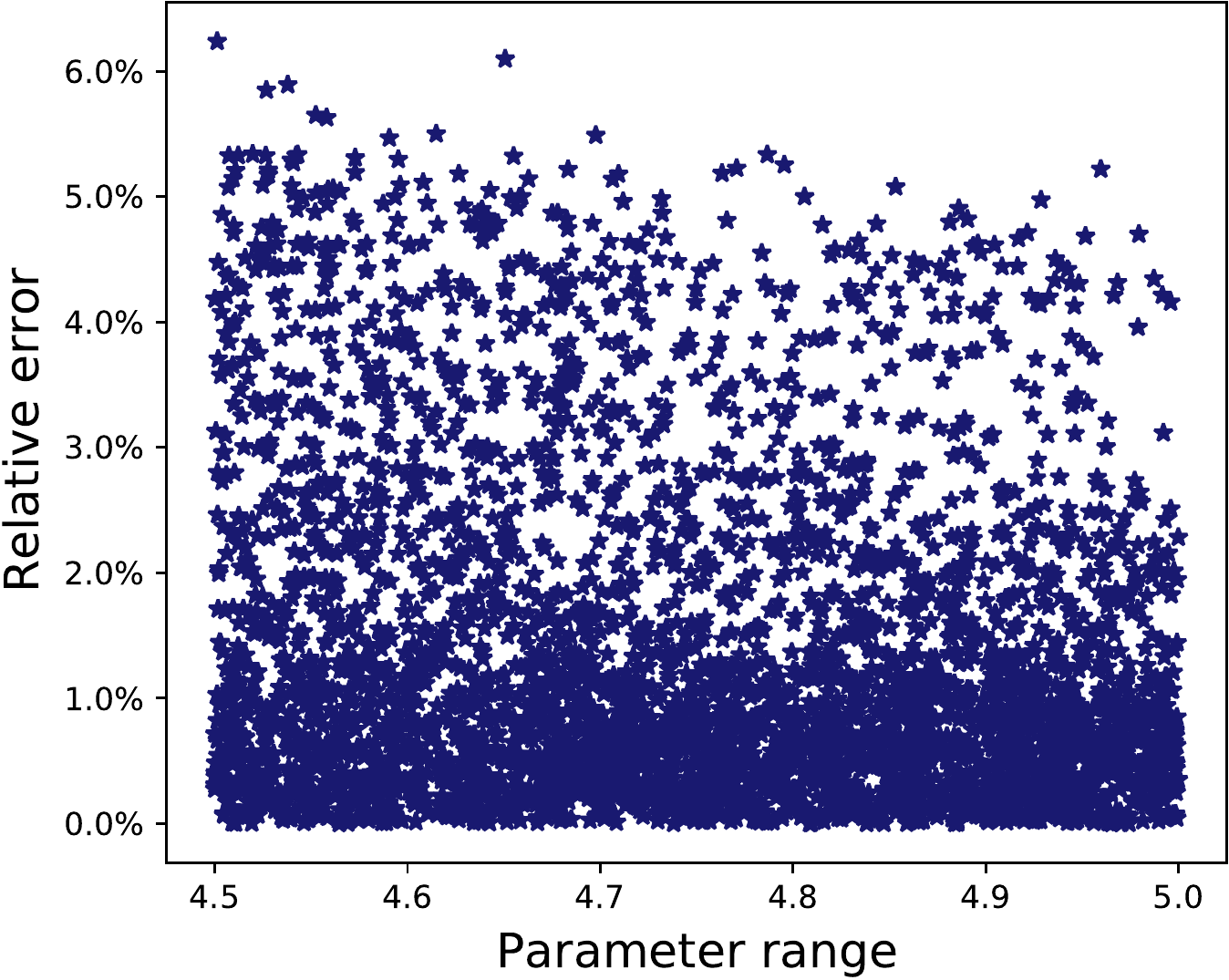} 
			& \vspace{-0.5cm}\begin{tabular}{rrr}
				& \textbf{Mean} & \textbf{Median}\\[3pt]
				$\boldsymbol{a}$: &$32.4\%$ & $24.6\%$\\
				$\boldsymbol{b}$: &$20.7\%$ & $17.9\%$\\
				$\boldsymbol{k}$: &$3.94\%$ & $3.24\%$\\
				$\boldsymbol{\alpha_0}$: &$0.12\%$ & $ 0.07 \%$\\
				$\boldsymbol{\alpha_1}$: &$ 4.29\%$ & $1.91\%$\\
				$\boldsymbol{\alpha_2}$: &$22.3\%$ & $17.2\%$\\
				$\boldsymbol{\alpha_3}$: &$1.34\%$ & $0.86\%$\\
			\end{tabular} & 
		\end{tabular}
	}
	\caption{Relative error for the model parameters calibrated with the grid-based learning approach.\label{adam_grid_dense}}
\end{figure}

\begin{figure}[!htbp]
	\setlength{\tabcolsep}{10pt}
	\resizebox{1\textwidth}{!}{
		\begin{tabular}{@{}>{\centering}m{0.5\textwidth}>{\centering\arraybackslash}m{0.5\textwidth}@{}}
			\textbf{Average relative error} & \textbf{Maximum relative error}\\
			\includegraphics[width=0.5\textwidth]{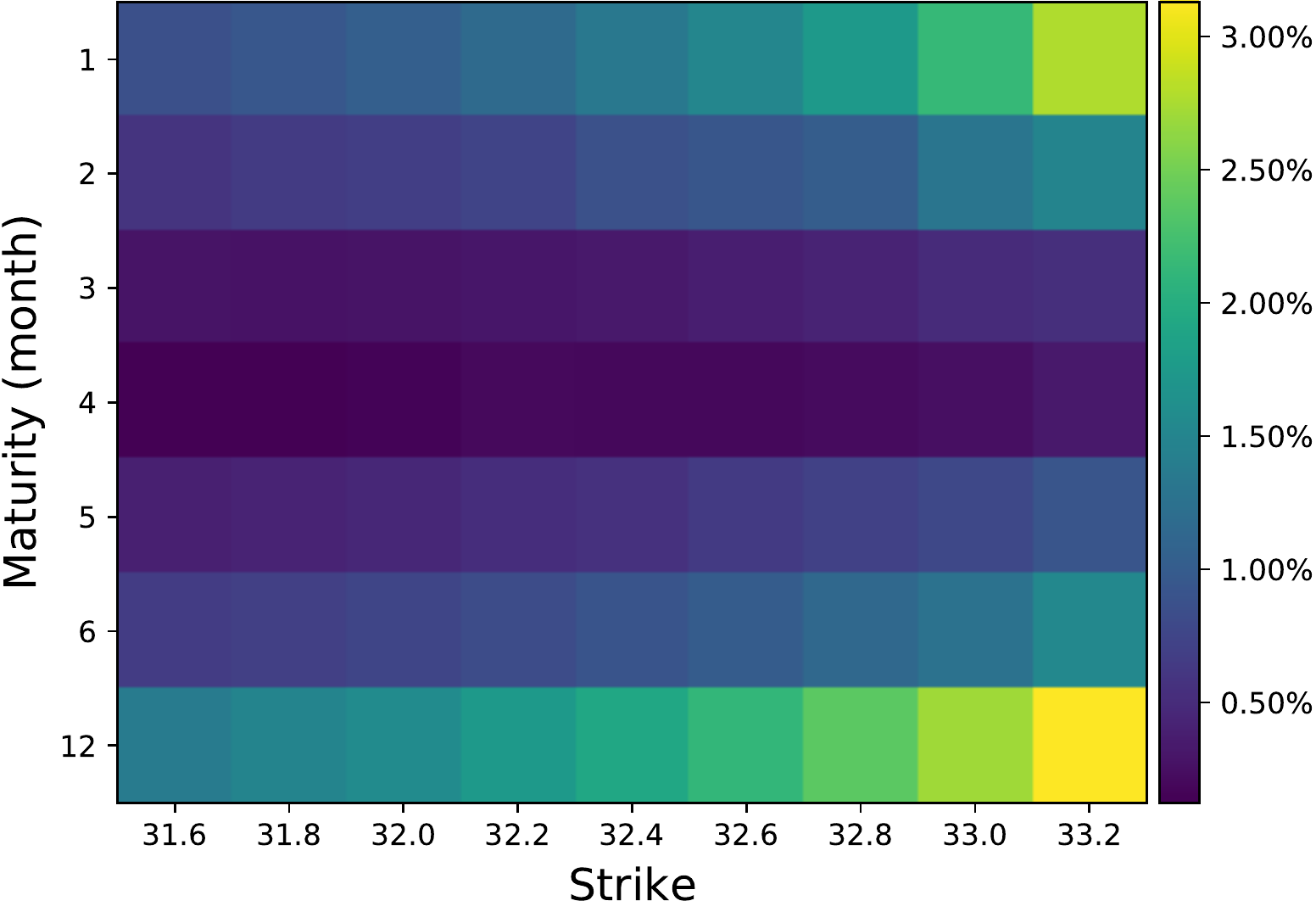} 
			&\includegraphics[width=0.5\textwidth]{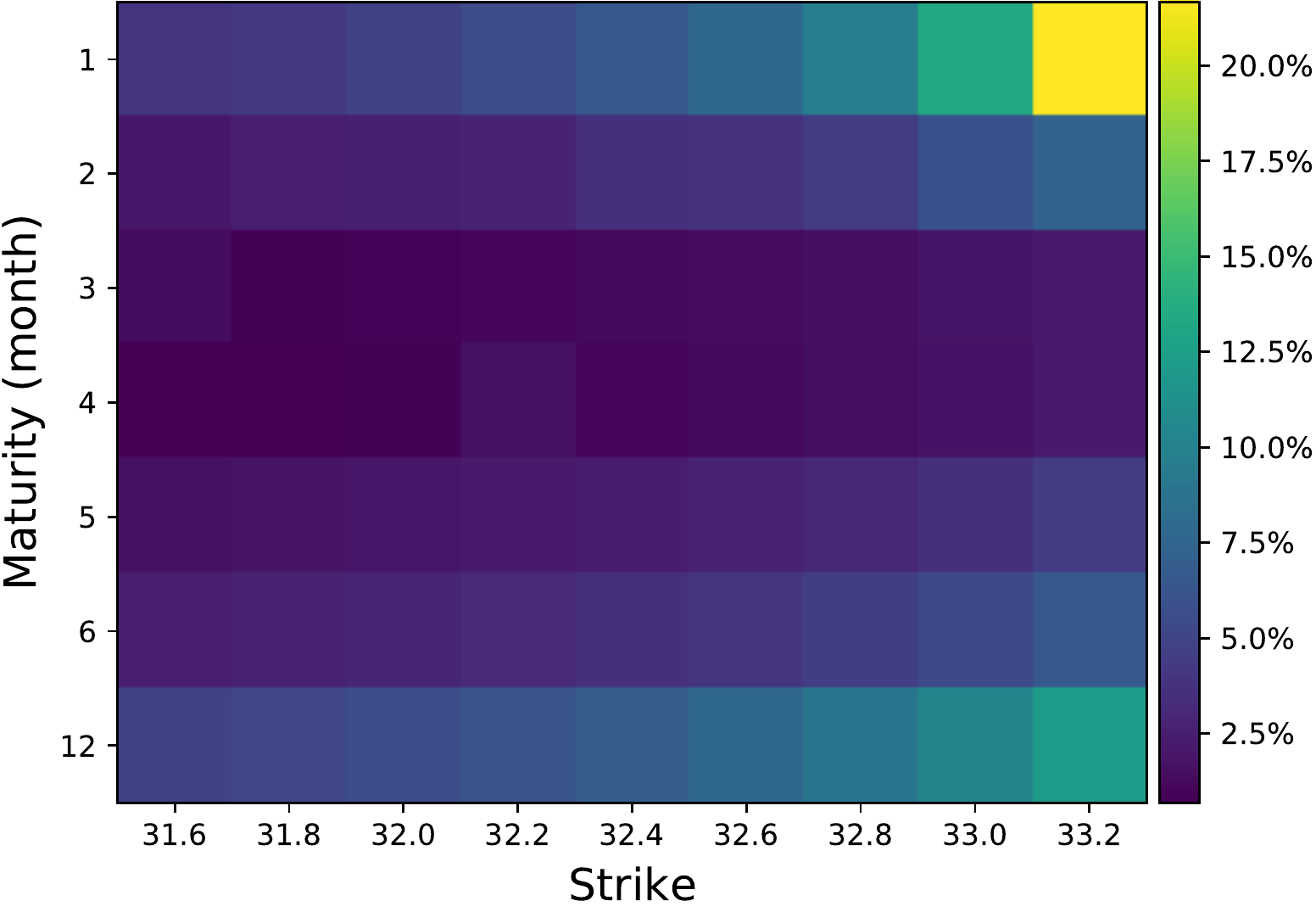}
		\end{tabular}
	}
	\caption{Average relative error and Maximum relative error after calibration with the grid-based learning approach.\label{step2_adam_grid_dense}}
\end{figure}

\subsection{Pointwise learning}
For $\Lambda\subset \R^2$ and $\Theta\subset \R^7$, we define a four-layer neural network $\mathcal{N}: \Lambda \times \Theta \to \R_{+}$ with nodes size $\mathbf{n}= (30, 30, 30, 1)$ and the ELU activation function. The number of parameters to calibrate is $M = 2191$. We consider a training set of size $N_{train}= 60000$ and a test set of size $N_{test} = 6000$. These are larger than the sets considered in the grid-based approach to reflect that now each sample corresponds to a single option value. We take $\theta = (a,b,k,\alpha_0, \alpha_1, \alpha_2, \alpha_3)\in\Theta$ with $\Theta$ as above, while $\lambda = (\tau, K)\in\Lambda^{point}= \Lambda^{point}_{\tau}\times \Lambda^{point}_K= [1/12, 1]\times[31.6, 33.2]$. Training and test sets have been generated as described in the grid-based approach.

We notice that $\Lambda^{grid}\subset \Lambda^{point}$: this allows us to compare the pointwise learning approach with the grid-based learning approach. To do that, we cluster the prices in the following way. We introduce a new grid $\hat{\Lambda} := \hat{\Lambda}_{\tau} \times \hat{\Lambda}_K$, where
\begin{align*}
	&\hat{\Lambda}_{\tau} :=  \{1/12, 1/12+1/24, 2/12+1/24, 3/12+1/24,4/12+1/24, 5/12+1/24, 6/12+1/24, 1\}\\
	& \hat{\Lambda}_K := \{31.6, 31.7, 31.9, 32.1, 32.3, 32.5, 32.7, 32.9, 33.1, 33.2\}
\end{align*}
and we label with $K=31.6$ the prices corresponding to a strike value in the interval $[31.6, 31.7]$, while we label with $K=31.8$ the prices corresponding to a strike in the interval $[31.7, 31.9]$, and so on, for each of the strike prices in $\Lambda^{grid}_K$. We do the same for the time to maturity $\tau$ and obtain a grid of clustered values corresponding to $\Lambda^{grid}$. Since each of the $N_{test} = 6000$ samples corresponds to a different parameter vector, we can not perform calibration starting from the test set. Instead, we use the test set from the grid-based approach, after shape adjustment. In this manner, each calibration is performed as described in Section \ref{pointwisesection}, starting from $N_{cal}=63$ different prices, corresponding to the same $\theta \in \Theta$ to be estimated.

In Figure \ref{step1_pointwise} we report the average relative error and the maximum relative error in the approximation step after clustering, as described above. For both we notice a better accuracy compared to the grid-based approach. As before, the worst accuracy is for contracts with the highest strike price. However, when it comes to calibration, the pointwise approach turns out to be worse. Indeed, the relative errors for $\hat{\theta}$ in Figure \ref{adam_pointwise} are on average worse than in the previous experiments. It can be noticed that for $a$, $b$ and $k$, the error is concentrated around $45\%$, $30\%$ and $5\%$ respectively, while before it was much more spread out starting from $0\%$. The accuracy for $a$ is not very good, almost reaching $50\%$, and the accuracy for $b$ is worse than before. On the other hand, the accuracy for $\alpha_2$ has substantially improved. In Figure \ref{step2_adam_pointwise} we can see that the relative error for $\hat{\Pi}$ after calibration is also worse than for the grid-based method.

\begin{figure}[!htbp]
	\setlength{\tabcolsep}{10pt}
	\resizebox{1\textwidth}{!}{
		\begin{tabular}{@{}>{\centering}m{0.5\textwidth}>{\centering\arraybackslash}m{0.5\textwidth}@{}}
			\multicolumn{2}{c}{\large \textbf{Training set}}\\
			\textbf{Average relative error} & \textbf{Maximum relative error}\\
			\includegraphics[width=0.5\textwidth]{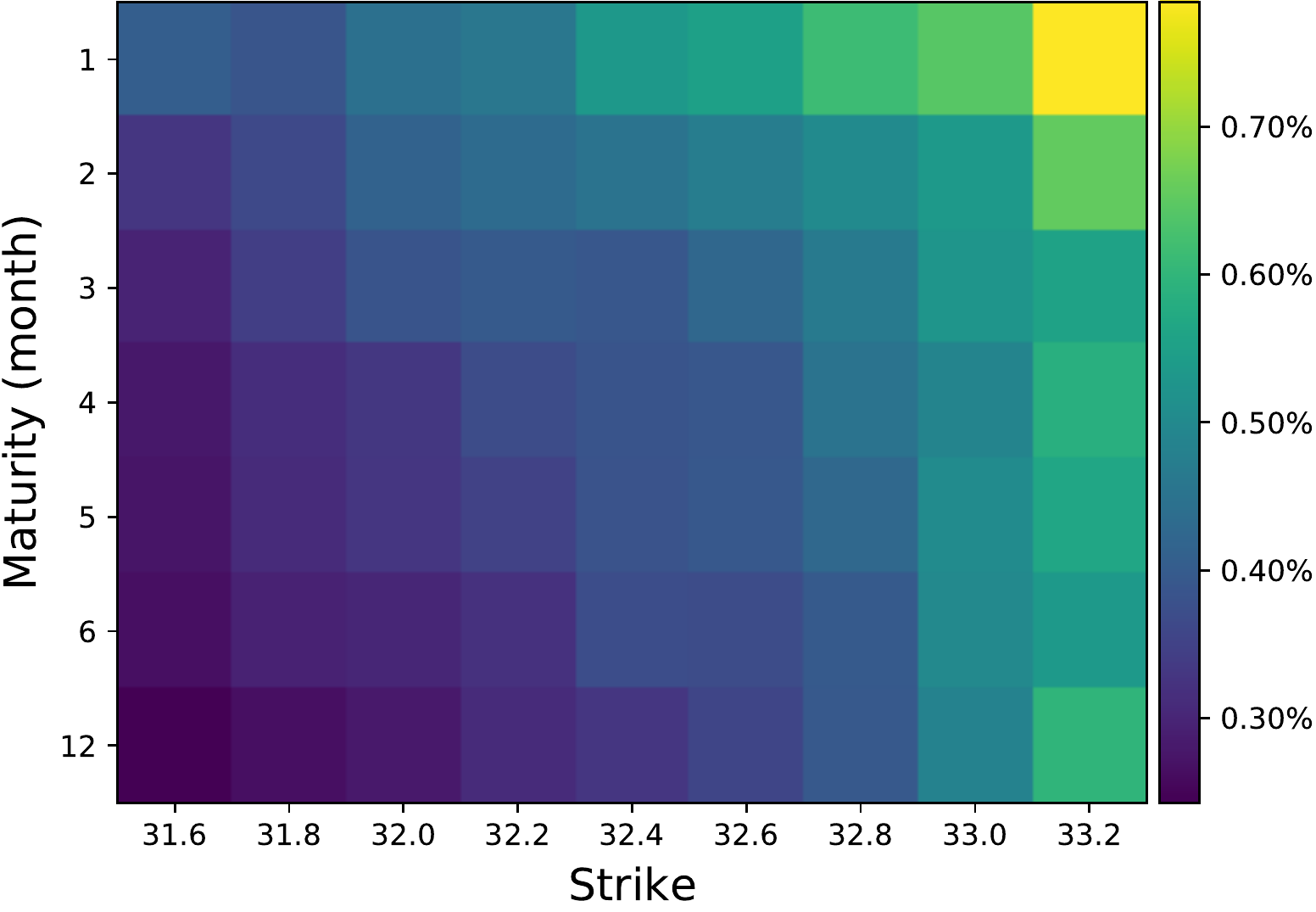} 
			&\includegraphics[width=0.5\textwidth]{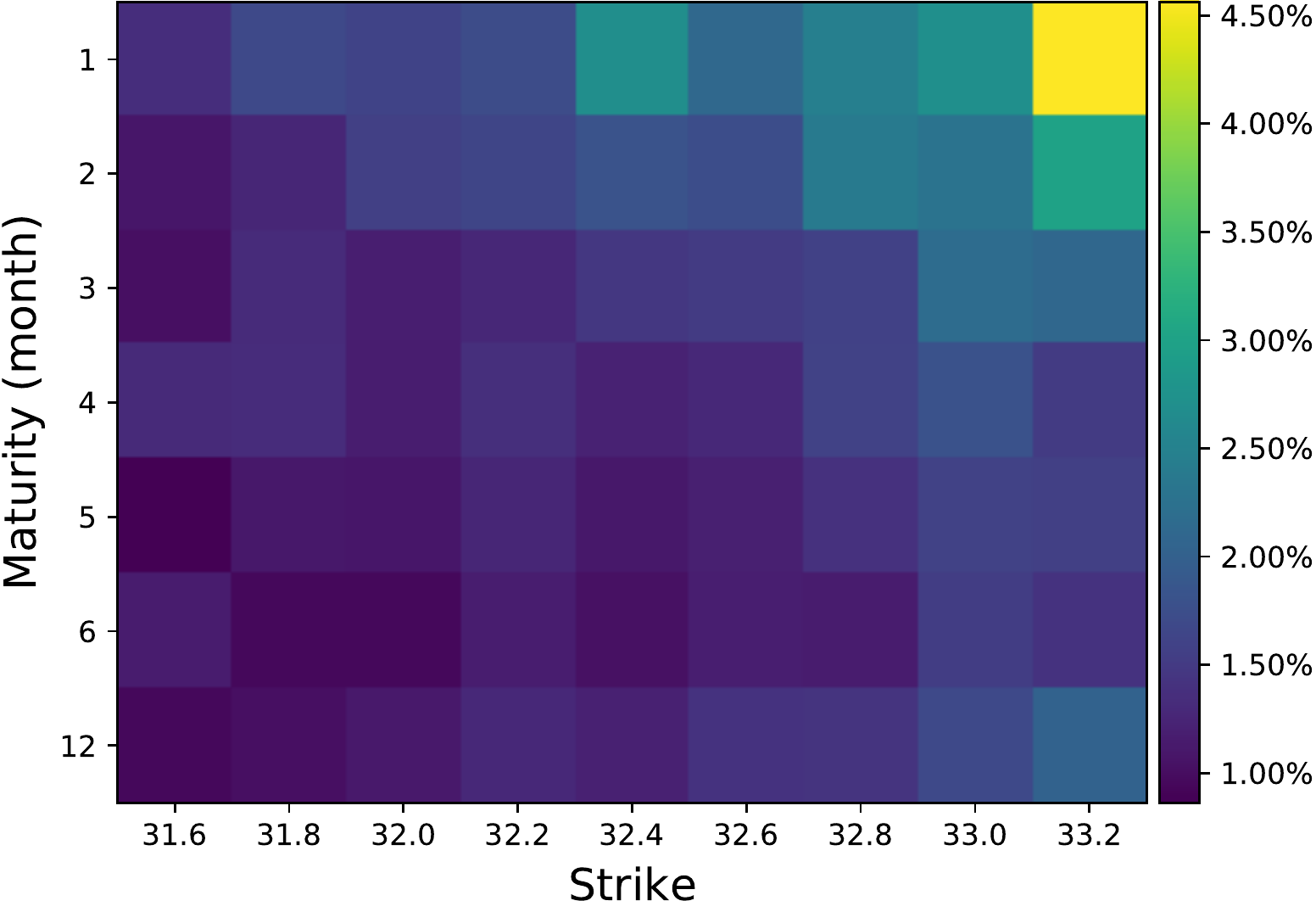}\\
			\multicolumn{2}{c}{\large \textbf{Test set}}\\
			\textbf{Average relative error} & \textbf{Maximum relative error}\\
			\includegraphics[width=0.5\textwidth]{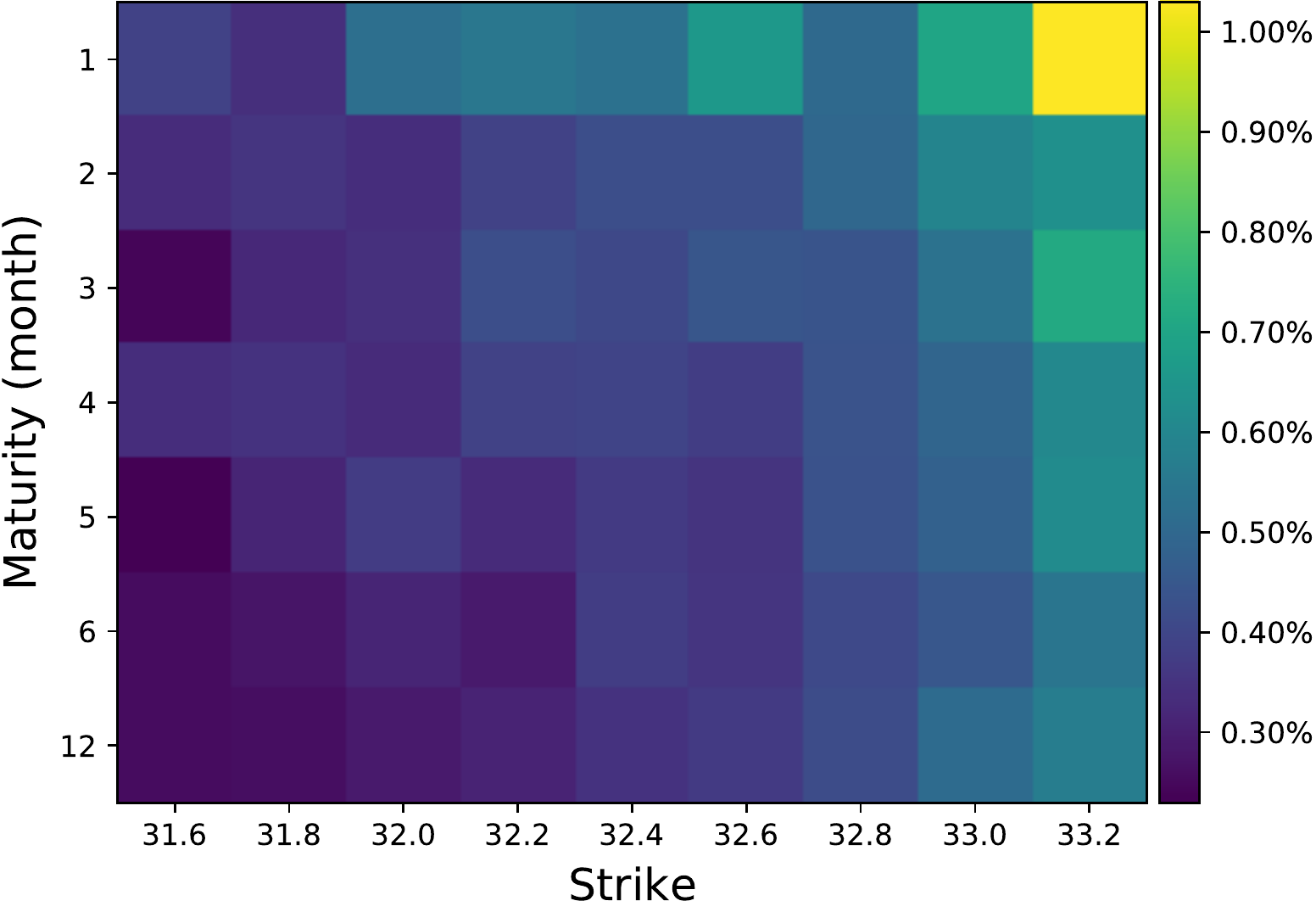} 
			&\includegraphics[width=0.5\textwidth]{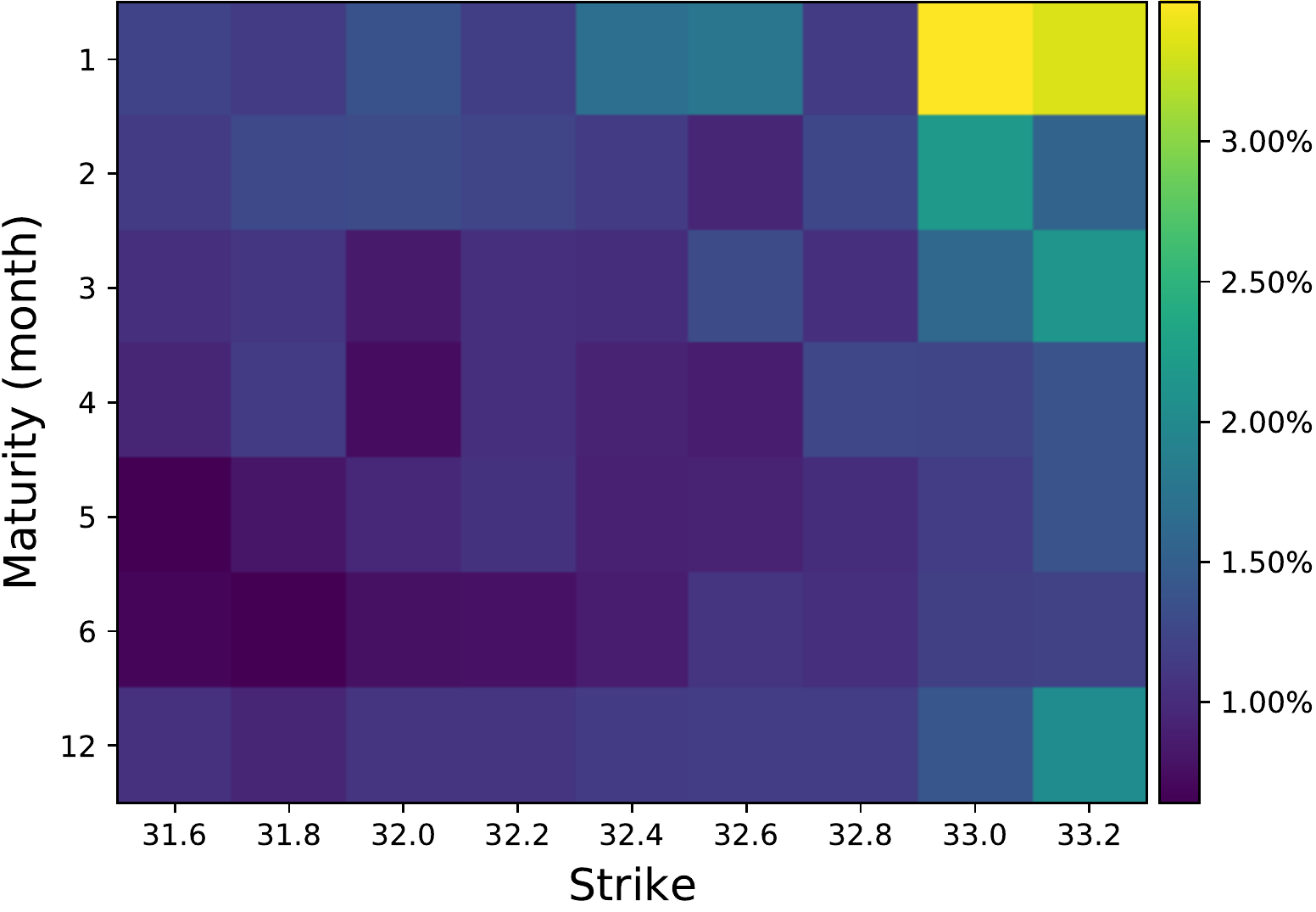}
		\end{tabular}
	}
	\caption{Average relative error and Maximum relative error in the approximation step with the pointwise learning approach.\label{step1_pointwise}}
\end{figure}

\begin{figure}[!htbp]
	\setlength{\tabcolsep}{5pt}
	\resizebox{1\textwidth}{!}{
		\begin{tabular}{@{}>{\centering}m{0.33\textwidth}>{\centering}m{0.33\textwidth}>{\centering\arraybackslash}m{0.33\textwidth}@{}}
			$\boldsymbol{a}$ & $\boldsymbol{b}$&$\boldsymbol{k}$\\
			\includegraphics[width=0.33\textwidth]{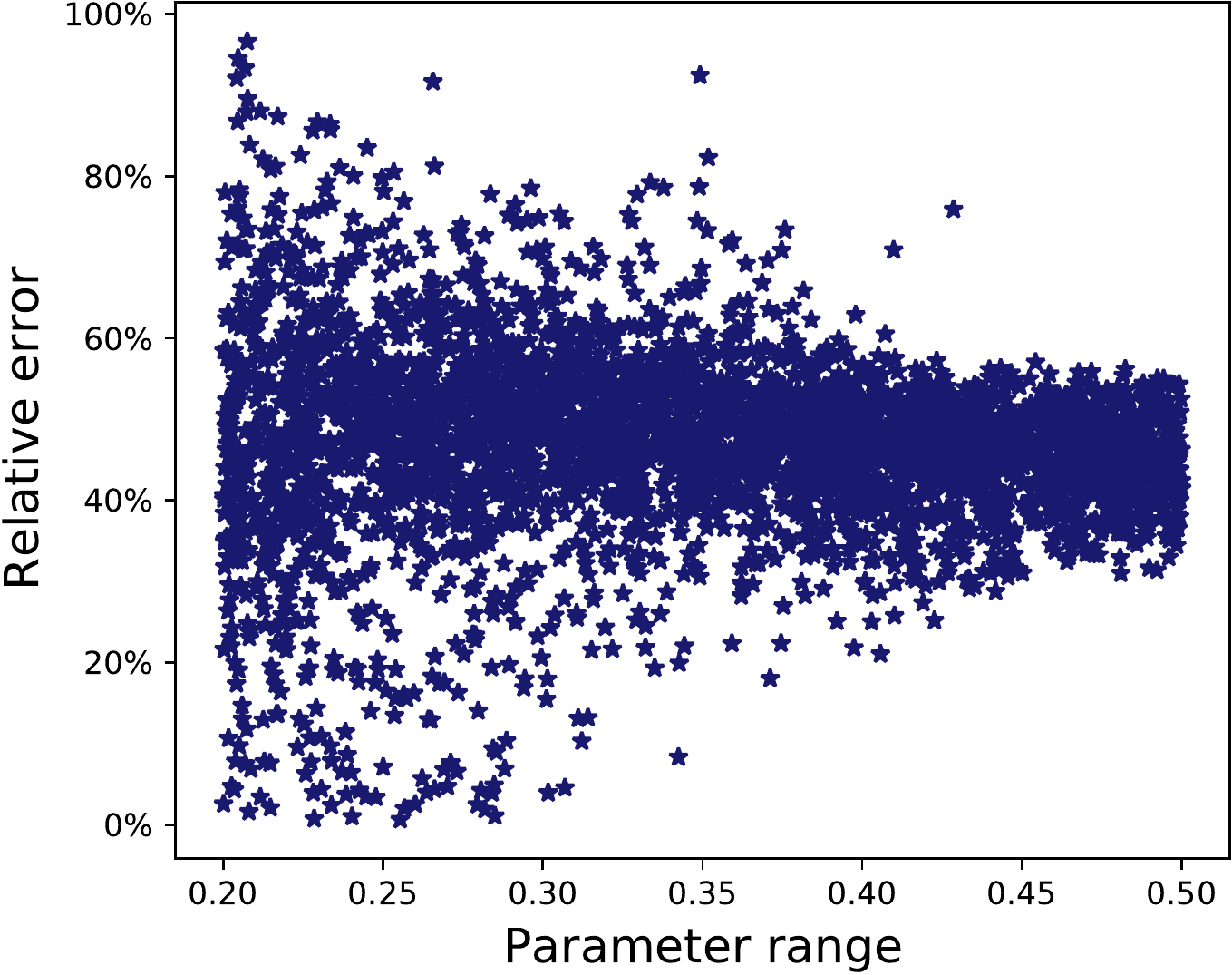} 
			&\includegraphics[width=0.33\textwidth]{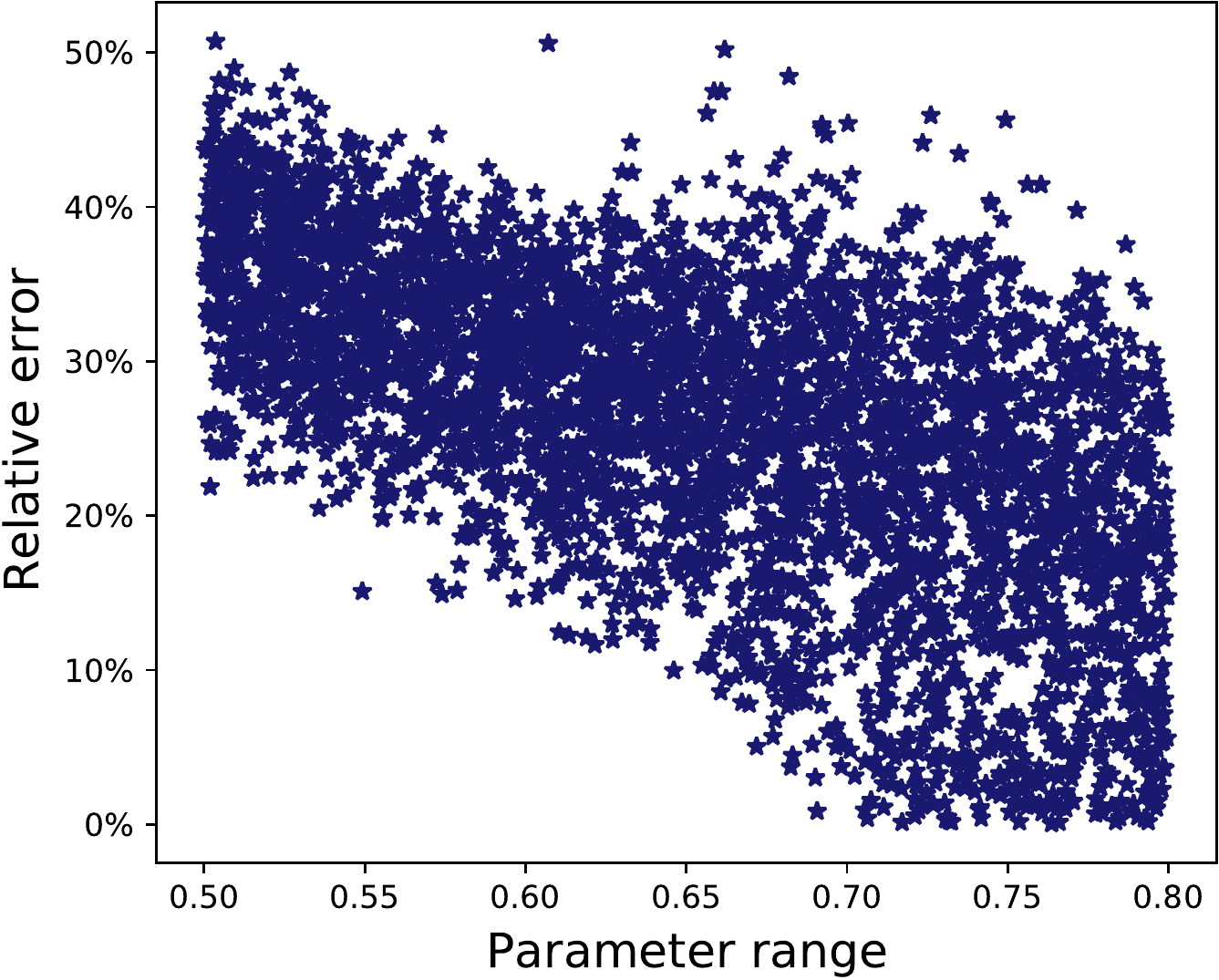}&\includegraphics[width=0.33\textwidth]{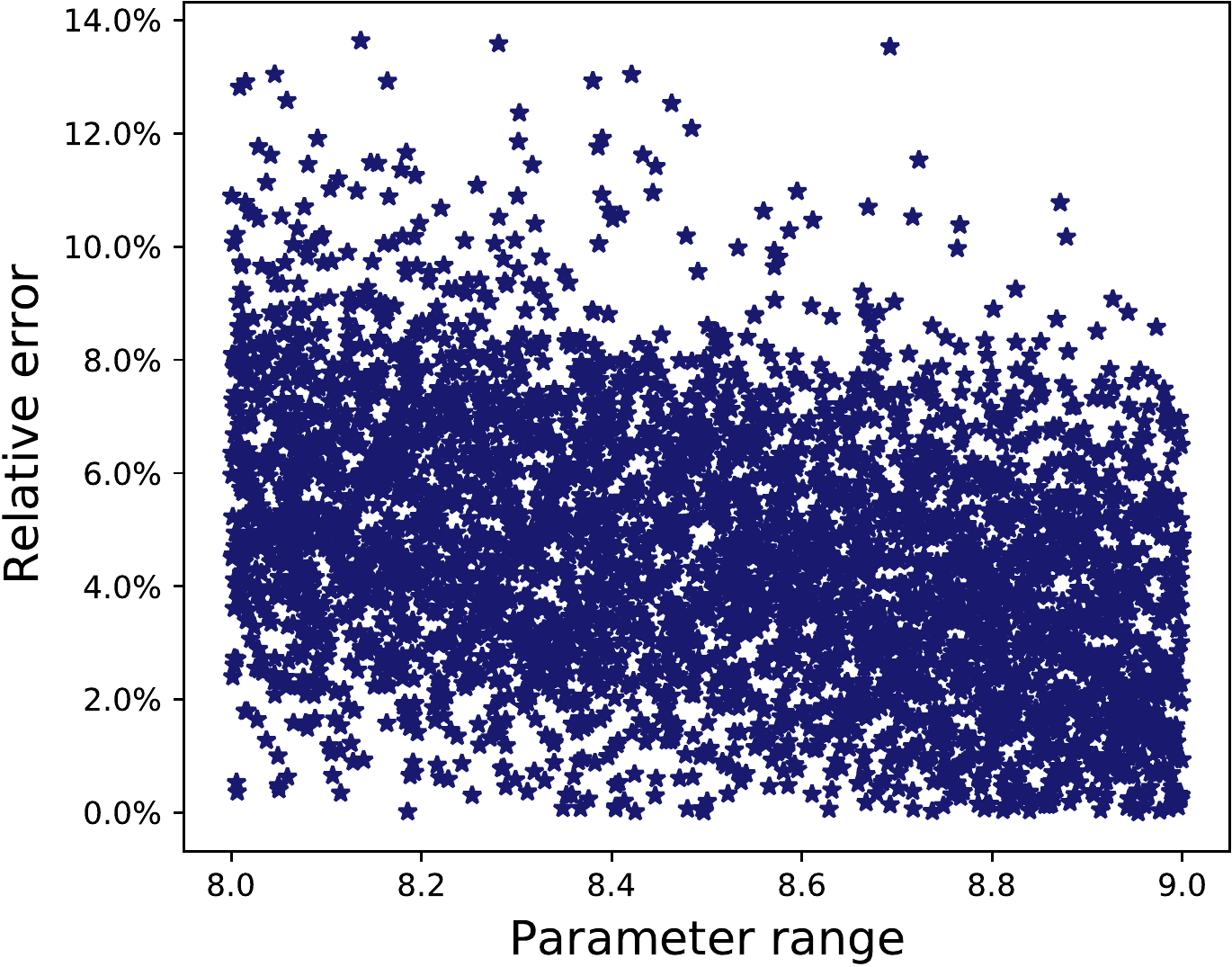}\\
			$\boldsymbol{\alpha_0}$ & $\boldsymbol{\alpha_1}$ &$\boldsymbol{\alpha_2}$\\
			\includegraphics[width=0.33\textwidth]{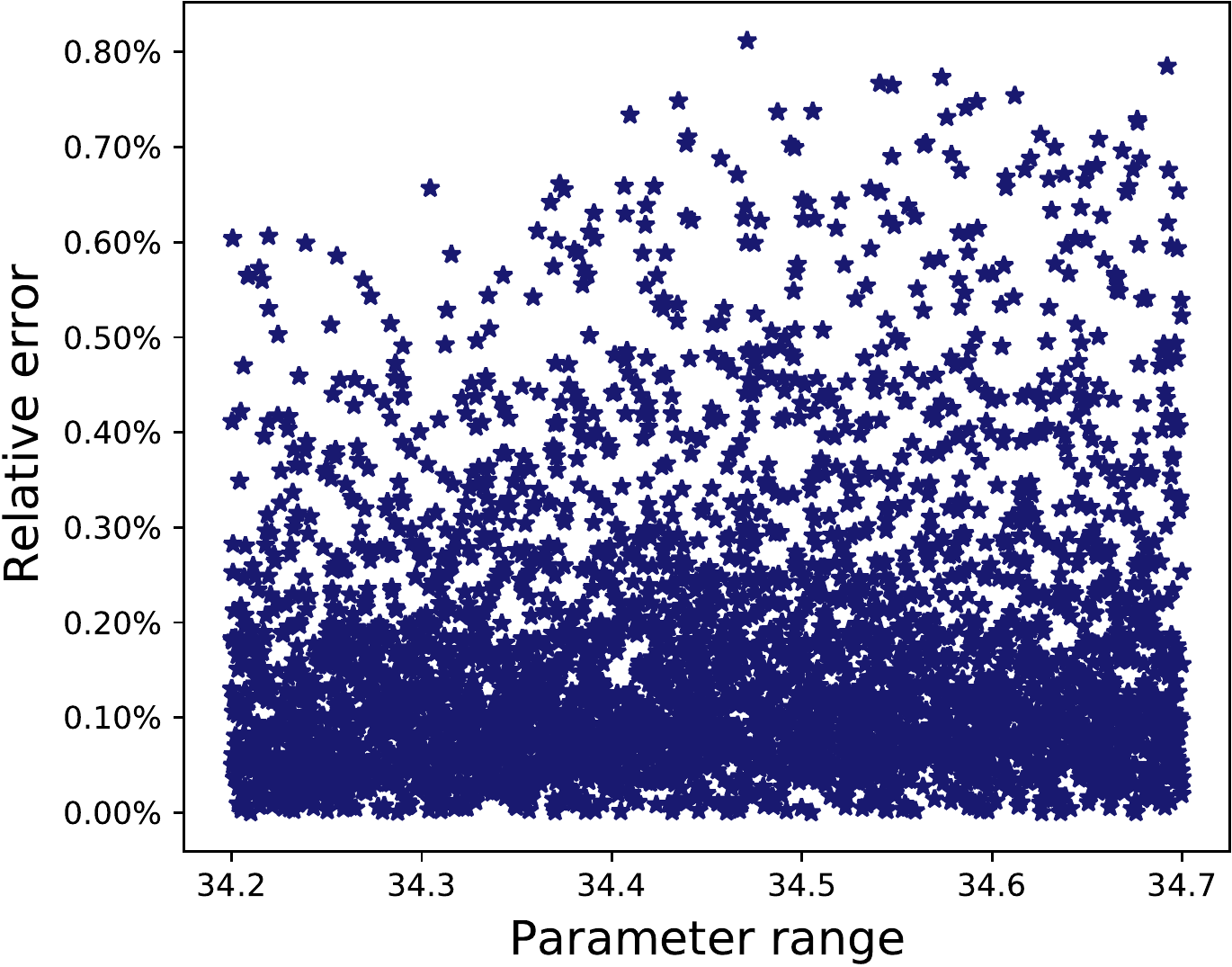} 
			&\includegraphics[width=0.33\textwidth]{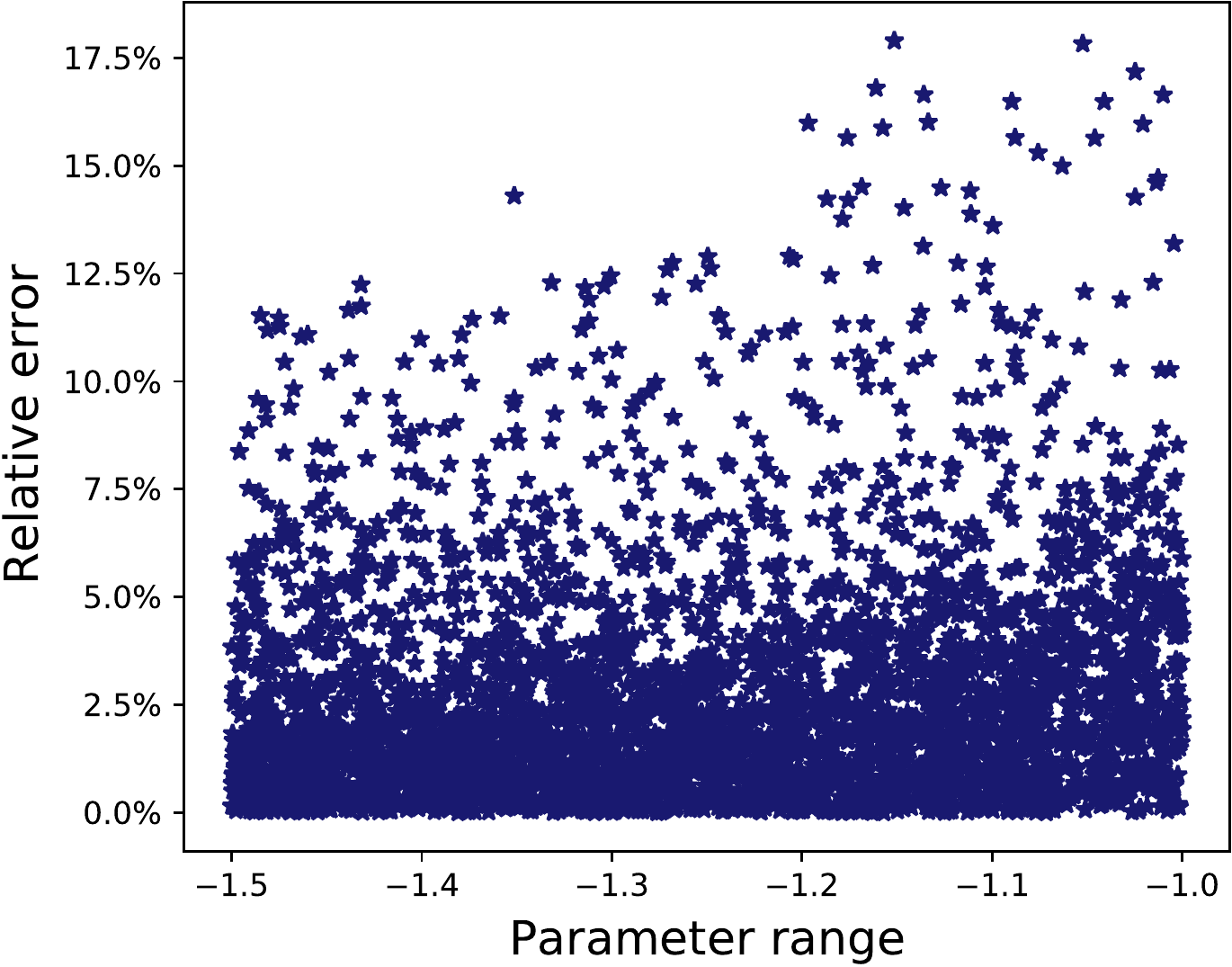}&\includegraphics[width=0.33\textwidth]{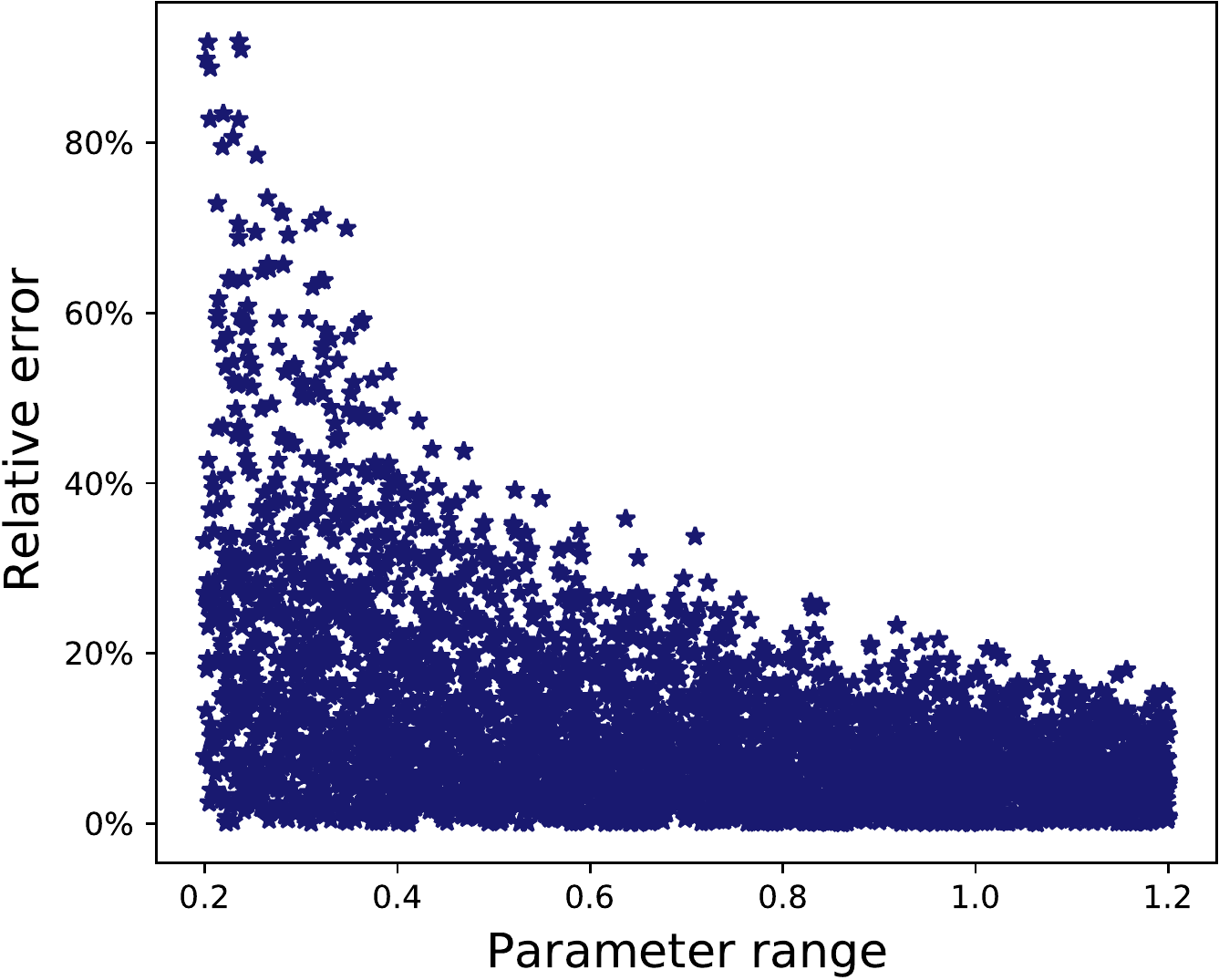}\\
			$\boldsymbol{\alpha_3}$ &  &\\
			\includegraphics[width=0.33\textwidth]{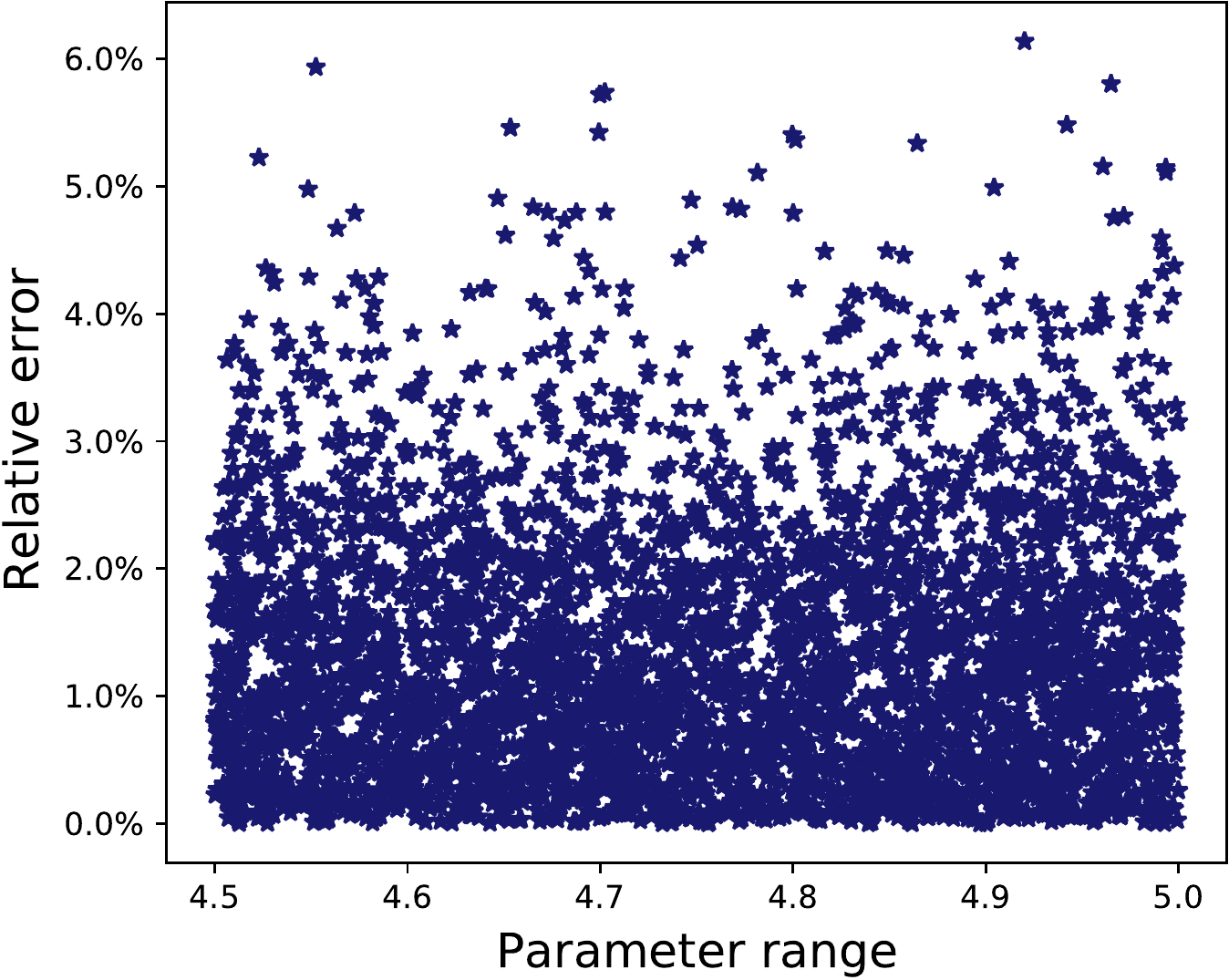} 
			& \vspace{-0.5cm}\begin{tabular}{rrr}
				& \textbf{Mean} & \textbf{Median}\\[3pt]
				$\boldsymbol{a}$: &$46.9\%$ & $47.3\%$\\
				$\boldsymbol{b}$: &$26.3\%$ & $27.6\%$\\
				$\boldsymbol{k}$: &$ 4.59 \%$ & $4.50\%$\\
				$\boldsymbol{\alpha_0}$: &$0.17\%$ & $ 0.12 \%$\\
				$\boldsymbol{\alpha_1}$: &$2.72\%$ & $1.89\%$\\
				$\boldsymbol{\alpha_2}$: &$11.2 \%$ & $7.67\%$\\
				$\boldsymbol{\alpha_3}$: &$1.33\%$ & $1.12\%$\\
			\end{tabular} & 
		\end{tabular}
	}
	\caption{Relative error for the model parameters calibrated with the pointwise learning approach.\label{adam_pointwise}}
\end{figure}

\begin{figure}[!htbp]
	\setlength{\tabcolsep}{10pt}
	\resizebox{1\textwidth}{!}{
		\begin{tabular}{@{}>{\centering}m{0.5\textwidth}>{\centering\arraybackslash}m{0.5\textwidth}@{}}
			\textbf{Average relative error} & \textbf{Maximum relative error}\\
			\includegraphics[width=0.5\textwidth]{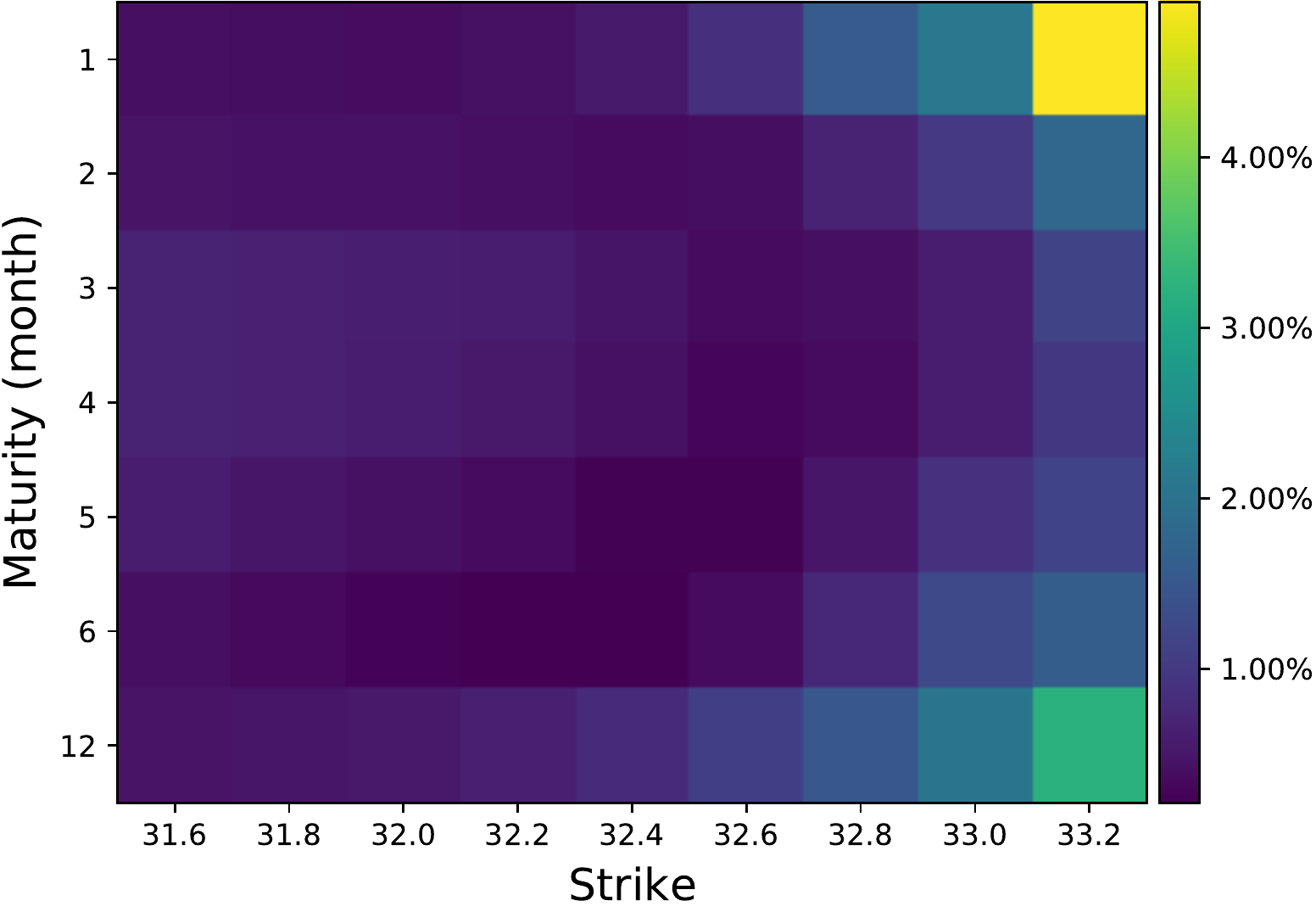} 
			&\includegraphics[width=0.5\textwidth]{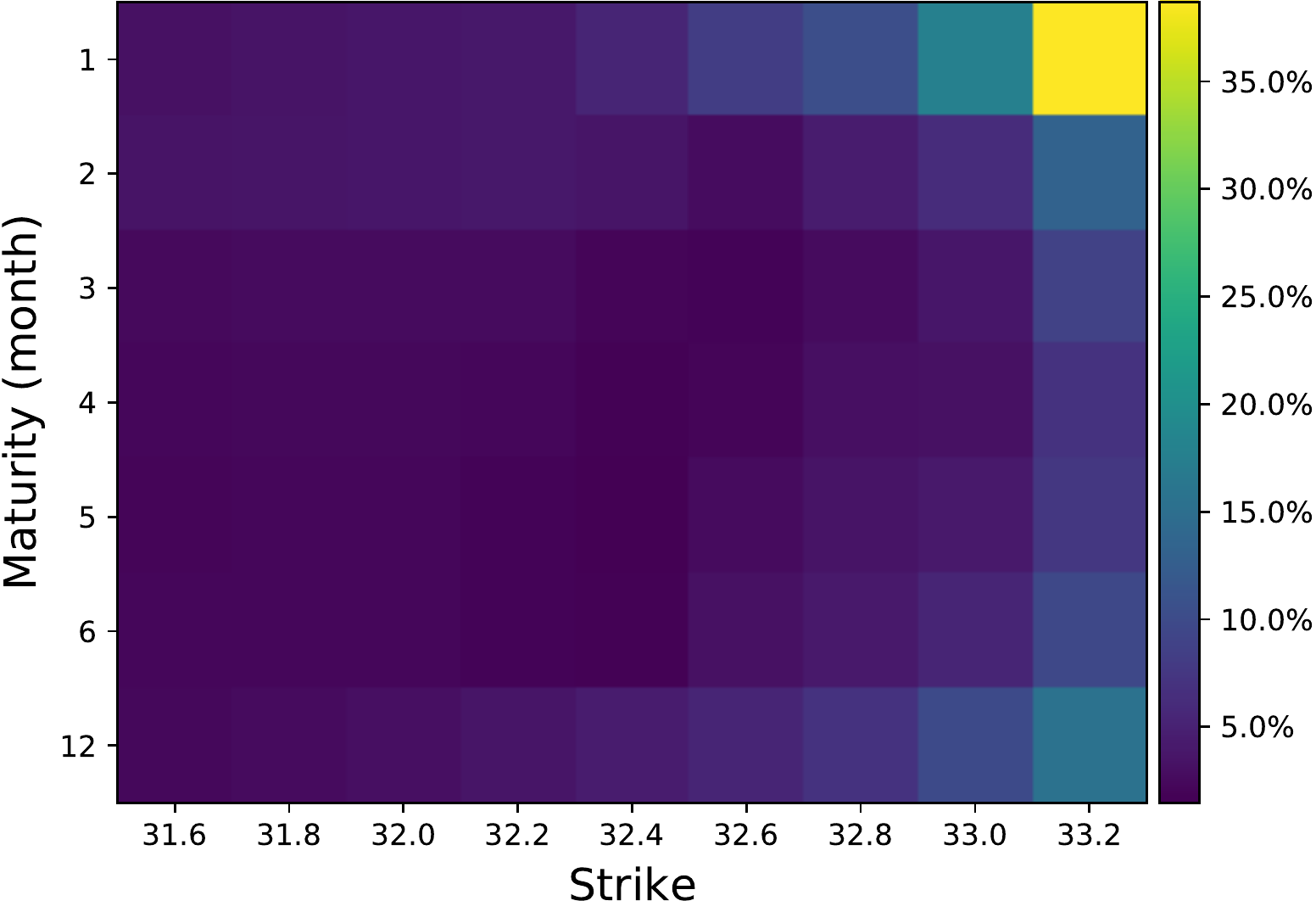}
		\end{tabular}
	}
	\caption{Average relative error and Maximum relative error after calibration with the pointwise learning approach.\label{step2_adam_pointwise}}
\end{figure}

\subsection{Bid-ask constraints}
In the setting of the grid-based learning approach, we test the bid-ask loss function defined in Section \ref{bidasksection}. More precisely, we take the previously trained neural network and use the bid-ask loss function for the calibration step.
Starting from the test set of the grid-based approach, we obtain bid and ask prices by computing $90\%$ and $110\%$, respectively, of the original prices. In Figure \ref{adam_BA} we report the relative error for the model parameters calibrated, which are a bit worse than the values found in the experiment with the same neural network but exact prices. However, considering that the information given to the network is much weaker due to the relatively wide bid-ask range, the results are surprisingly good and the overall error is of similar magnitude as in the rest of the experiments. 

We calculate the rate of mismatched prices, namely the percentage of prices which do not lie within the bid-ask constraints. In Figure \ref{step2_adam_BA}, on the left we have the plot of the mismatched price rate that we observe before calibration, namely the rate we obtain by plugging into the neural network the starting parameters that we need in order to initialize the optimization routine. On the right, the plot shows the percentage of mismatched prices after calibration. The final rate is $0\%$ for almost all the contracts, except for $(\tau, K) = (1, 33.2)$ and for $(\tau, K) = (1/12, 33.2)$. A random sub-sample of $100$ prices is reported in Figure \ref{bidask_example} where, in the top panel (corresponding to $(\tau, K) = (1, 33.2)$), we can notice several prices being outside the constraints (marked with the symbol $\scriptstyle\bigstar$), in the mid panel (corresponding to $(\tau, K) = (1/12, 33.2)$) only some of the prices are outside the constraints, while in the bottom panel (corresponding to $(\tau, K) = (4/12, 32.2)$) all the prices are within the constraints (marked with the symbol $\times$), referring to a rate of mismatching equal to $0\%$. However, the prices lying outside the bid-ask interval are indeed very close to the boundaries, which confirms the suitability of the bid-ask loss function whenever exact prices are not available.  We have been testing the bid-ask loss function also with more narrow constraints. In this case, it is sufficient to increase the number of iterations for the optimizing routine to obtain very similar results.

\begin{figure}[!htbp]
	\setlength{\tabcolsep}{5pt}
	\resizebox{1\textwidth}{!}{
		\begin{tabular}{@{}>{\centering}m{0.33\textwidth}>{\centering}m{0.33\textwidth}>{\centering\arraybackslash}m{0.33\textwidth}@{}}
			$\boldsymbol{a}$ & $\boldsymbol{b}$&$\boldsymbol{k}$\\
			\includegraphics[width=0.33\textwidth]{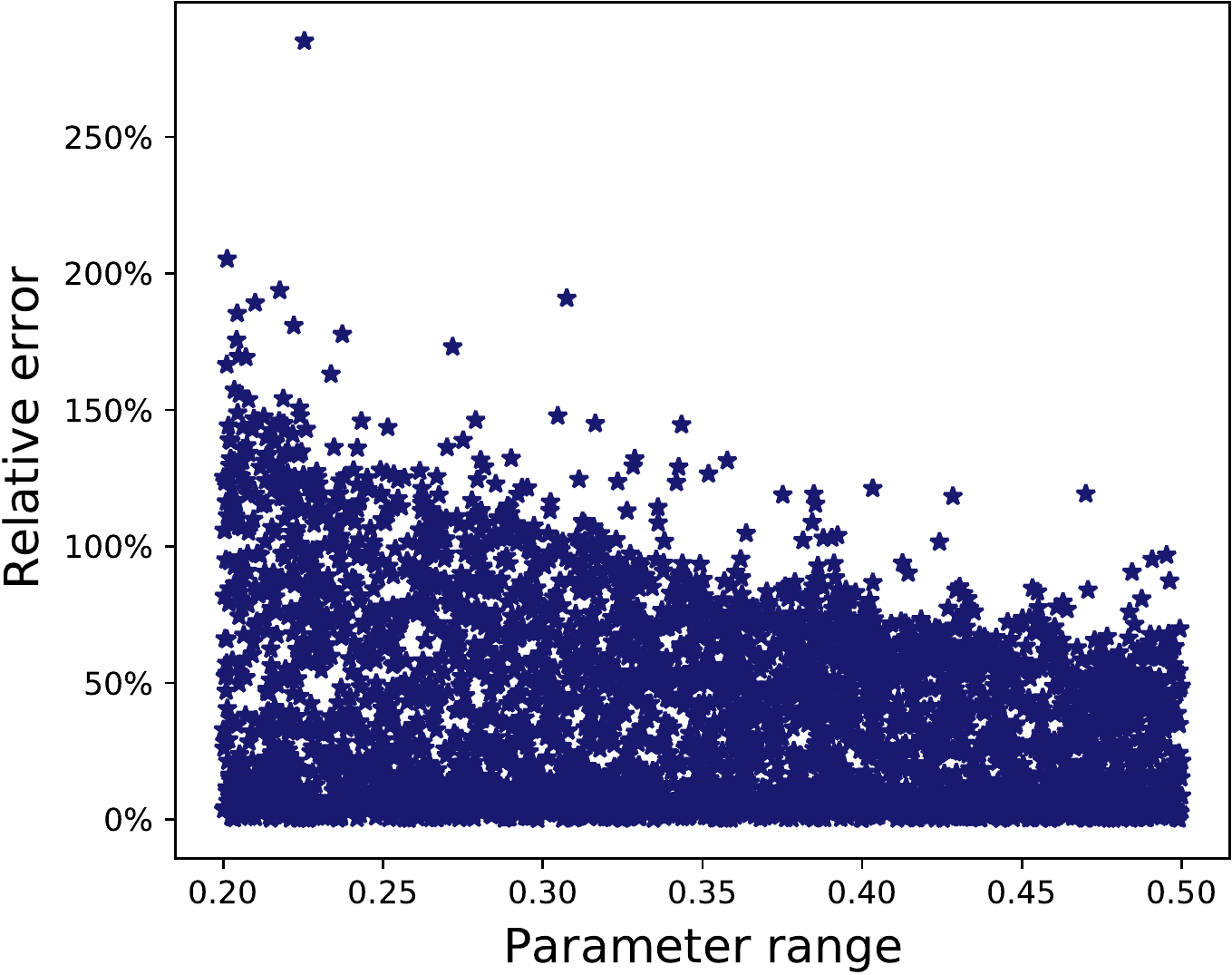} 
			&\includegraphics[width=0.33\textwidth]{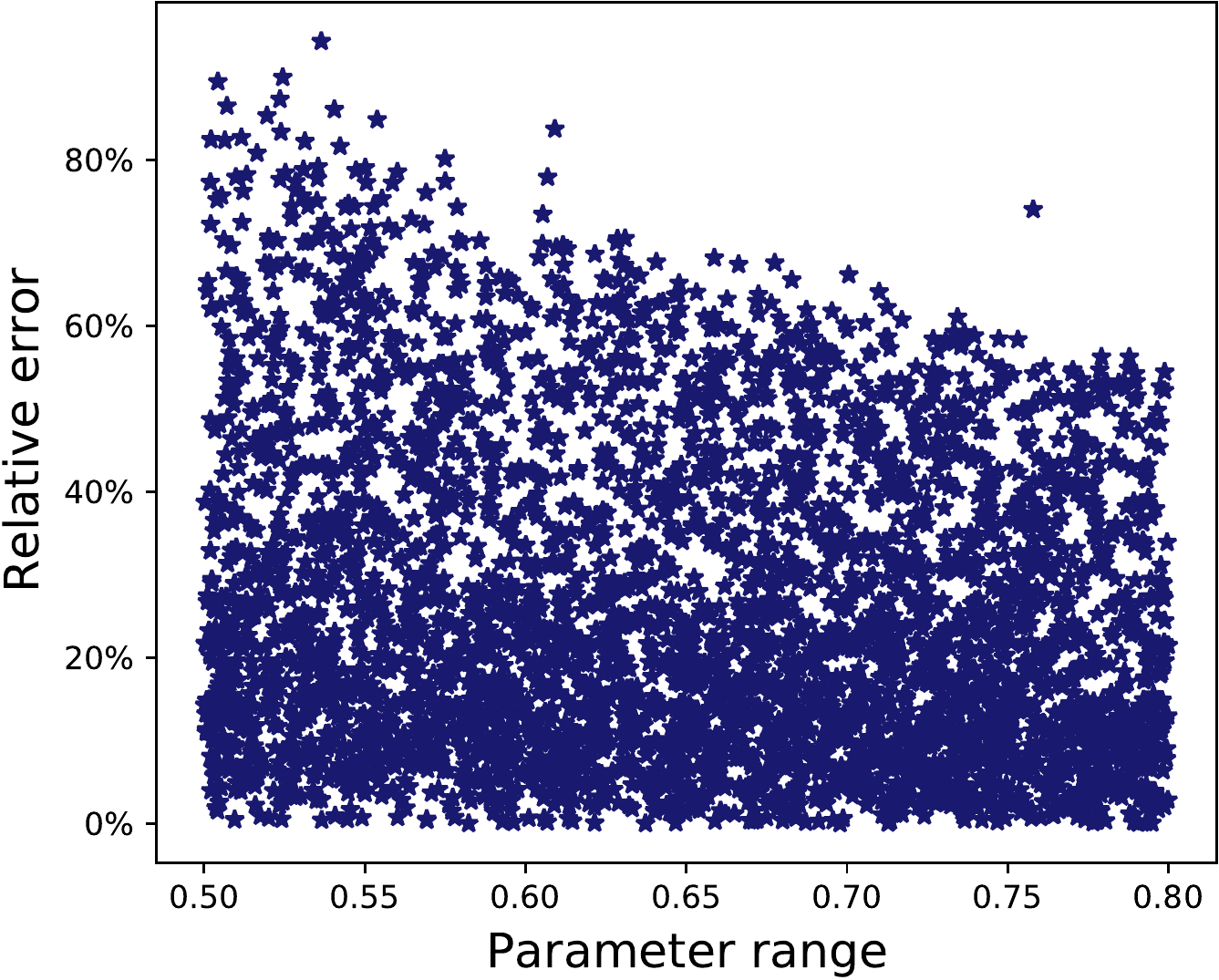}&\includegraphics[width=0.33\textwidth]{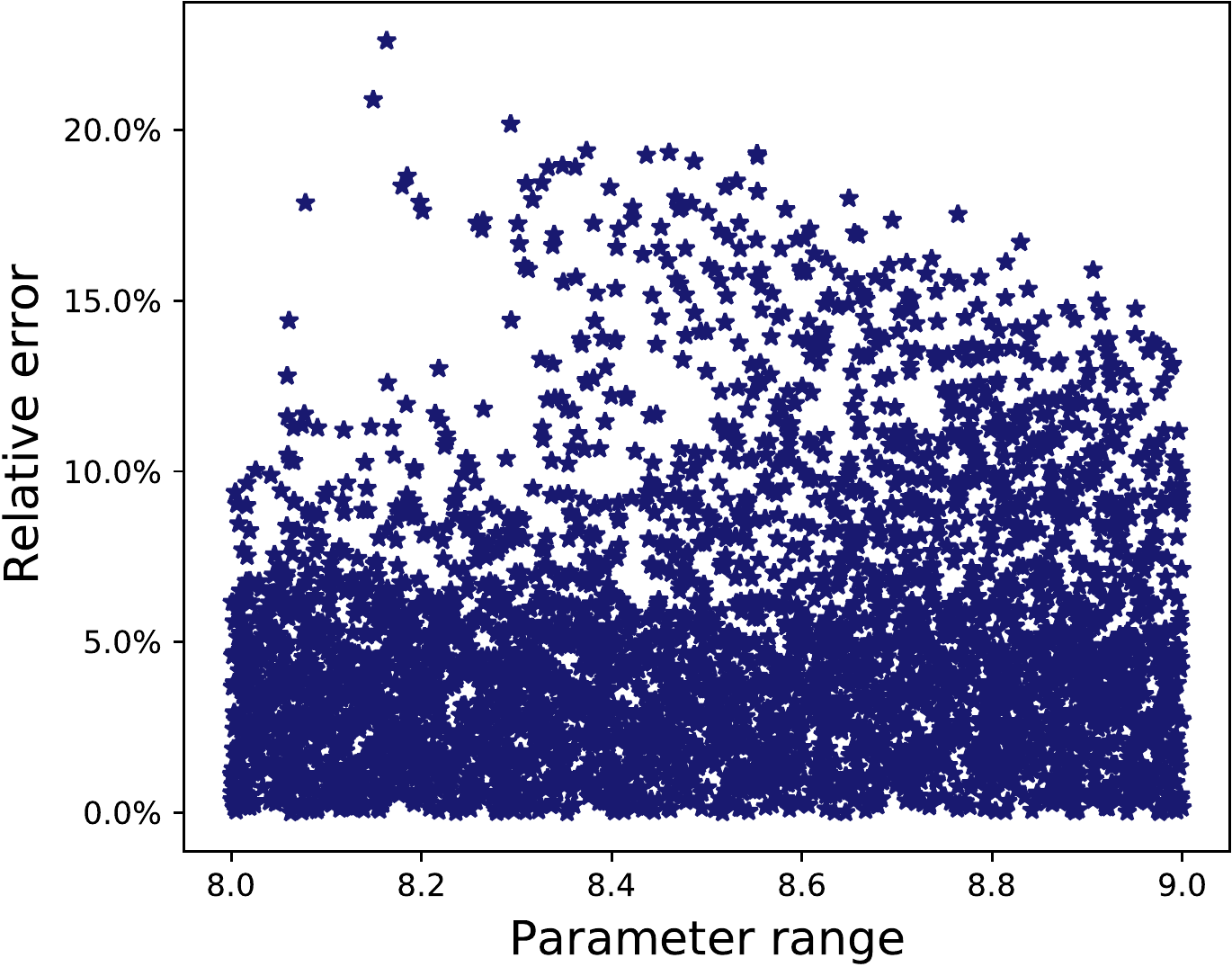}\\
			$\boldsymbol{\alpha_0}$ & $\boldsymbol{\alpha_1}$ &$\boldsymbol{\alpha_2}$\\
			\includegraphics[width=0.33\textwidth]{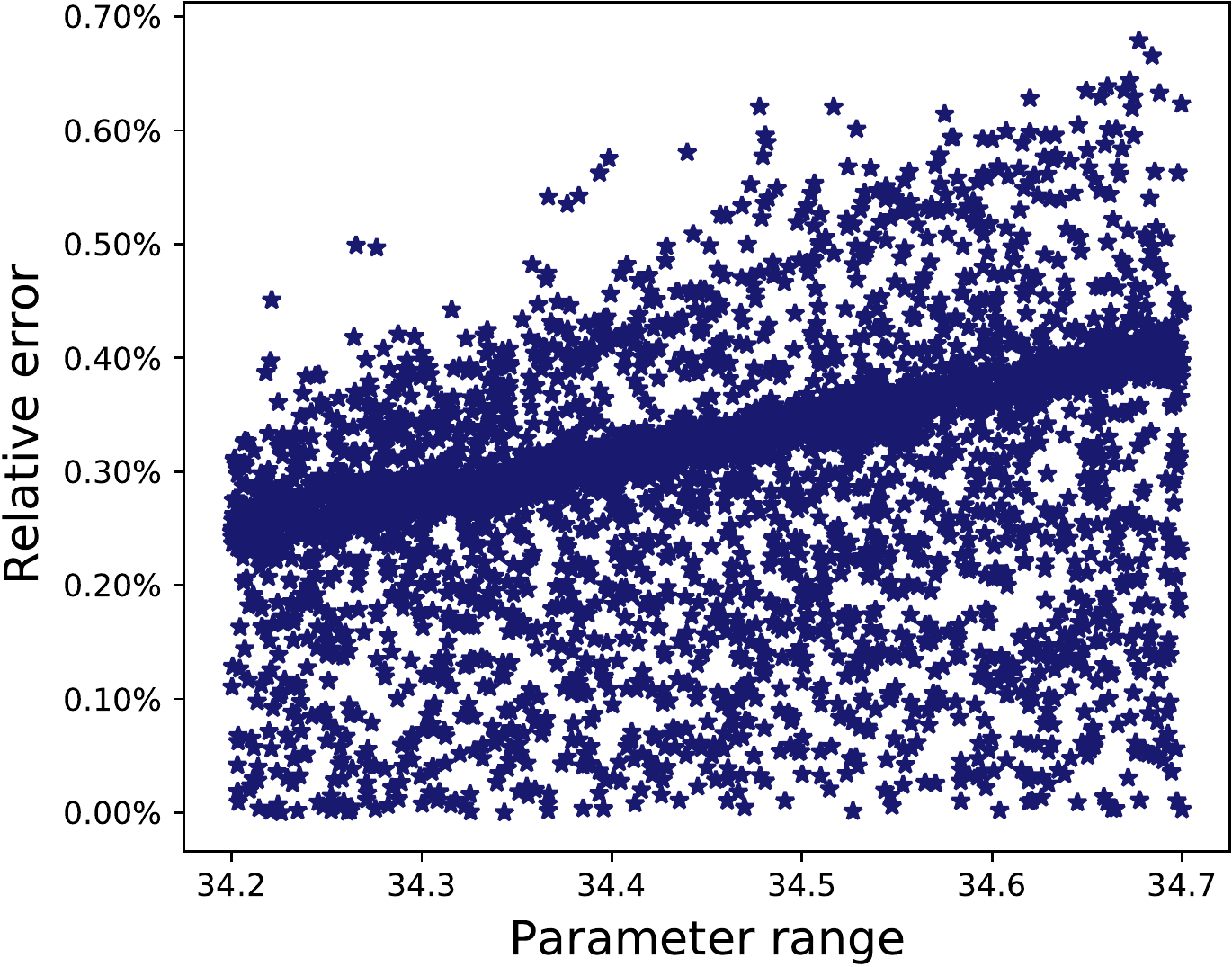} 
			&\includegraphics[width=0.33\textwidth]{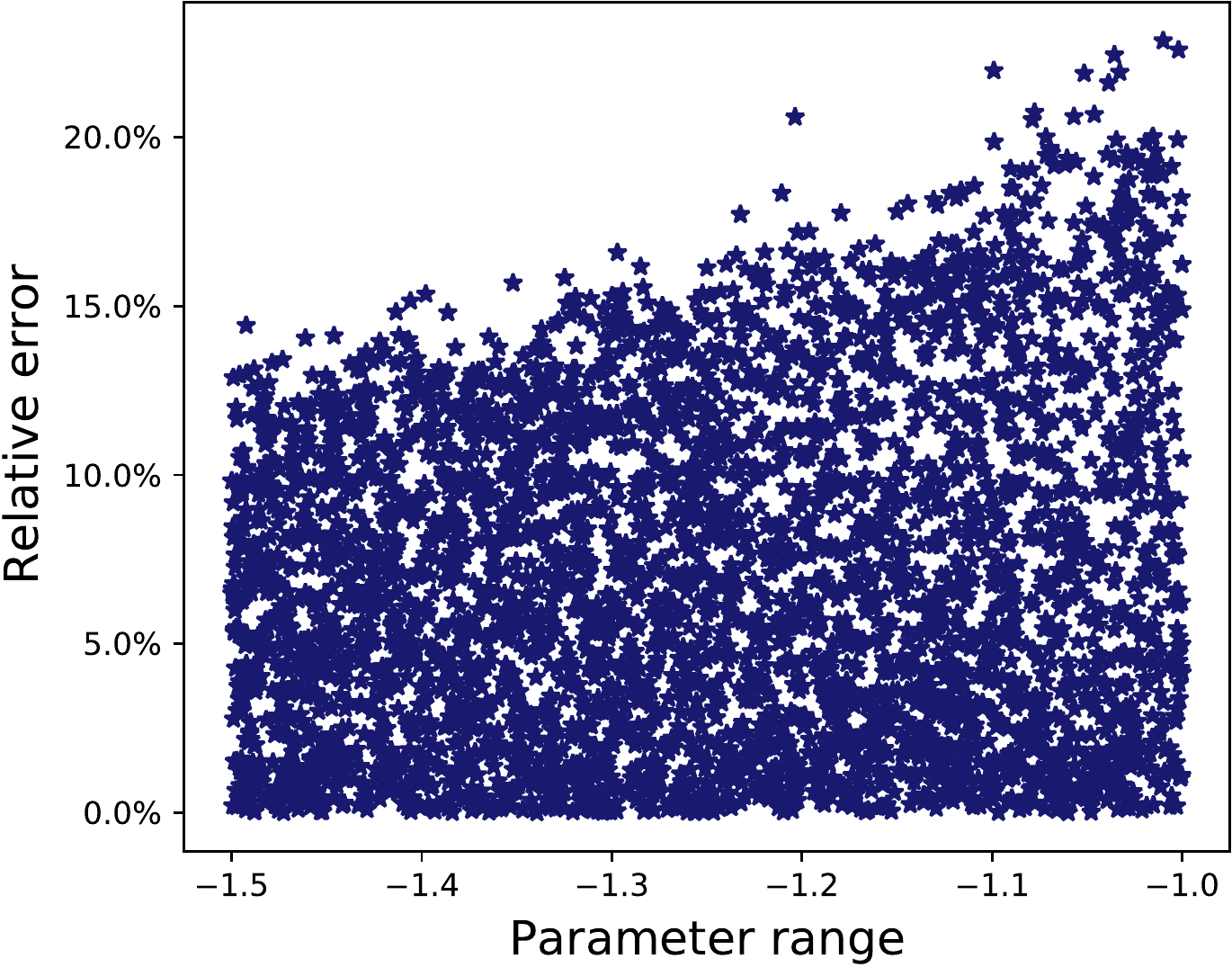}&\includegraphics[width=0.33\textwidth]{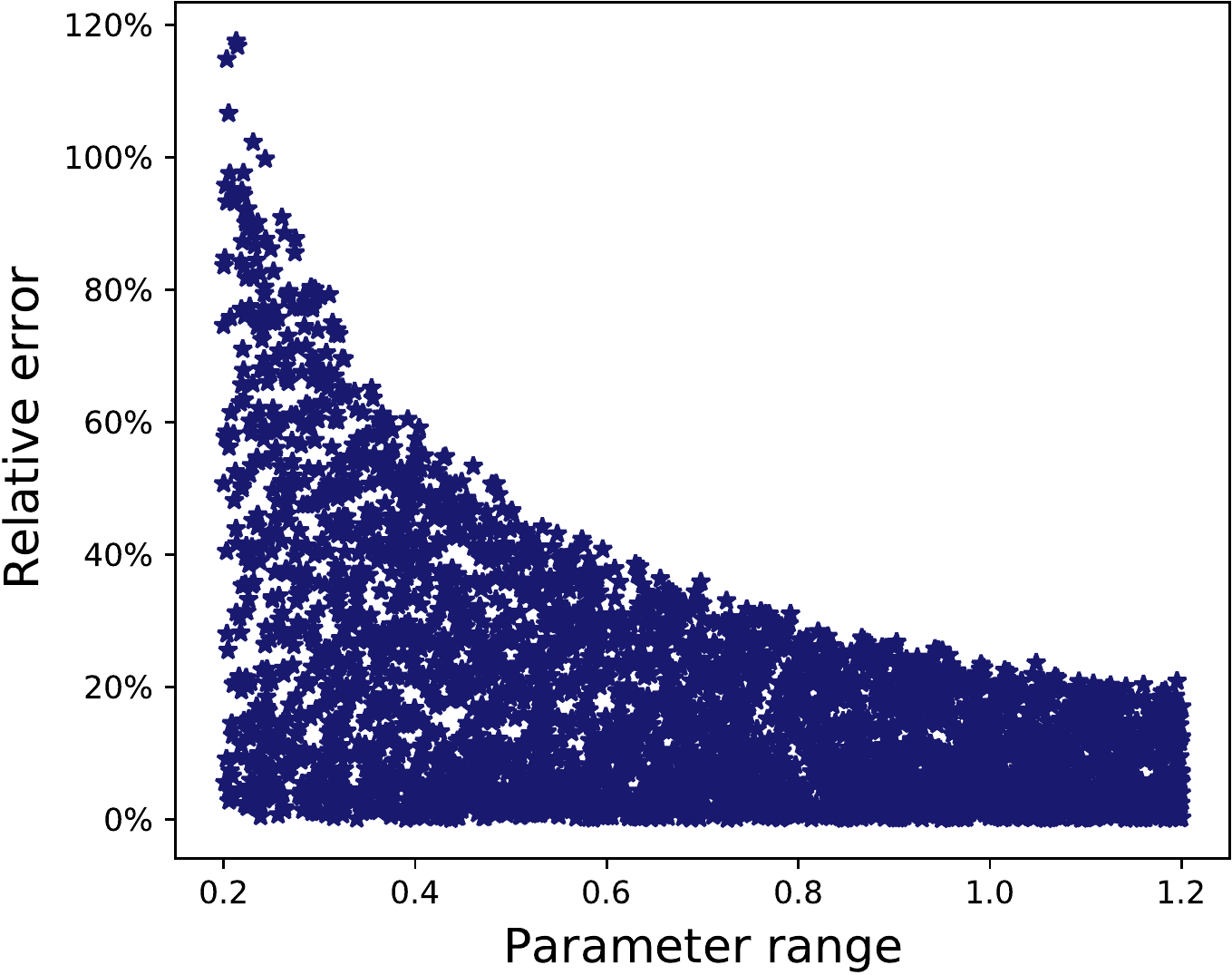}\\
			$\boldsymbol{\alpha_3}$ &  &\\
			\includegraphics[width=0.33\textwidth]{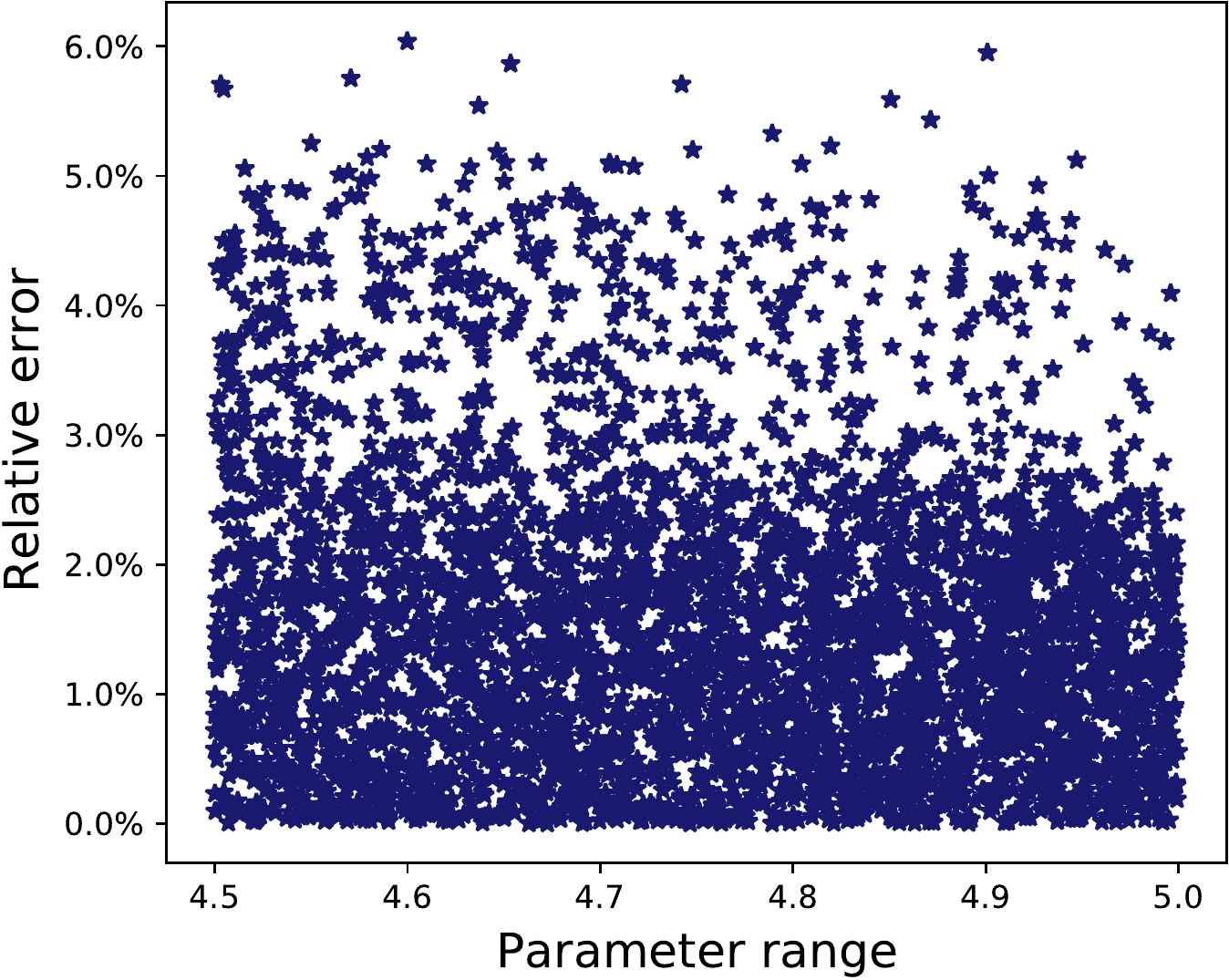} 
			& \vspace{-0.5cm}\begin{tabular}{rrr}
				& \textbf{Mean} & \textbf{Median}\\[3pt]
				$\boldsymbol{a}$: &$40.7\%$ & $33.3\%$\\
				$\boldsymbol{b}$: &$26.0\%$ & $21.3\%$\\
				$\boldsymbol{k}$: &$ 4.95 \%$ & $4.07\%$\\
				$\boldsymbol{\alpha_0}$: &$0.29\%$ & $ 0.30 \%$\\
				$\boldsymbol{\alpha_1}$: &$7.06\%$ & $6.47\%$\\
				$\boldsymbol{\alpha_2}$: &$17.5 \%$ & $12.5\%$\\
				$\boldsymbol{\alpha_3}$: &$1.57\%$ & $1.37\%$\\
			\end{tabular} & 
		\end{tabular}
	}
	\caption{Relative error for the model parameters calibrated in the grid-based learning approach with the bid-ask constraints.\label{adam_BA}}
\end{figure}

\begin{figure}[!htbp]
	\setlength{\tabcolsep}{10pt}
	\resizebox{1\textwidth}{!}{
		\begin{tabular}{@{}>{\centering}m{0.5\textwidth}>{\centering\arraybackslash}m{0.5\textwidth}@{}}
			\textbf{Starting mismatch rate} & \textbf{Final mismatch rate}\\
			\includegraphics[width=0.5\textwidth]{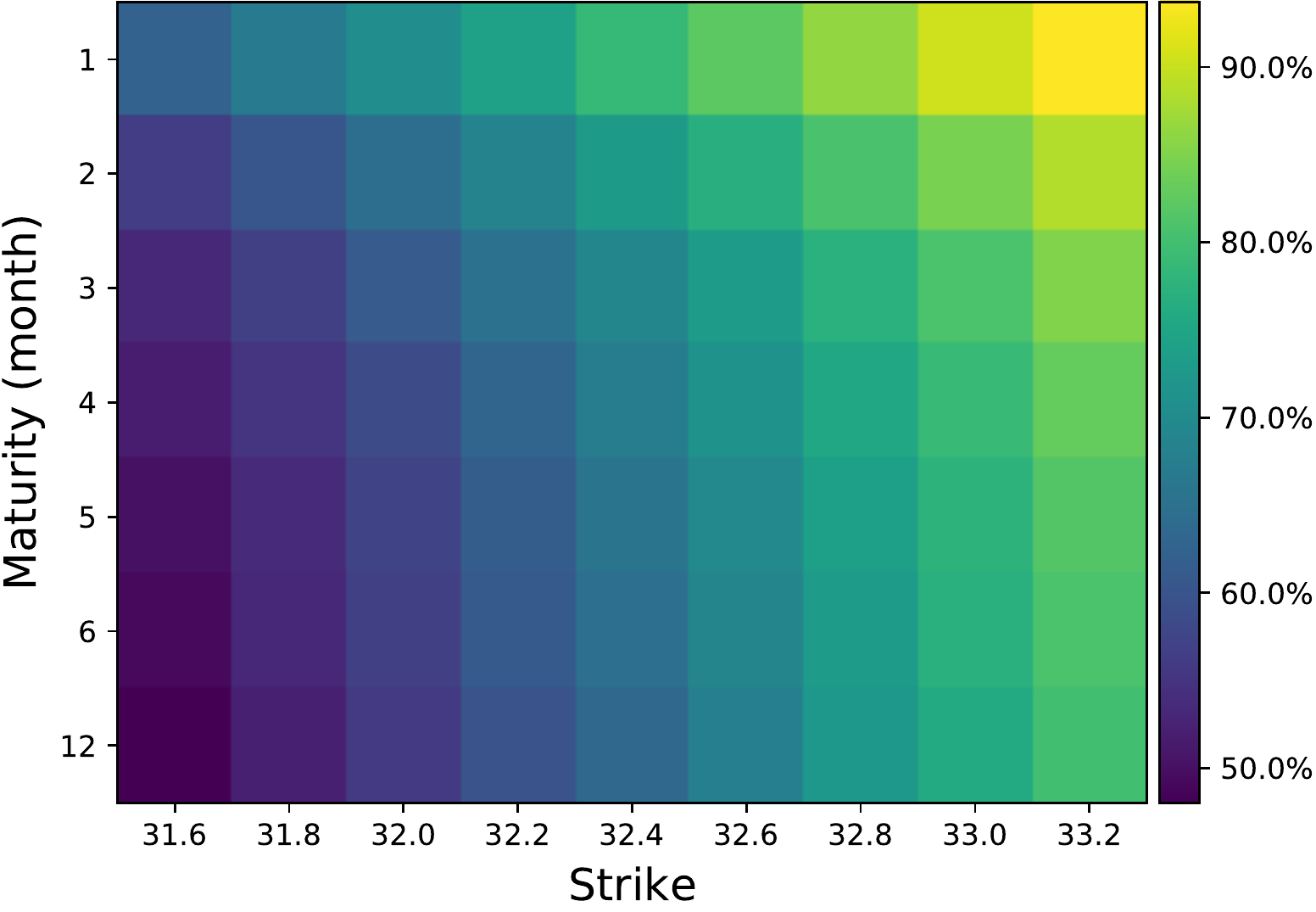} 
			&\includegraphics[width=0.5\textwidth]{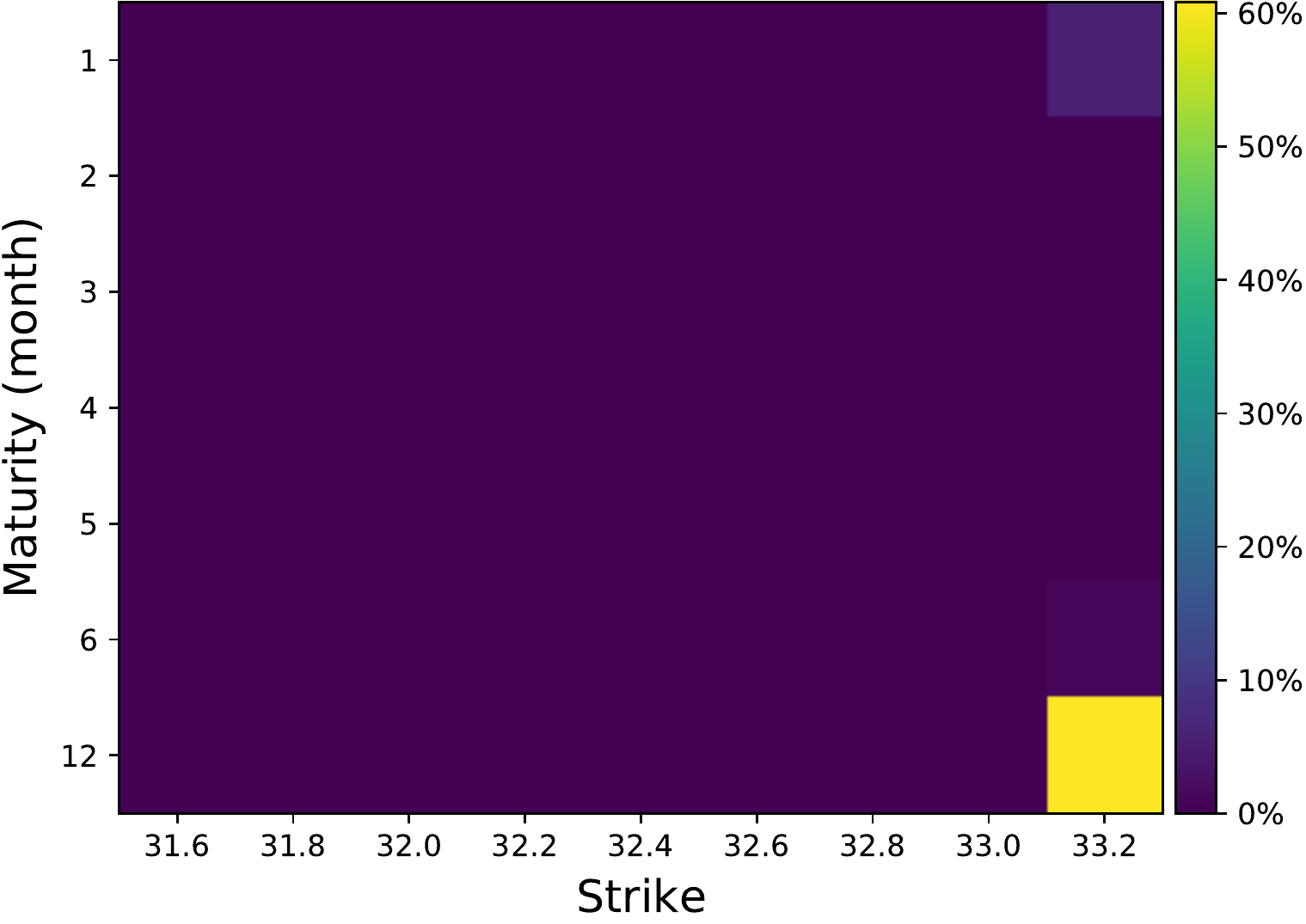}
		\end{tabular}
	}
	\caption{Starting mismatch rate and Final mismatch rate after calibration in the grid-based learning approach with the bid-ask constraints.\label{step2_adam_BA}}
\end{figure}

\begin{figure}[!htbp]
	\setlength{\tabcolsep}{0pt}
	\resizebox{1\textwidth}{!}{
		\begin{tabular}{@{}>{\centering\arraybackslash}m{\textwidth}@{}}
			{\large\textbf{Bid-ask constraint examples}} \\
			$K = 33.2$, $\tau =1$\\ 
			\includegraphics[width=\textwidth, height=4cm,keepaspectratio]{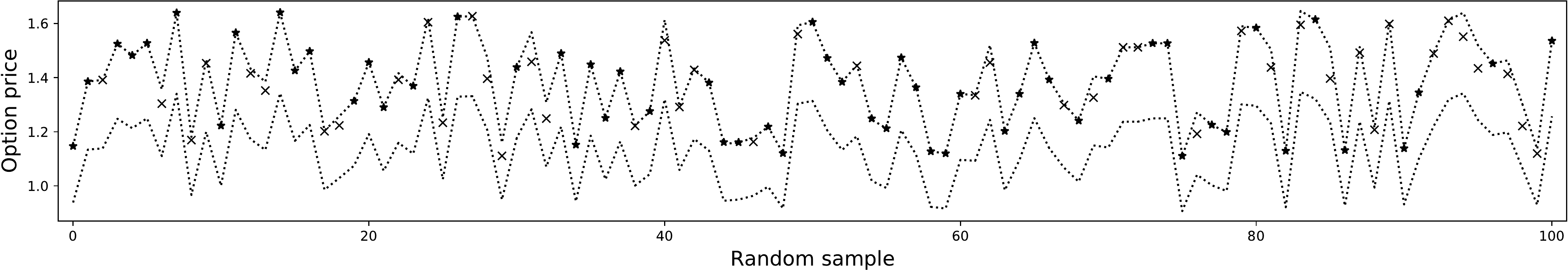} \\
			$K = 33.2$, $\tau =1/12$\\ 
			\includegraphics[width=\textwidth, height=4cm,keepaspectratio]{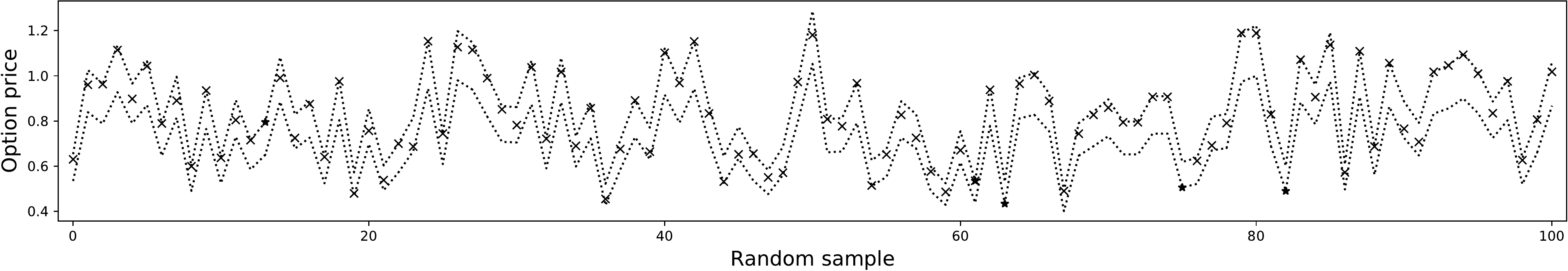} \\
			$K = 32.2$, $\tau =4/12$\\ 
			\includegraphics[width=\textwidth, height=4cm,keepaspectratio]{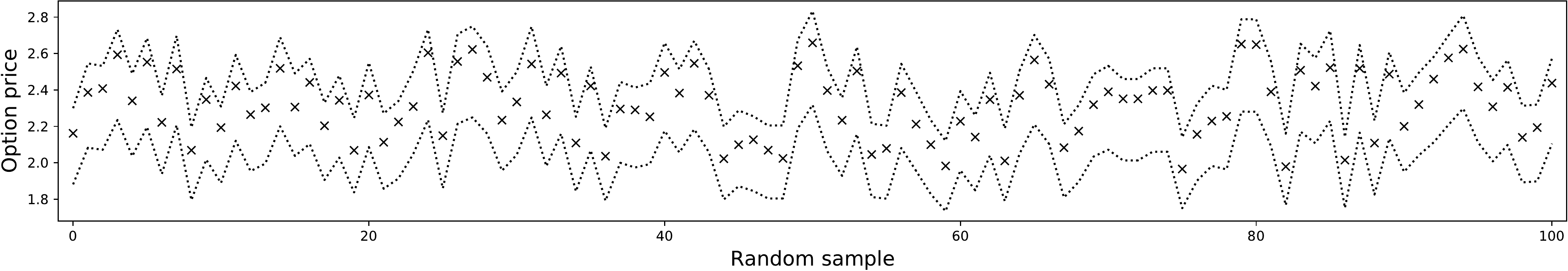}
		\end{tabular}
	}
	\caption{Examples of a $20\%$ bid-ask constraint with the grid-based learning approach. The symbol $\times$ indicates the prices which after calibration are inside the constraints, the symbol $\scriptstyle\bigstar$ the ones outside. \label{bidask_example}}
\end{figure}

\section{Conclusions, remarks and further ideas}\label{conclusion}
The main contributions of this paper can be summarized as follows. We propose a fully parametrized HJM model for the forward curve in energy markets. This involves the specification of a class of volatility operators that map the square integrable infinite dimensional noise into the state space of the forward curve. We then derive pricing formulas for European options on swaps, forward contracts delivering over a period of time. For calibrating the model, we propose an innovative approach based on deep neural networks adapting the recent work by \cite{nn}. We test the accuracy of deep neural networks for model calibration in the infinite-dimensional HJM framework. 

The numerical experiments are satisfactory and reveal that machine learning techniques are a valuable tool to make infinite-dimensional models more tractable. The prices obtained with the neural network, both, before and after calibration, are indeed quite accurate. Also the calibration for markets with large bid-ask prices works well with a new loss function proposed here. On the other hand, the original meaning of the parameters in the underlying model may get lost during this procedure. This phenomenon may be due to the non-injectivity of the pricing map, but also of the neural network itself. This observation stresses the importance of proper benchmarking of a neural network based approach to avoid overfitting. As stated in \cite{deep1}, computing the distance between the true model parameters and the estimated ones, might not be the optimal way to quantify the accuracy of the two-step approach. In \cite{deep1} the authors therefore suggest to switch to a Bayesian viewpoint and observe the posterior distribution of the model parameters, considered as a random variable. We explore the issue of non-injectivity in more detail in Appendix \ref{injectivity}.

Although non-injectivity is an issue somehow intrinsic to the problem, one might still be able to improve the accuracy in calibration. Since forward contracts are traded themselves, it may be interesting to divide the calibration step in two parts. By standard techniques for interpolation, one could use the forward prices to calibrate the initial forward curve parameters. These can then be considered as fixed in the main calibration step, so that the vector of parameters to recover with the neural network is smaller. 
One might also decide to fix the initial forward curve already at the level of the approximation step, and then train the neural network with a lower dimensional input vector $\theta$. 
However, this would require to train the network more frequently to be updated with the current market information. 

Another possibility without changing the setup at all, is to include forward contracts as options with strike $K=0$ and time to maturity $\tau=t$ equal to the evaluation time. These contracts most efficiently allow to recover the parameters for the Nelson-Siegel curve. Given these considerations, our approach to calibrate the entire parameter vector $\theta$ only based on call options, is somehow likely to lead to a larger error than the one obtained in practice. In a fully fledged option pricing platform results are likely to be much better than in this case study which serves merely as a proof of concept. For instance, when pricing standard options in equity markets based on a Black Scholes model, only the implied volatility is calibrated, while the current stock price, dividends and interest rate are obtained from other sources.

\appendix
\section{Proofs of the main results}
\label{appendix:proofs}
We report in this section the proofs to the main results in the order they appear in the paper.

\subsection{Proof of Theorem \ref{kernel:cond}}
\label{kernel:cond:proof}
For every $x\in \R_{+}$ and $f \in \h_{\alpha}$, by the Cauchy-Schwarz inequality we can write that
\begin{equation*}
	\int_{\set} \left| \kappa_t (x,y,f) h(y) \right| dy \leq \left( \int_{\set}  \kappa_t (x,y,f)^2 dy \right)^{1/2} \left( \int_{\set}  h ( y)^2 dy \right)^{1/2} < \infty,
\end{equation*}
which is bounded, since $\kappa_t (x,\cdot,f) \in \h$ for every $x\in  \mathbb{R}_+$ and every $f \in \h_{\alpha}$ by Assumption 1, and because $h\in \h$. Thus $\sigma_t(f) h$ is well defined for all $h \in \h$. 

We need to show that $\sigma_t(f) h \in \h_{\alpha}$ for every $f \in \h_{\alpha}$. We start by noticing that for every $x\in \R_{+}$ the following equality holds:
\begin{equation*}
	\frac{\partial \sigma_t(f) h (x)}{\partial x} = \int_{\set} \frac{\partial \kappa_t (x,y,f) }{\partial x} h(y) dy,
\end{equation*} 
where the differentiation under the integral sign is justified by Dominated Convergence because of Assumption 2 and $\int_{\set}  \bar{\kappa}_x (y) h(y) dy < \infty$ . Moreover, by Assumption 3 and the Cauchy-Schwarz inequality, we find that 
$$ \int_{\mathbb{R}_+} \left( \int_{\set} \frac{\partial \kappa_t (x,y,f) }{\partial x} h(y) dy  \right)^2 \alpha (x) dx \leq  \norm{h}^2 \int_{\mathbb{R}_+} \norm{\frac{\partial \kappa_t (x,\cdot,f) }{\partial x}}^2  \alpha (x) dx < \infty ,$$
which shows that $\sigma_t(f) h\in \h_{\alpha}$ and boundedness of the operator $\sigma_t(f)$ for each $f \in \h_{\alpha}$. 

\subsection{Proof of Theorem \ref{dualthrm}}
\label{dualthrm:proof}
We first observe that for every $h\in \h$ and every $f, f_1\in \h_{\alpha}$, it holds that
\begin{align*}
	& \int_{\R_{+}}\int_{\R_{+}}\left|\frac{\partial \kappa_t (x,y,f) }{\partial x} f_1'(x)\alpha(x)h(y)\right|dydx =  \int_{\R_{+}} | f_1'(x)\alpha(x) | \int_{\R_{+}} \left|\frac{\partial \kappa_t (x,y,f) }{\partial x} \right| \left| h(y)\right| dydx \\
	&\!\leq \!  \int_{\R_{+}} | f_1'(x) | \alpha^{1/2}(x)  \norm{ \frac{\partial \kappa_t (x,\cdot, f) }{\partial x} }\alpha^{1/2}(x)  \norm{h} dx \leq   \norm{h} \norm{f_1}_{\alpha}  \!\left(\!\int_{\R_{+}}  \norm{ \frac{\partial \kappa_t (x,\cdot,f) }{\partial x} }^2\alpha (x)   dx\!\right)^{1/2}
\end{align*} 
where we used the Cauchy-Schwarz inequality twice. By Assumption 3. this is bounded and allows us to apply the Fubini Theorem and calculate as follows:
\begin{align*}
	\langle \sigma_t(f) h, f_1\rangle_{\alpha}  & = f_1(0) \int_{\set}\kappa_t (0,y,f) h(y)dy +\int_{\R_{+}}\frac{\partial \sigma_t(f) h (x)}{\partial x}  f_1'(x)\alpha(x)dx\\
	& =f_1(0) \int_{\set}\kappa_t (0,y,f) h(y)dy +\int_{\R_{+}}   \int_{\set} \frac{\partial \kappa_t (x,y,f) }{\partial x} h(y) dy   f_1'(x)\alpha(x)dx \\
	& = \int_{\set} \left( f_1(0)\kappa_t (0,y,f)   +\int_{\mathbb{R}_+} \frac{\partial \kappa_t (x,y,f) }{\partial x}    f_1'(x)\alpha(x) dx  \right) h(y) dy\\
	& =  \int_{\set}\sigma_t(f)^{*}f_1(y) h(y)dy = \langle h, \sigma_t(f)^{*}f_1\rangle,
\end{align*}
for $\sigma_t(f)^{*}f_1$ defined by 
\begin{equation*}
	\sigma_t(f)^{*}f_1(y):= f_1(0)\kappa_t (0,y,f) +\int_{\R_{+}}\frac{\partial \kappa_t (x,y,f) }{\partial x}  f_1'(x)\alpha(x)dx = \langle \kappa_t (\cdot,y,f), f_1\rangle_{\alpha}.
\end{equation*}
From \cite[Theorem 6.1]{funcanal}, $\sigma_t(f)^{*}$ is the unique adjoint operator of $\sigma_t(f)$, for $f \in \h_{\alpha}$. 

\subsection{Proof of Theorem \ref{kernel:lipcond}}
\label{kernel:lipcond:proof}
We start with the growth condition. For $h\in \h$ and $f_1\in\h_{\alpha}$ we can write that
\begin{align*}
	\norm{\sigma_{t}(f_1)h}^2_{\alpha}& =\left(\sigma_{t}(f_1) h (0)\right)^2 + \int_{\mathbb{R}_+} \left( \frac{\partial \sigma_t (f_1)  h (x)}{\partial x} \right)^2 \alpha (x) dx  \\
	&= \left( \int_{\set} \kappa_{t} (0,y,f_1) h(y) dy \right)^2 + \int_{\mathbb{R}_+} \left( \int_{\set} \frac{\partial \kappa_{t} (x,y, f_1) }{\partial x} h(y) dy  \right)^2 \alpha (x) dx \\
	& \leq  \norm{ \kappa_{t} (0,\cdot, f_1) }^2 \norm{h}^2 + \int_{\mathbb{R}_+} \norm{\frac{\partial \kappa_{t} (x,\cdot, f_1) }{\partial x}}^2 \norm{h}^2 \alpha (x) dx \nonumber\\
	&\leq C(t)^2 (1+ |f_1(0)|)^2 \norm{h}^2 + \int_{\mathbb{R}_+} C(t)^2 f_1'(x)^2 \norm{h}^2 \alpha (x) dx\\
	&\leq  2C(t)^2 (1+ \norm{f_1}_{\alpha})^2 \norm{h}^2,
\end{align*}
where we have used the Cauchy-Schwarz inequality, together with the inequality $\abs{f_1(0)}\leq \norm{f_1}_{\alpha}$ and Assumption 2. With some abuse of notation, it follows that $\norm{\sigma_{t}(f_1)}_{\mathcal{L}(\h, \h_{\alpha})} \leq C(t)(1+\norm{f_1}_{\alpha})$ for a suitably chosen constant $C(t)$. Similarly, from Assumption 1 it follows that 
\begin{align*}
	&\norm{(\sigma_{t}(f_1)-\sigma_{t}(f_2))h}^2_{\alpha} = \left( \int_{\set} \left(\kappa_{t} (0,y, f_1) - \kappa_{t} (0,y, f_2)\right)h(y) dy \right)^2 + \\
	& \qquad\qquad \qquad\qquad \qquad + \int_{\mathbb{R}_+} \left( \int_{\set} \left(\frac{\partial \kappa_{t} (x,y,f_1) }{\partial x} - \frac{\partial \kappa_{t} (x,y, f_2) }{\partial x} \right)h(y) dy  \right)^2 \alpha (x) dx \\
	& \leq  \norm{ \kappa_{t} (0,\cdot, f_1) - \kappa_{t} (0,\cdot, f_2)}^2 \norm{h}^2 +\int_{\mathbb{R}_+} \norm{\frac{\partial \kappa_{t} (x,\cdot, f_1) }{\partial x}- \frac{\partial \kappa_{t} (x,\cdot, f_2) }{\partial x}}^2 \norm{h}^2 \alpha (x) dx \\
	&\leq  C(t)^2  \abs{f_1(0) -f_2(0)}^2  \norm{h}^2 + \int_{\mathbb{R}_+} C(t)^2 (f_1'(x)- f_2'(x) )^2 \norm{h}^2 \alpha (x) dx\\
	&\leq  2C(t)^2 \norm{f_1 -f_2}_{\alpha}^2 \norm{h}^2,
\end{align*}
from which $\norm{\sigma_{t}(f_1)-\sigma_{t}(f_2)}_{\mathcal{L} (\h , \h_{\alpha})}\leq C(t) \norm{f_1 -f_2}_{\alpha} $ for a suitably chosen $C(t)$, which proves the Lipschitz continuity of the volatility operator, and concludes the proof. 

\subsection{Proof of Proposition \ref{weight:cond}}
\label{weight:cond:proof}
In order for the volatility operator $\sigma_t$ to be well defined, we need to check that the function $\kappa_t$ introduced in equation \eqref{kappa} satisfies the assumptions of Theorem \ref{kernel:cond}. We start by observing that $\kappa_t(x,\cdot) \in \h$ if and only if $\omega\in \h$. Then, we can calculate the derivative
\begin{equation*}
	\frac{\partial \kappa_t(x,y)}{\partial x} = a(t)e^{-bx}\left(\omega '(x-y)-b\omega(x-y)  \right),
\end{equation*} 
which, in particular, by Assumption 2 is bounded by
\begin{equation*}
	\abs{\frac{\partial \kappa_t(x,y)}{\partial x}} \leq a(t)e^{-bx} \bar{\omega}_x (y).
\end{equation*} 
For the $\h$-norm we then have that
\begin{equation*}
	\norm{	\frac{\partial \kappa_t(x,\cdot)}{\partial x} }^2 = 
	\int_{\set}\left(\frac{\partial \kappa_t(x,y)}{\partial x}\right)^2 dy \le a(t)^2 e^{-2bx} C_1^2 < \infty,
\end{equation*} 
where we have used that $\norm{\bar{\omega}_x}\leq C_1$, which implies that Assumption 3 in Theorem \ref{kernel:cond} is satisfied for $\alpha$ such that $\int_{\R_+} e^{-2bx} \alpha(x) <\infty$.
Finally, the Lipschitz condition is trivially satisfied and the growth condition is fulfilled because $a(t)$ is bounded. 

\subsection{Proof of Lemma \ref{deliveryfunction}}
\label{deliveryfunction:proof}
For $w$ in equation \eqref{wswap}, we get that $w_{\ell}(v)= \frac{1}{\ell}$ and $\mathcal{W}_{\ell}(u) = \frac{u}{\ell}$. Then $$q_{\ell}^{w}(x, y) =  \frac{1}{\ell}\left(\ell -y + x\right)\I_{[0,\ell]}(y-x),$$ and from equations \eqref{Dell} and \eqref{int} we can write that
\begin{equation*}
	\D_{\ell}^w(g_t)(x) = g_t+ \frac{1}{\ell}\int_{0}^{\infty}\left(\ell -y + x\right)\I_{[0,\ell]}(y-x)g_t'(y)dy = g_t+ \frac{1}{\ell}\int_{x}^{x+\ell}\left(\ell -y + x\right)g_t'(y)dy.
\end{equation*}
Integration by parts gives the result.

\subsection{Proof of Proposition \ref{sigma:integral:prop}}
\label{sigma:integral:prop:proof}
Let $f := \D_{\ell}^{w*}\delta_{T_1-s}^*(1)$. We start by applying the covariance operator to $h := \sigma_s(g_s)^*f$:
\begin{align*}
	\left(\q\sigma_s(g_s)^*f\right)(x) &= \int_{\set}q(x,y)\sigma_s(g_s)^*f(y)dy \\&= \int_{\set}q(x,y)\left\langle \kappa_s (\cdot,y, g_s), f\right\rangle_{\alpha}dy = \left\langle \int_{\set}q(x,y) \kappa_s (\cdot,y, g_s)dy, f\right\rangle_{\alpha},
\end{align*}
where we used Theorem \ref{dualthrm} and the linearity of the scalar product. Further, we apply $\sigma_s(g_s)$:
\begin{align*}
	\left(\sigma_s(g_s)\q\sigma_s(g_s)^*f\right)(x) &= \int_{\set} \kappa_s(x,z, g_s) \left(\q\sigma_s(g_s)^*f\right)(z) dz  \\
	&=  \left\langle \int_{\set} \int_{\set}\kappa_s(x,z,g_s)q(z,y) \kappa_s (\cdot,y, g_s)dy dz , f\right\rangle_{\alpha}= \left\langle\Psi_s(x, \cdot) , f\right\rangle_{\alpha},
\end{align*}
for $\Psi_s(x, \cdot) :=  \int_{\set} \int_{\set}\kappa_s(x,z, g_s)q(z,y) \kappa_s (\cdot,y, g_s)dy dz$. We go now back to the definition of $f$:
\begin{align*}
	&\left(\sigma_s(g_s)\q\sigma_s(g_s)^*\right) \left(\D_{\ell}^{w*}\delta_{T_1-s}^*(1)\right)(x) = \left\langle\Psi_s(x, \cdot) , \D_{\ell}^{w*}\delta_{T_1-s}^*(1)\right\rangle_{\alpha} \\
	&=  \left\langle \D_{\ell}^{w}\Psi_s(x, \cdot) ,\delta_{T_1-s}^*(1)\right\rangle_{\alpha} = \delta_{T_1-s}\left(\D_{\ell}^{w}\Psi_s(x, \cdot)\right) =\left(\D_{\ell}^{w}\Psi_s\right)(x, T_1-s).
\end{align*}
By Lemma \ref{deliveryfunction} we can write that
\begin{align*}
	\left(\D_{\ell}^{w}\Psi_s\right)(x, T_1-s) &= \int_{\R_{+}}d_{\ell}(T_1-s, u)\Psi_s(x,u)du \\&= \int_{\R_{+}}\int_{\set} \int_{\set}d_{\ell}(T_1-s, u)\kappa_s(x,z, g_s)q(z,y) \kappa_s (u,y, g_s)dy dzdu,
\end{align*}
to which, finally, we apply the operator $\delta_{T_1-s}\D_{\ell}^{w}$:
\begin{align*}
	&\delta_{T_1-s}\D_{\ell}^{w}\left(\sigma_s(g_s)\q\sigma_s(g_s)^*\right) \left(\D_{\ell}^{w*}\delta_{T_1-s}^*(1)\right) \\&= \int_{\R_{+}}d_{\ell}(T_1-s, v)\left(\sigma_s(g_s)\q\sigma_s(g_s)^*\right) \left(\D_{\ell}^{w*}\delta_{T_1-s}^*(1)\right)(v)dv\\
	&= \int_{\R_{+}}\int_{\R_{+}}\int_{\set} \int_{\set}d_{\ell}(T_1-s, v)d_{\ell}(T_1-s, u)\kappa_s(v,z, g_s)q(z,y) \kappa_s (u,y, g_s)dy dzdudv,
\end{align*}
finalizing the proof. 

\subsection{Proof of Proposition \ref{prop1} }
\label{appendix:vol}
We consider the representation
\begin{equation}
	\label{sigma:app}
	\Sigma^2_s = a^2\int_{\R_{+}}\int_{\R_{+}}e^{-bu}e^{-bv}d_{\ell}(T_1-s, u)d_{\ell}(T_1-s, v) \mathcal{A}(u,v)du dv,
\end{equation}
where we have introduced
$$\mathcal{A}(u,v) := \int_{\R}\int_{\R}\omega(v-z)q(z,y)\omega(u-y)dy dz, \qquad  u,v\in\R_{+}.$$
By applying (repeatedly) the integration by parts, and since $\omega''$ is null, we obtain
\begin{align}
	\mathcal{A}(u,v) &
	=  \int_{\R}\omega(v-z)\left(\int_{\R}e^{-k|z-y|}\omega(u-y)dy\right) dz \notag\\
	&=\int_{\R}\omega(v-z)\left(\int_{-\infty}^ze^{-k(z-y)}\omega(u-y)dy+\int_{z}^{\infty}e^{-k(y-z)}\omega(u-y)dy\right) dz\notag\\
	& = \label{A}
	\frac{2}{k}\int_{\R}\omega(v-z)\omega(u-z)dz.
\end{align}
By substituting equation \eqref{A} into \eqref{sigma:app}, we get that
\begin{align*}
	\Sigma^2_s &=  \frac{2a^2}{k}\int_{\R}\int_{\R_{+}}\int_{\R_{+}}e^{-bu}e^{-bv}d_{\ell}(T_1-s, u)d_{\ell}(T_1-s, v) \omega(v-z)\omega(u-z)dzdu dv\\
	&=  \frac{2a^2}{k}\int_{\R}\left(\int_{\R_{+}}e^{-bu}d_{\ell}(T_1-s, u) \omega(u-z)du\right)^2 dz\\
	&=  \frac{2a^2}{k\ell^2}\int_{\R}\left(\int_{T_1-s}^{T_1-s+\ell}e^{-bu} \omega(u-z)du\right)^2 dz,
\end{align*}
where we used the definition of $d_{\ell}$ in Lemma \ref{deliveryfunction}. By integration by parts, we get
\begin{align*}
	\Sigma^2_s &=  \frac{2a^2}{k\ell^2b^4}\int_{\R}\left( e^{-b(T_1-s)}\left(b\omega(T_1-s-z)+\omega'(T_1-s-z) \right)+\right.\\&\qquad \qquad \qquad \qquad  \left.-e^{-b(T_1-s+\ell)}\left(b\omega(T_1-s+\ell-z)+\omega'(T_1-s+\ell-z) \right)\right)^2 dz\\
	&= \frac{2a^2}{k\ell^2b^4}\left(e^{-2b(T_1-s)}\mathcal{B}_1(s) -2e^{-2b(T_1-s)}e^{-b(T_1-s+\ell)}\mathcal{B}_2(s)+e^{-2b(T_1-s+\ell)}\mathcal{B}_3(s)\right),
\end{align*}
where we introduced
\begin{align*}
	& \mathcal{B}_1(s) := \int_{\R} \left(b\omega(T_1-s-z)+\omega'(T_1-s-z) \right)^2dz,\\
	& \mathcal{B}_2(s) := \int_{\R}\left(b\omega(T_1-s-z)+\omega'(T_1-s-z) \right)\left(b\omega(T_1-s+\ell-z)+\omega'(T_1-s+\ell-z) \right) dz,\\
	& \mathcal{B}_3(s) :=\int_{\R}\left(b\omega(T_1-s+\ell-z)+\omega'(T_1-s+\ell-z) \right)^2 dz.
\end{align*}
By using the definition of $\omega$ in equation \eqref{omega}, we get that
\begin{align*}
	\mathcal{B}_1(s) &= \int_{T_1-s-1}^{T_1-s+1} \left(b(1-|T_1-s-z|)-\mathrm{sgn}(T_1-s-z) \right)^2dz\\
	&= \int_{T_1-s-1}^{T_1-s} \left(b(1-T_1+s+z)-1 \right)^2dz+ \int_{T_1-s}^{T_1-s+1} \left(b(1+T_1-s-z)+1 \right)^2dz \\
	&= \frac{2}{3}\left(b^2+3\right),
\end{align*}
where $\mathrm{sgn}$ denotes the sign function. Similarly,
\begin{equation*}
	\mathcal{B}_2(s) = \frac{b^2}{6}\left(3(\ell-2)\ell^2+4\right)-3\ell+2, \qquad 
	\mathcal{B}_3(s) = \frac{2}{3}\left(b^2+3\right).
\end{equation*}
By substituting these findings and rearranging the terms, we get that
\begin{equation*}
	\Sigma^2_s = \frac{2a^2}{kb^4\ell^2}\,e^{-2b(T_1-s)} \left\{ \frac{2}{3}\left(b^2+3\right)\left(1+e^{-2b\ell}\right)-2e^{-b\ell}\left(\frac{b^2}{6}\left(3(\ell-2)\ell^2+4 \right)-3\ell+2 \right) \right\},
\end{equation*}
which concludes the proof. 

\section{The non-injectivity issue}
\label{injectivity}
From the numerical experiments, we observe that the accuracy achieved in calibration is not particularly convincing, especially for the parameters regarding the volatility and the covariance operator, $a$, $b$ and $k$. Slightly better results were obtained for the Nelson-Siegel curve parameters, $\alpha_0$, $\alpha_1$ and $\alpha_3$, with the exception of $\alpha_2$ (see Figure \ref{adam_grid_dense}, \ref{adam_pointwise} and \ref{adam_BA}). On the other hand, the relative error for the price approximation after calibration shows high degree of accuracy (Figure \ref{step2_adam_grid_dense}, \ref{step2_adam_pointwise} and \ref{step2_adam_BA}). We may conclude that the original meaning of the model parameters is lost in the approximation step. Indeed, as pointed out in \cite{deep1}, it is somehow to be expected that the neural network is non-injective in the input parameters on large part of the inputs domain. We shall briefly analyse this. 

The pricing formula \eqref{priceVexp}, once fixed the strike $K$ and the time to maturity $\tau$, crucially depends on $\xi$ and $\mu(g_{t})$ as derived, respectively, in Proposition \ref{prop2} and equation \eqref{NSdrift}:
\begin{equation*}
	\Pi(t) = e^{-r(\tau-t)}\left\{\xi \phi\left(\frac{\mu(g_t)-K}{\xi}\right) + \left(\mu(g_t)-K\right)\Phi\left(\frac{\mu(g_t)-K}{\xi}\right)\right\}.
\end{equation*}
However, $\xi$ is only a scale, while $\mu(g_{t})$ is more influential on the final price level since it defines the distance from the strike price $K$. Let us first focus on $\xi$:
\begin{equation*}
	\xi^2 = \frac{a^2}{kb^5\ell^2}\left(e^{-2b(T_1-\tau)}-e^{-2b(T_1-t)}\right)B(b^2, e^{-b\ell}),	
\end{equation*}
where $B(b^2, e^{-b\ell})$ simply indicates a term proportional to $b^2$ and $e^{-b\ell}$. In the front coefficient, a decrease in $a$ might be, for example, compensated by an increase in $b$ or $k$, and vice versa, meaning that several combinations of values for $a$, $b$, and $k$ lead to the same overall $\xi$. Thus we may suspect that it can be hard for the neural network to identify the right vector of parameters despite reaching good level of accuracy for the price. 

In Figure \ref{inj_grid_dense} we report an example of non-injectivity with respect to the parameters $a$, $b$ and $k$ that we have observed in the grid-based learning approach. Here the neural network is not injective when all the parameters, except one, are fixed, and is only little sensitive to the change in the parameters. This also explains the struggle in calibration.

\begin{figure}[!tbp]
	\setlength{\tabcolsep}{5pt}
	\resizebox{1\textwidth}{!}{
		\begin{tabular}{@{}>{\centering}m{0.33\textwidth}>{\centering}m{0.33\textwidth}>{\centering\arraybackslash}m{0.33\textwidth}@{}}
			{\large $\boldsymbol{a}$} & {\large$\boldsymbol{b}$}&{\large$\boldsymbol{k}$}\\
			\includegraphics[width=0.33\textwidth]{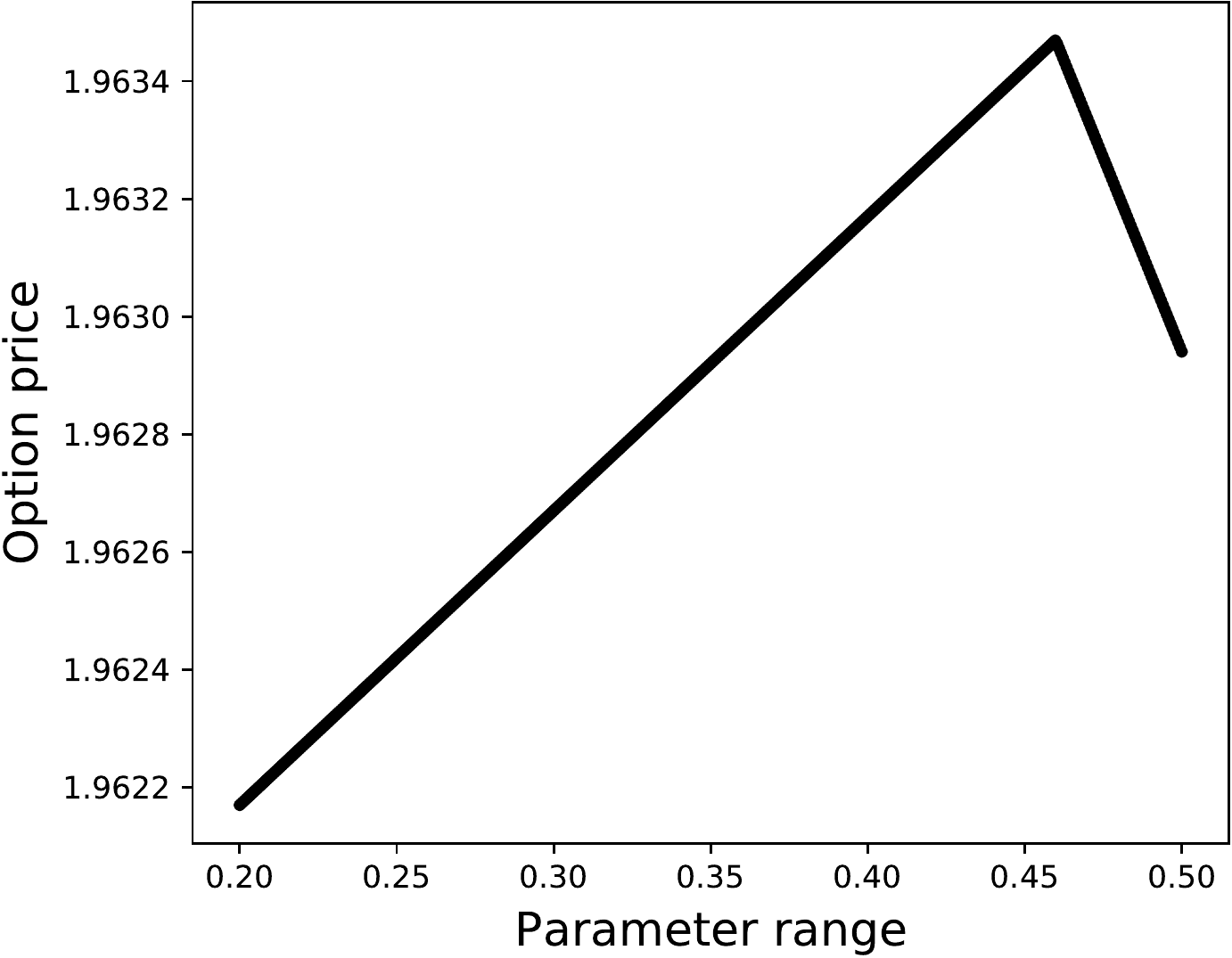} 
			&\includegraphics[width=0.33\textwidth]{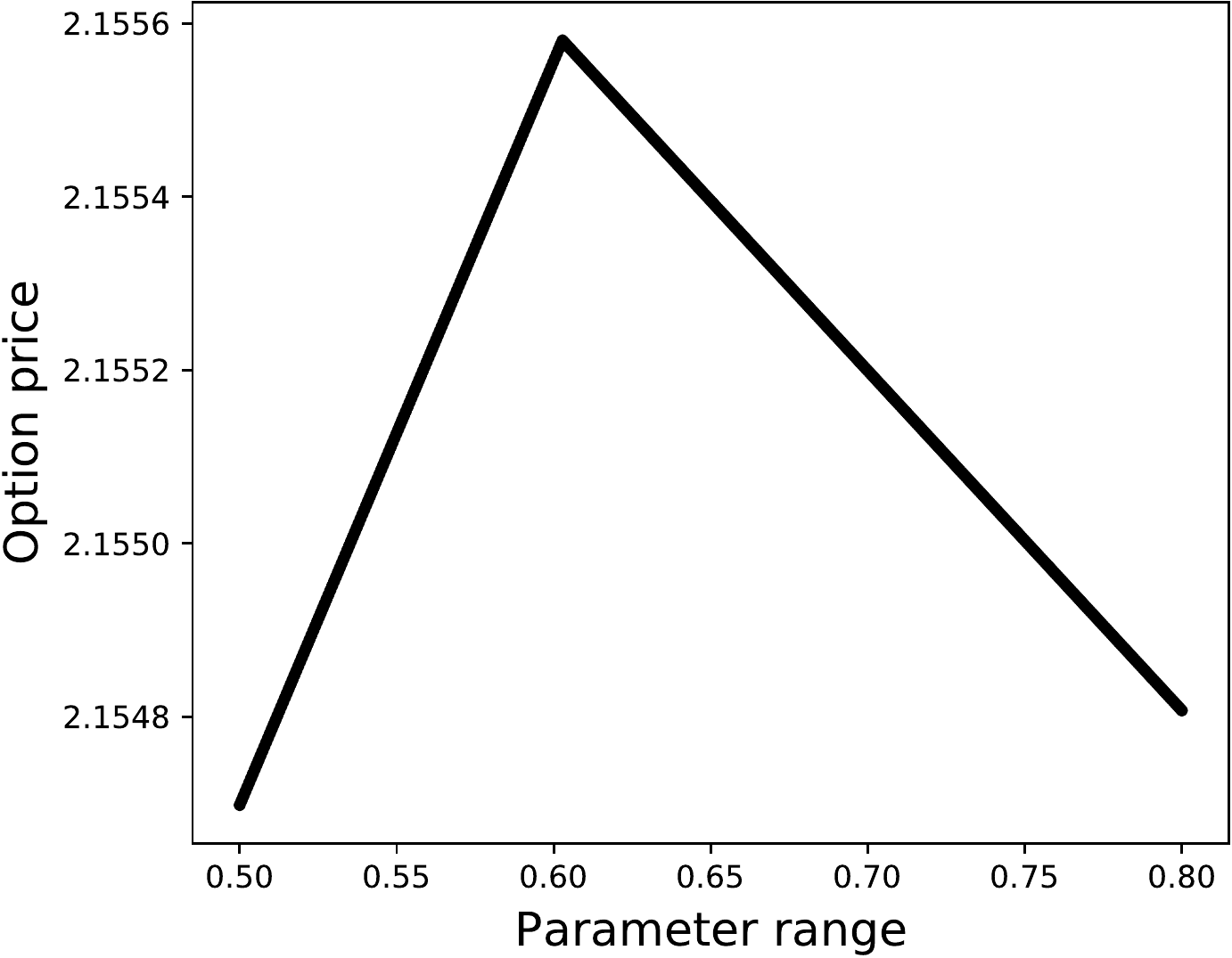}&\includegraphics[width=0.33\textwidth]{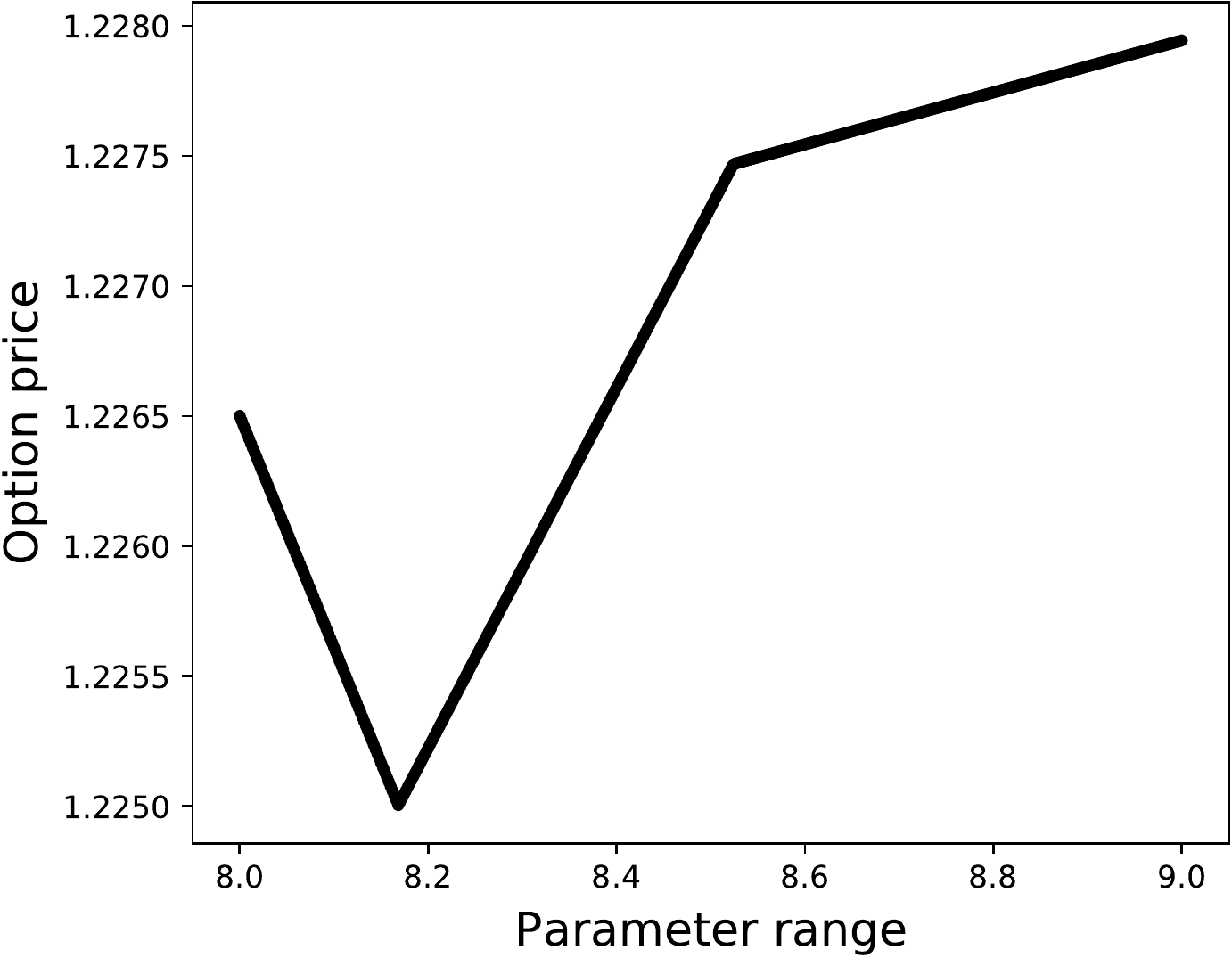}\\
			$K = 32.4 $, $\tau =3/12$& $K = 33.2$, $\tau = 1$ & $K = 31.8 $, $\tau = 1/12$ \\
			{\footnotesize$\theta=[a, 0.7, 8.0, 34.4, -1.0, 0.5, 5.0]$} & {\footnotesize$\theta=[0.45, b, 8.5, 34.3, -1.0, 0.5, 5.0]$} & {\footnotesize$\theta=[0.5, 0.8,  k, 34.4, -1.0, 0.5, 5.0]$}
		\end{tabular}
	}
	\caption{Examples of non injectivity in the grid-based learning approach. In each image, only one of the parameters is varying, while the rest is fixed.\label{inj_grid_dense}}
\end{figure}

Similar observations can be done for the drift:
\begin{equation*}
	\mu(g_t) = \alpha_0+\frac{e^{-\alpha_3(T_1-t)}}{\alpha_3\ell} \left(\alpha_1+\alpha_2+\alpha_2\alpha_3(T_1-t) \right)-\frac{e^{-\alpha_3(T_1+\ell-t)}}{\alpha_3\ell}\left(\alpha_1+\alpha_2+\alpha_2\alpha_3(T_1+\ell-t) \right).
\end{equation*}
Here the role of $\alpha_0$ is specific since it defines the starting level of the curve, and indeed $\alpha_0$ is the parameter that gets the best accuracy in estimation. However, $\alpha_2$ appears first added to $\alpha_1$ and then multiplied by $\alpha_3$, making it hard for the neural network to outline its role. In the Nelson-Siegel curve in equation \eqref{NS}, $\alpha_2$ defines the position of the "bump", but the drift $\mu(g_t)$ is obtained by integrating the curve within the delivery period of the contract. This integration smoothens the curve and makes it hard to locate the "bump".  This might explain why the accuracy in estimating $\alpha_2$ is worse than for the other Nelson-Siegel parameters.

We conclude the article with the following theorem showing that it is possible to construct ReLU neural networks which act as simple linear maps.
\begin{thrm}
	\label{th:inj}
	Let $A \in \mathbb{R}^{p \times d}$. Then for any $L \geq 2$ and any $\mathbf{n}= (d, n_1, \ldots, n_{L-1}, p)$ with $n_i \geq 2d$, $i=1,\ldots, (L-1)$, there exists an $L$-layer  ReLU neural network $\mathcal{N} \colon \mathbb{R}^d \to \mathbb{R}^p$ with dimension $\mathbf{n}$, which satisfies 
	\[  \mathcal{N}(x) = Ax, \quad\quad \text{for all } x \in \mathbb{R}^d.  \]
\end{thrm}
\begin{proof}
	We follow a similar approach to \cite[Section 8.5]{nina}. Let $\nu_i\ge 0$ be such that $n_i = 2d+\nu_i$ for $i=1, \dots,(L-1)$. For $I_d$ the identity matrix of dimension $d$, we define the following weights:
	\begin{align*}
		&V_1 := \begin{bmatrix}
			I_d& -I_d & O_1
		\end{bmatrix}^{\top},\\
		& V_i := \begin{bmatrix}
			I_d& -I_d & O_i \end{bmatrix}
		\begin{bmatrix}
			I_d& -I_d & O_{i-1}
		\end{bmatrix}^{\top}, \quad i=2, \dots, (L-1),\\
		& V_L := A   \begin{bmatrix}
			I_d& -I_d & O_{L-1} \end{bmatrix},
	\end{align*}
	where $\top$ denotes the transpose operator. Here $O_i \in \R^{d\times \nu_i}$ are matrices with all entries equal to $0$ to compensate the matrix dimension in such a way that $V_i \in \R^{n_{i}\times n_{i-1}}$ for $i=1,\dots,(L-1)$. By considering zero-biases vectors $v_i$, the linear maps $H_i$ introduced in the neural network definition in equation \eqref{NNdef} coincide then with the matrices $V_i$.
	
	We observe that for every $x\in \R^d$, the ReLU activation function satisfies
	$$x = \rho(x)-\rho(-x) = \begin{bmatrix}
		I_d & -I_d
	\end{bmatrix}
	\rho\left(
	\begin{bmatrix}
		I_d &-I_d
	\end{bmatrix}^{\top}
	x
	\right),$$
	where the activation function is meant to act component wise.
	By straightforward calculation, one can then see that the neural network defined here satisfies the equality $\mathcal{N}(x) = Ax$ for every $x\in \R^d$, which means that it acts on $x$ as a linear map. 
\end{proof}

Theorem \ref{th:inj} proves that we can construct a ReLU $L$-layer neural network which corresponds to a linear map. As there are infinitely many non-injective linear maps (the zero-map being a trivial example), it is then possible to construct infinitely many non-injective ReLU neural networks. Obviously, this does not show that a non-injective network, such as the one constructed in the proof of Theorem \ref{th:inj}, will also minimize the objective function used for training. It might however give a glimpse to understand that neural networks are not very likely to be injective in their input parameters.

\paragraph*{Acknowledgements}
	The authors would like to thank Vegard Antun for precious coding support and related advice, and Christian Bayer for useful discussions. The authors are also grateful to two anonymous referees for their valuable comments which helped to improve the exposition of the paper with the goal to reach a larger audience.

\section*{Conflict of interest}
The authors declare that they have no conflict of interest.



\end{document}